\documentclass{article}

\PassOptionsToPackage{numbers}{natbib}
\usepackage[preprint]{neurips_2024}

\usepackage{amsmath}
\usepackage{amssymb}
\usepackage{mathtools}
\usepackage{amsthm}
\usepackage{rotating}
\usepackage{multirow}
\DeclareMathOperator*{\argmin}{arg\,min}
\DeclareMathOperator*{\argmax}{arg\,max}

\theoremstyle{plain}
\newtheorem{theorem}{Theorem}[section]
\newtheorem{proposition}[theorem]{Proposition}

\theoremstyle{definition}

\theoremstyle{remark}
\newtheorem{remark}[theorem]{Remark}
\usepackage{algorithm}
\usepackage{algorithmic}
\usepackage{caption}
\usepackage{wrapfig}
\usepackage{subcaption}
\usepackage{graphicx}
\usepackage[export]{adjustbox}
\usepackage{minitoc}

\usepackage[bottom]{footmisc}
\usepackage{tikz}
\usetikzlibrary{bayesnet}
\usepackage{bbm}
\newcommand{\vect}[1]{\boldsymbol{\mathbf{#1}}}

\usepackage[utf8]{inputenc} 
\usepackage[T1]{fontenc}    
\usepackage{hyperref}       
\usepackage{url}            
\usepackage{booktabs}       
\usepackage{amsfonts}       
\usepackage{nicefrac}       
\usepackage{microtype}      
\usepackage{colortbl}
\usepackage{xcolor}         

\title{Model Selection for SLOPE Models: A Bayesian Perspective}

\author{%
  Fabio Feser \\
  Department of Mathematics\\
  Imperial College London\\
  \texttt{ff120@ic.ac.uk} \\
   \And
   Marina Evangelou \\
   Department of Mathematics\\
   Imperial College London\\
}

\begin{document}
\maketitle

\begin{abstract}
Sorted $\ell_1$ Penalized Estimation (SLOPE) models, that perform either variable or group selection, control the false discovery rate (FDR) under orthogonal settings with known noise, but such settings are rare in practice. Under general conditions, cross-validation is the default model selection approach for SLOPE, yet it targets predictive performance rather than FDR control. We address this gap for the SLOPE family of models by proposing new Bayesian approaches, Bayesian Group SLOPE (BGSLOPE) and Bayesian Sparse-group SLOPE (BSGS). BGSLOPE and BSGS embed group-based SLOPE models into a spike-and-slab framework, with BSGS providing a continuous spike-and-slab framework for sparse-group models. We further introduce Two-step Orthogonal (TSO), which transforms a general setting into an orthogonal one to recover SLOPE's FDR control properties. Through extensive synthetic and real data studies comparing all major model selection strategies for SLOPE models, the proposed Bayesian models consistently control FDR, achieve higher power, and outperform competing methods in prediction. 
\end{abstract}
\section{Introduction}
Suppose we observe a design matrix $\mathbf{X} \in \mathbb{R}^{n\times p}$ and a response vector $\mathbf{y}\in\mathbb{R}^n$, where $p$ denotes the number of features and $n$ the number of observations. Consider the linear regression setup $\mathbf{y} \mid \boldsymbol\beta, \sigma^2 \sim \vect{\mathcal{N}}_n(\mathbf{X}\boldsymbol\beta, \sigma^2 \mathbf{I}_n)$, where $\boldsymbol{\beta}\in\mathbb{R}^p$ and $\boldsymbol\epsilon \sim \vect{\mathcal{N}}_n(\mathbf{0},\mathbf{I}_n \sigma^2)$ such that $\sigma^2 >0$ denotes the \textit{noise}. Then, a general penalized regression model is given by
\begin{equation}\label{eqn:penalised_problem}
    \hat{\boldsymbol\beta}(\lambda) \in \argmin_{\boldsymbol{\beta}\in \mathbb{R}^p} \left\{ \frac{1}{2n}\|\mathbf{y} - \mathbf{X}\boldsymbol{\beta}\|_2^2+\lambda J(\boldsymbol{\beta};\mathbf{v})\right\},
\end{equation}
where $J(\cdot)$ is a convex penalty norm, $\mathbf{v} \succeq 0$ are weights, and $\lambda>0$ is the tuning parameter. 

For any penalized regression model, a key problem is model selection via choosing the \textit{tuning parameter} $\lambda$. It controls the sparsity of the fitted model and is proportional to the noise $\sigma^2$. Most proposed schemes for $\lambda$ are explicit functions of the noise \citep{Sun2012ScaledRegression}. As the true noise level is rarely known in practice, estimation or tuning of it is required.

In this manuscript, we consider the problem of model selection for \textit{Sorted $\ell_1$ Penalized Estimation} (SLOPE) models (Section \ref{section:slope-background}). SLOPE models use sorted norms to achieve false discovery rate (FDR) control under orthogonal designs, making them widely used in genetics and machine learning \citep{10.1371/journal.pone.0269369,Gossmann2015,KREMER2020105687,riccobello2023sparsegraphicalmodellingsorted,virouleau2017highdimensionalrobustregressionoutliers}. Grouping information, often found in genetics through biological pathways, can help inform inference and lead to better prediction. In this vein, SLOPE was extended to group regression in Group SLOPE (gSLOPE) \citep{Brzyski2019GroupPredictors} and Sparse-group SLOPE (SGS) \citep{Feser2023Sparse-groupFDR-control}.

Model selection is particularly challenging for SLOPE since the sparsity of the fitted model determines the FDR. An incorrect $\lambda$, therefore, directly compromises FDR control. Theory guarantees FDR control only under orthogonal $\mathbf{X}$ with known noise ($\lambda = 1$), leaving general settings without guidance.

\subsection{Model selection approaches}\label{section:model_selection_approaches}
There are many approaches for model selection for penalized models; here, we broadly group them into \textit{tuning} and \textit{estimation} regimes. To tune $\lambda$, we typically fit a path of $l$ models $\mathcal{M} = \{\hat{\boldsymbol{\beta}}(\lambda_1),\ldots,\hat{\boldsymbol{\beta}}(\lambda_l)\}$, where $\lambda_1\geq \ldots \geq \lambda_l > 0$, and choose the optimal model, $\hat{\mathcal{M}}$, from the path according to some metric. On the other hand, estimation of $\lambda$ provides a model $\hat{\mathcal{M}} = \hat{\boldsymbol\beta}(\hat{\lambda})$. Both approaches aim to obtain a final optimal solution $\hat{\boldsymbol{\beta}}$ and active set $\hat{S}_v = \mathrm{supp}({\hat{\boldsymbol\beta}}) := \{i \in [p] : \hat\beta_i \neq 0\}$, where $[p] :=  \{1,\ldots,p\}$, which, importantly for SLOPE models, defines the FDR. 

We briefly discuss the relevant approaches here, while the methods used in the simulation studies of Section \ref{section:simulation_study} are described in detail in Appendix \ref{section:other_approaches}.

\paragraph{Tuning along a path.} The most common tuning approach is cross-validation (CV), using predictive performance as the selection metric. In genetics, however, the primary focus is model recovery for FDR control; an objective that does not necessarily align with prediction, particularly since CV tends to select overfitting models \citep{Leng2006,Yang2005}. When FDR control is the goal, no CV criterion provides finite-sample guarantees \citep{Bogdan2015SLOPEAdaptiveOptimization}, and \citet{Feser2023Sparse-groupFDR-control} show that CV tends to select the most saturated model on the SLOPE path, leading to high FDR levels. Many CV variants exist \citep{Bates02042024,Golub01051979,ROBERTS2014198,Yu02102014} (see \citet{10.1214/09-SS054} for a summary).

The Knockoffs procedure \citep{Barber2015ControllingKnockoffs} builds on CV by generating synthetic negative-control copies of the input data, independent of the response but matching the input structure. CV is applied to the augmented design, and the Knockoffs filter is used to remove noise variables from $\hat{\mathcal{M}}$, providing FDR control (Appendix \ref{section:knockoffs}). 

A limitation of path-based tuning is that only one model is selected, discarding information from the rest of the path. Stability selection \citep{Meinshausen2010StabilitySelection} addresses this by aggregating results across the full path, providing family-wise error rate guarantees. Stability selection extensions for FDR control exist \citep{Ahmed2011FalseStudies}, although our experiments show unsatisfactory performance and high computational cost compared to competing methods (Appendix \ref{appendix:fdr_sim_study}).

\paragraph{Noise estimation.} Estimation approaches directly estimate $\sigma^2$ and use it to fit the final model. Popular procedures include \citet{Fan2012VarianceRegression}, which uses a two-stage data splitting technique, and \citet{Dicker2014VarianceModels}, which derives consistent method-of-moments estimators without sparsity assumptions on $\boldsymbol{\beta}$ (see \citet{tibnoise} for a comprehensive study).

Scaled regression (Section \ref{section:scaled_regression}) jointly estimates the coefficients and noise by iteratively scaling the tuning parameter by the mean squared residual. It has been adapted for SGS as AS-SGS (Appendix \ref{section:as_sgs}) \citep{Feser2023Sparse-groupFDR-control}.   
\paragraph{Bayesian approaches.}
An alternative estimation approach comes from learning the noise alongside the other model parameters in a Bayesian penalized regression model (Section \ref{section:bayes_penalized_reg}). A common approach to this is the spike-and-slab framework (Section \ref{section:spike-and-slab}), which models the noise and signal using a mixture prior. There are both continuous and point-mass implementations of spike-and-slab priors.

\citet{Jiang2022AdaptiveData} extends the continuous spike-and-slab framework to SLOPE, forming the Adaptive Bayesian SLOPE (ABSLOPE) (Section \ref{section:abslope}). Bayesian group-based models have been developed for the lasso using point-mass priors \citep{Xu2015BayesianEstimation} and continuous priors \citep{Bai2022Spike-and-SlabModels}. The models proposed in this manuscript extend the continuous spike-and-slab framework to group-based SLOPE models. Section \ref{section:bayesian_intro} outlines the relevant Bayesian background, and further Bayesian model selection techniques can be found in \citep{10.1214/10-BA607,Johnson01062012,10.5555/3722577.3722762,WAGNER20121256,WASSERMAN200092}.
  


\paragraph{Additional approaches.} Other model selection approaches that do not necessarily fall into the above categories include post-selection inference \citep{10.1214/15-AOS1371,10.1093/biomet/asac070,10.1214/14-AOS1221}, information criteria \citep{1100705,10.1214/aos/1176325766,10.1093/biomet/87.4.731,10.1214/aos/1176344136,https://doi.org/10.1111/1467-9868.00191,WANG20111141} (a review of these is found in \citet{1311138}), square-root lasso \citep{10.1093/biomet/asr043}, rank lasso \citep{Wang01102020}, ET-LASSO \citep{Yang2019ET-Lasso:Data}, TREX  \citep{Lederer_Müller_2015}, and estimating the FDR \citep{luo2024estimatingfdrvariableselection,9914611}. For additional approaches to model selection and parameter tuning in high-dimensional regression, see relevant summaries \citep{Homrighausen,LEE20191,annurev:/content/journals/10.1146/annurev-statistics-030718-105038}.

\paragraph{A motivating simulation study.} To assess whether existing approaches reliably achieve FDR control for SLOPE models, we conducted a short simulation study comparing methods for estimating $\lambda$ or directly controlling the FDR, including CV, Knockoffs, stability selection variants, ET-LASSO, and noise estimation procedures (Appendix \ref{appendix:fdr_sim_study}). None of the approaches produced FDR control across all settings: even the best-performing method, Knockoffs, substantially exceeded the target under moderate and strong correlation. Based on this study, we selected Knockoffs for its comparatively strong performance, and CV for its widespread use, for further investigation in the main simulation study (Section \ref{section:simulation_study}).

These results highlight a fundamental limitation of frequentist model selection for SLOPE: theoretical FDR guarantees hold only under orthogonal designs with known noise. Existing approaches cannot reliably recover this property in general settings. A Bayesian formulation offers a principled alternative, estimating $\sigma^2$ jointly with $\boldsymbol{\beta}$ within the same model structure. This avoids the misalignment of CV, which tunes $\lambda$ for predictive performance rather than FDR control. \citet{Jiang2022AdaptiveData} demonstrates this advantage empirically for SLOPE via ABSLOPE (Section \ref{section:abslope}), and the spike-and-slab LASSO has more broadly been shown to outperform frequentist penalized regression in variable selection \citep{Bai2021Spike-and-slabLASSO,Rockova2018TheLASSO}. The Bayesian formulation further provides uncertainty quantification through inclusion probabilities and extends naturally to hierarchical models that facilitate information sharing across groups \citep{Bernardo1994,Gelman01082006}.

\subsection{Contributions}
In this manuscript, we develop Bayesian models for group-based SLOPE models for the purpose of model selection under general settings. The models learn the noise along with the other model parameters. To do this, we incorporate the gSLOPE and SGS norms into a spike-and-slab framework to form the Bayesian gSLOPE (BGSLOPE) (Section \ref{section:bgslope}) and the Bayesian SGS (BSGS) (Section \ref{section:bsgs}). BGSLOPE and BSGS introduce novel group-based spike-and-slab frameworks applicable to the wider class of Ordered Weighted $\ell_1$ (OWL) models (Section \ref{section:extentions}), with BSGS further extensible to other sparse-group formulations such as the sparse-group lasso \citep{Simon2013}. Our models inherit many of the advantages of Bayesian formulations, including uncertainty quantification and adaptivity to the underlying sparsity structure.

We compare our Bayesian methodology and ABSLOPE (Section \ref{section:abslope}) with existing model selection approaches across all SLOPE variants, providing practitioners with practical insights for choosing an appropriate method. The methods are evaluated on extensive synthetic (Section \ref{section:simulation_study}) and real-data (Section \ref{section:real_data}) studies, with primary emphasis on FDR control and secondary emphasis on predictive performance. Beyond the Bayesian methodology, we extend several existing approaches to group-based SLOPE models, including scaled regression (Appendix \ref{section:scaled_regression}) and, to our knowledge, the first sparse-group Knockoffs implementation (Appendix \ref{section:sparse_group_knockoffs}). We also consider two-step approaches designed to combine the advantages of multiple methods (Appendix \ref{section:two-step}) and propose a new variant, Two-step Orthogonal (TSO) (Section \ref{section:tso}).

\section{Background}\label{section:bayesian_intro}
\subsection{Bayesian penalized regression}\label{section:bayes_penalized_reg}


The connection between penalized regression and Bayesian inference was first established by observing that the lasso estimator corresponds to the posterior mode when each regression coefficient is assigned an independent double exponential (Laplace) prior \citep{Tibshirani1996RegressionLasso}. Formally, any penalized estimator admits a Bayesian interpretation given by
\begin{align}
&\pi(\boldsymbol{\beta} \mid \lambda) = \prod_{j=1}^{p} \pi(\beta_j \mid \lambda), \;\; \text{where the penalty function satisfies}\\
&J(\boldsymbol{\beta}; \lambda) = -\log \pi(\boldsymbol{\beta} \mid \lambda) = -\sum_{j=1}^{p} \log \pi(\beta_j \mid \lambda).  
\end{align}
Applied to the lasso, we have $\pi(\boldsymbol{\beta} \mid \lambda) = \prod_{j=1}^{p} \frac{\lambda}{2} e^{-\lambda |\beta_j|}$ \citep{Rockova2018TheLASSO}. The Bayesian lasso shares the non-separability property of frequentist SLOPE, allowing variables to share information, but is known to produce biased estimates with suboptimal coefficient shrinkage \citep{Ghosh2016AsymptoticSparsity, Rockova2016}.
\subsection{Spike-and-slab framework}\label{section:spike-and-slab}
The spike-and-slab framework applied to the lasso overcomes issues with bias and suboptimal shrinkage \citep{Bai2021Spike-and-slabLASSO}. 
The \textit{continuous} spike-and-slab prior is of the form \citep{George1993VariableSampling}
\begin{align}
\pi\left(\boldsymbol{\beta} \mid \boldsymbol{\gamma}, \sigma^2\right) &=\prod_{i=1}^p\left[\left(1-\gamma_i\right) \mathcal{N}\left(0, \sigma^2 \tau_0^2\right)+\gamma_i \mathcal{N}\left(0, \sigma^2 \tau_1^2\right)\right]\\
\pi(\boldsymbol\gamma \mid \theta) &=\prod_{i=1}^p \theta^{\gamma_i}(1-\theta)^{1-\gamma_i}, \quad \theta \sim \pi(\theta), \quad \sigma^2 \sim \pi\left(\sigma^2\right),
\end{align}
where $\tau_0^2$ models the noise (the \textit{spike}) and $\tau_1^2$ captures the signal (the \textit{slab}), such that $0 < \tau_0^2 \ll \tau_1^2$, $\boldsymbol{\gamma} \in \{0,1\}^p$ is a binary inclusion vector over $2^p$ models, and $\theta \in (0,1)$ is a mixture proportion. A key advantage of the spike-and-slab model is its ability to automatically adapt to the underlying sparsity structure of the data through the mixing parameter $\theta$ \citep{Bai2021Spike-and-slabLASSO}.

A non-continuous alternative uses a point-mass spike and heavy-tailed slab \citep{Mitchell1988BayesianRegression}. Both formulations enable simultaneous variable selection and parameter estimation, mitigating over-shrinkage by placing sufficient prior mass on large effect sizes \citep{Jiang2022AdaptiveData}. The spike-and-slab lasso (SSL) \citep{Rockova2018TheLASSO} replaces the Gaussian densities with Laplace densities $\psi(\boldsymbol{\beta} \mid \lambda) = (\lambda/2)e^{-\lambda |\boldsymbol{\beta}|}$, with $\lambda$ chosen separately for the spike and slab components.       

\subsection{SLOPE models}\label{section:slope-background}
SLOPE is given by the sorted norm $J_\text{slope}(\boldsymbol\beta;\mathbf{v}) = \sum_{i=1}^p v_i|\beta|_{(i)}$, where $v_1 \geq \ldots \geq v_p \geq 0$ and $|\beta|_{(1)}\geq \ldots \geq |\beta|_{(p)}$ \citep{Bogdan2015SLOPEAdaptiveOptimization}. The weights are designed for variable FDR control and are calculated as the quantiles of a standard Gaussian $v_i^\text{BH} = \Phi^{-1}(1-\frac{q_v i}{2p})$ (termed the \textit{Benjamini-Hochberg} (BH) sequence), where $q_v \in (0,1)$ is the desired variable FDR level. Aside from the FDR properties, SLOPE also clusters strongly correlated features, is asymptotically minimax, and is adaptive to unknown sparsity proportions (which complements a similar property spike-and-slab models hold) \citep{figueiredo2014sparseestimationstronglycorrelated,Su2016SLOPEMinimax}. SLOPE models have non-separable norms, so are fitted using proximal algorithms.

This manuscript focuses on group-based SLOPE models: assume the variables sit within $m$ non-overlapping groups $G_1,\dots,G_m$ of sizes $p_1,\ldots, p_m$, then SLOPE was extended to group regression by \textit{Group SLOPE} (gSLOPE) \citep{Brzyski2019GroupPredictors}, with the norm
\begin{equation}\label{eqn:gslope_pen}
    J_\text{gslope}(\boldsymbol\beta;\mathbf{w}) = \sum_{j=1}^m \sqrt{p_j} w_j \|\boldsymbol{\beta}^{(j)}\|_2,
\end{equation}
where $\boldsymbol\beta^{(j)}\in\mathbb{R}^{p_j}$ is a vector of the variable coefficients in a group $j \in [m]$. As with SLOPE, the coefficients and weights are ordered: $\sqrt{p_1}\|\boldsymbol{\beta}^{(1)}\|_2 \geq \ldots\geq \sqrt{p_m}\|\boldsymbol{\beta}^{(m)}\|_2$ and $w_1\geq \ldots \geq w_m \geq 0$. The ordered weights are calculated as
\begin{equation}\label{eqn:gslope_pen_max}
    w_j= \max_{k = 1, \ldots, m} \left\{ \frac{1}{\sqrt{p_k}} F^{-1}_{\chi_{p_k}} \left(1 - \frac{q_g j}{m} \right) \right\},
\end{equation}
where $F_{\chi_{p_k}}$ is the CDF of a $\chi$ distribution with $p_k$ degrees of freedom and $q_g \in (0,1)$ is the desired group FDR level. The maximum criterion in the weights leads to conservative group FDR control and can be relaxed in practice (leading to the mean sequence in Equation \ref{eqn:gslope_pen_mean}).

\textit{Sparse-group SLOPE} (SGS) \citep{Feser2023Sparse-groupFDR-control} combines the strengths of SLOPE and gSLOPE via a convex combination of $\alpha \in [0,1]$:
\begin{equation}\label{eqn:sgs_pen}
    	J_\text{sgs}(\boldsymbol{\beta}; \alpha, \mathbf{v},\mathbf{w}) =  \alpha \sum_{i=1}^{p}v_i |\beta|_{(i)} +  (1-\alpha)\sum_{j=1}^{m}w_j \sqrt{p_j} \|\boldsymbol{\beta}^{(j)}\|_2,
\end{equation}
where the sorting procedures of both SLOPE and gSLOPE apply onto the corresponding penalties. Unlike gSLOPE, SGS does not require all variables within an active group to be active, allowing noise variables to be shrunk to zero. \citet{Feser2023Sparse-groupFDR-control} derives penalty sequences for bi-level FDR control (Equations \ref{eqn:sgs_var_pen_max_2} and \ref{eqn:sgs_grp_pen_max_2} in Appendix \ref{appendix:slope_weights}), finding that the best performance uses a combination of the SGS variable mean (Equation \ref{eqn:sgs_var_pen_mean}) and gSLOPE group mean (Equation \ref{eqn:gslope_pen_mean}) sequences. These weights are used throughout this manuscript, except for AS-SGS (Appendix \ref{section:scaled_regression}).   
\subsection{SLOPE in a Bayesian context}\label{section:abslope}
\citet{Sepehri2016TheSLOPE} develops a Bayesian formulation of SLOPE, yielding \textit{Bayesian SLOPE}, given by 
\begin{equation}\label{eqn:bayesian_slope}
    \pi (\boldsymbol\beta \mid \sigma^2;\mathbf{v}) \propto e^{\frac{-1}{\sigma}\sum_{i=1}^p v_i |\beta|_{(i)}} \propto \prod_{i=1}^p \exp\left(-\frac{1}{\sigma}v_{r_v(\boldsymbol\beta,i)}|\beta_i| \right).
\end{equation}
The penalty has been reformulated using the rank function $r_v(\boldsymbol\beta, i) \in [p]$, which gives the rank of $\beta_i$ in decreasing order. This prior was incorporated into the spike-and-slab framework in the \textit{Adaptive Bayesian SLOPE} (ABSLOPE) \citep{Jiang2022AdaptiveData}. SLOPE's non-separability and adaptivity to unknown sparsity make it a natural fit for the spike-and-slab framework.      

The ABSLOPE prior is given by
\begin{equation}\label{eqn:abslope_prior}
    \pi\left(\boldsymbol\beta \mid \boldsymbol\gamma, c, \sigma^2 ; \mathbf{v}\right) \propto c^{\sum_{i=1}^p \mathbbm{1}\left(\gamma_i=1\right)} \prod_{i=1}^p \exp \left\{-a_i\left|\beta_i\right| \frac{1}{\sigma} v_{r_v(\mathbf{A} \boldsymbol\beta, i)}\right\},
\end{equation}
where $\boldsymbol\gamma \in \{0,1\}^p$ is the variable inclusion vector and $\mathbf{A} = \text{diag}(a_1,\ldots,a_p)$ is a diagonal matrix with elements $a_i = c \gamma_i + (1-\gamma_i)$, where $c \in (0, 1)$ represents the ratio of the average signal strengths between the active and non-active variables \citep{Jiang2022AdaptiveData}. The spike-and-slab influence comes via $c$ and can be seen more clearly by writing the ABSLOPE prior as
\begin{equation*}
     \pi\left(\boldsymbol\beta \mid \boldsymbol\gamma, c, \sigma^2 ; \mathbf{v}\right) \propto c^{\sum_{i=1}^p \mathbbm{1}\left(\gamma_i=1\right)} \prod_{i=1}^p \exp \left\{-\frac{1}{\sigma}\left[\underbrace{c\gamma_i\left|\beta_i\right|v_{r_v(\mathbf{A} \boldsymbol\beta, i)}}_{\text{slab}} + \underbrace{(1-\gamma_i)\left|\beta_i\right| v_{r_v(\mathbf{A} \boldsymbol\beta, i)}}_{\text{spike}}\right] \right\},
\end{equation*}
so that since $c < 1$, the signal (slab) variables are penalized less; this is the purpose of the matrix $\mathbf{A}$.

The ABSLOPE prior (Equation \ref{eqn:abslope_prior}) is equivalent to the Bayesian SLOPE prior (Equation \ref{eqn:bayesian_slope}) under the weighted design matrix $\mathbf{X}\mathbf{A}^{-1}$, so the \textit{maximum a posteriori} (MAP) under ABSLOPE reduces to solving the frequentist SLOPE with this weighted design \citep{Jiang2022AdaptiveData}. ABSLOPE was originally developed for missing data settings by jointly modeling $\mathbf{X}$. While this is not a problem of interest in this manuscript, we note that our methodology readily extends to cover this case.           

ABSLOPE is a hierarchical Bayesian model with priors
\begin{align}
   \pi(\boldsymbol\gamma \mid \theta) = \prod_{i=1}^p \theta^{\gamma_i} (1-\theta)^{1-\gamma_i}, \quad \theta\sim\text{Beta}(d_1, d_2), \quad c \sim \mathcal{U}[0,1], \quad \pi(\sigma^2) \propto \sigma^{-2}, 
\end{align}
where $\theta = \mathbb{P}(\gamma_j = 1; \theta)$ represents the sparsity level of the coefficients (with $d_1,d_2>0$ defining the Beta prior distribution) and an improper uninformative prior is placed on $\sigma^2$. Values of $\theta$ close to zero generate sparse models, which can be seen by $\mathbb{E}[\theta] = d_1/(d_1+d_2)$. ABSLOPE is optimized using the Stochastic Approximation Expectation–Maximization (SAEM) algorithm  (Section \ref{section:saem}) \citep{Lavielle2014MixedApproach}.

\section{Group-based Bayesian SLOPE models}\label{section:bayesian}  
We extend the spike-and-slab framework (Section \ref{section:spike-and-slab}) to group-based SLOPE models to give Bayesian Group SLOPE (BGSLOPE) (Section \ref{section:bgslope}) and Bayesian Sparse-group SLOPE (BSGS) (Section \ref{section:bsgs}). Both BGSLOPE and BSGS are fitted using the SAEM algorithm (Section \ref{section:saem}). Both models are developed within the linear regression framework, as SLOPE's FDR properties are derived under this setting. The methodology can be extended to broader settings by modifying the likelihood term, though this is beyond the scope of this manuscript. 

The closest existing work is \citet{Xu2015BayesianEstimation}, which develops Bayesian group and sparse-group lasso models using point-mass spike-and-slab priors and Gibbs sampling. Our models differ in three key ways. First, we use SLOPE norms specifically designed for FDR control, rather than lasso-type ones. This allows our models to inherit the useful properties of SLOPE: non-separability, adaptivity to unknown sparsity, and the clustering property. Second, we adopt continuous priors, which are less computationally prohibitive in high dimensions \citep{Bai2021Spike-and-slabLASSO}. Third, we fit via an EM algorithm rather than Gibbs sampling. 

The non-separability of SLOPE further distinguishes our approach: coefficients must be updated jointly, enabling information sharing across variables \citep{Rockova2016}, which is incompatible with the independent priors of \citet{Xu2015BayesianEstimation}. In the sparse-group setting, both their Bayesian SGL and our BSGS use two sets of binary indicators for bi-level selection, but BSGS represents the first continuous spike-and-slab framework for sparse-group models. Other Bayesian sparse-group models include \citep{Cai2020BIVAS:Applications,Chen2016BayesianSelection,Steiger2018Sparse-GroupReconstruction}.

\subsection{Bayesian gSLOPE (BGSLOPE)}\label{section:bgslope}
This manuscript presents the first Bayesian implementation for gSLOPE. First, the gSLOPE penalty (Equation \ref{eqn:gslope_pen}) is rewritten as $\sum_{j=1}^m \sqrt{p_j}w_{r_g(\boldsymbol\beta,j)}\|\boldsymbol\beta^{(j)}\|_2$, where $r_g(\boldsymbol\beta, j) \in  [m]$ is the rank of $\boldsymbol\beta^{(j)}$ among the group elements (using the $\sqrt{p_j}\|\boldsymbol\beta^{(j)}\|_2$ ordering). gSLOPE can be seen as the posterior mode under the following Laplace prior
\begin{equation}\label{eqn:bayesian_gslope_prior}
\pi(\boldsymbol\beta \mid \sigma^2; \mathbf{w}) = C \prod_{j=1}^m \exp \left(-\frac{1}{\sigma} w_{r_g(\boldsymbol\beta,j)}\sqrt{p_j}\|\boldsymbol\beta^{(j)}\|_2 \right),
\end{equation}
with the normalizing constant $C$ given in Theorem \ref{thm:gslope_generalising_constant}. 

Extending the ABSLOPE framework (Section \ref{section:abslope}) to gSLOPE, the \textit{Bayesian Group SLOPE} (BGSLOPE) is defined as the hierarchical model
\begin{align}
    \pi\left(\boldsymbol\beta \mid \boldsymbol\gamma, c, \sigma^2 ;  \mathbf{w} \right) &\propto c^{\sum_{j=1}^m p_j \mathbbm{1}\left(\gamma_j=1\right) } \prod_{j=1}^m \exp \left\{-\frac{1}{\sigma}\tilde{a}_{j}\|\boldsymbol\beta^{(j)}\|_2\sqrt{p_j}  w_{r_g(\tilde{\mathbf{A}} \boldsymbol\beta, j)}\right\}, \label{eqn:bgslope_prior}\\
    \pi(\boldsymbol\gamma \mid \theta) &= \prod_{j=1}^m \theta^{\gamma_j}(1-\theta)^{1-\gamma_j}, \quad\theta \sim \text{Beta}(d_1,d_2),\quad \pi(\sigma^2) \propto \sigma^{-2},\quad c\sim\mathcal{U}[0,1],\nonumber
\end{align}
where $\boldsymbol\gamma\in\{0,1\}^m$ is the group inclusion vector under an independent Bernoulli product prior, $\theta$ is a mixture parameter representing the sparsity of the fitted groups (with $d_1, d_2>0$ prior hyperparameters) and $c$ ensures the signal groups receive less penalization (also seen as the signal strength ratio between active and non-active groups). As smaller values of $\theta$ correspond to sparser underlying models, setting $d_1 \ll d_2$ encourages sparsity in the model. The weight matrix is given by $\tilde{\mathbf{A}} =\text{diag}(\tilde{a}_{1},\ldots,\tilde{a}_{m})\in\mathbb{R}^{p\times p}$ with entries $\tilde{A}_{ii} = \tilde{a}_{j}, \forall i \in G_j$ (each variable in a group has the same entry), where $\tilde{a}_j = c\gamma_j + (1-\gamma_j)$. The BGSLOPE dependency graph is shown in Figure \ref{fig:bgslope_graph}.

BGSLOPE, through the inclusion parameter $\boldsymbol\gamma$, provides measures of both uncertainty quantification and group-level importance; aspects that are not available in the frequentist gSLOPE. Additionally, as the parameter $\theta$ governs the degree of shrinkage applied to each group, it allows BGSLOPE to adapt to varying levels of group sparsity in the data.

Proposition \ref{propn:gslope_equiv} demonstrates that the MAP estimate under the BGSLOPE prior is equivalent to the solution under the gSLOPE Laplace prior, when using the matrix $\vect{X}\tilde{\vect{A}}^{-1}$ as input. By extension, this also makes it equivalent to the frequentist gSLOPE formulation, after reparameterization. The proof of this result is provided in Appendix \ref{appendix:bgslope_theory}.

\begin{proposition}\label{propn:gslope_equiv}
If the coefficients $\mathbf{z} = (z_1, z_2, \ldots, z_p)$ follow the gSLOPE Laplace prior (Equation \ref{eqn:bayesian_gslope_prior}), then $\boldsymbol{\beta} =\tilde{\mathbf{A}}^{-1}\mathbf{z}$ follows the BGSLOPE prior (Equation \ref{eqn:bgslope_prior}).
\end{proposition}

\begin{figure}[H]
    \begin{minipage}[b]{0.44\textwidth}
        \centering
        \begin{tikzpicture}[scale=0.8, node distance=2cm]
            \node[latent, rectangle, draw] (beta) {$\boldsymbol{\beta}$};
            \node[latent, rectangle, draw, above=of beta] (sigma) {$\sigma^2$};
            \node[latent, draw, right=of beta, fill=gray!30] (y) {$\mathbf{y}$};
            \node[latent, draw, left=of beta] (gamma) {$\boldsymbol\gamma$};
            \node[latent, draw, below=of beta] (c) {$c$};
            \node[latent, draw, left=of gamma] (theta) {$\theta$};

            \edge {sigma} {beta};
            \edge {sigma} {y};
            \edge {gamma} {beta};
            \edge {c} {beta};
            \edge {beta} {y};
            \edge {theta} {gamma};
        \end{tikzpicture}
\captionsetup{justification=centering}
        \caption{BGSLOPE dependency graph.}
        \label{fig:bgslope_graph}
    \end{minipage}%
    \begin{minipage}[b]{0.45\textwidth}
        \centering
        \begin{tikzpicture}[
            node distance=1.5cm,
            circ/.style={circle, draw, minimum size=0.8cm},
            rect/.style={rectangle, draw, minimum size=0.8cm},
            arrow/.style={->, >=Stealth}
        ]

        \node[circ] (thetag) at (-3,3) {$\theta_g$};
        \node[circ] (thetav) at (-3,1) {$\theta_v$};

        \node[circ] (gamma) at (-1,3) {$\boldsymbol\gamma$};
        \node[circ] (delta) at (-1,1) {$\boldsymbol\delta$};

        \node[rect] (beta) at (1,2) {$\boldsymbol{\beta}$};

        \node[circ] (c_g) at (0,0) {$c_g$};
        \node[circ] (c_v) at (2,0) {$c_v$};

        \node[rect] (sigma) at (2.5,3.5) {$\sigma^2$};
        \node[circ, fill=gray!30] (y) at (4,2) {$\mathbf{y}$};

        \edge {thetag} {gamma};
        \edge {thetav} {delta};
        \edge {gamma} {beta};
        \edge {delta} {beta};
        \edge {c_g} {beta};
        \edge {c_v} {beta};
        \edge {sigma} {beta};
        \edge {sigma} {y};
        \edge {beta} {y};
        \edge {gamma} {delta};
        \end{tikzpicture}
\captionsetup{justification=centering}
        \caption{BSGS dependency graph.}
        \label{fig:bsgs_graph}
    \end{minipage}
\end{figure}

\subsection{Bayesian SGS (BSGS)}\label{section:bsgs}
SGS (Equation \ref{eqn:sgs_pen}) applies two levels of penalization and sorting to produce bi-level selection, which makes its Bayesian formulation more involved. It has a Laplace prior given by
\begin{equation}\label{eqn:bayesian_sgs_prior}
    \pi(\boldsymbol\beta \mid \sigma^2; \mathbf{w},\mathbf{v}) \propto \exp \left(-\frac{1}{\sigma} \alpha\sum_{i=1}^p v_{r_v(\boldsymbol\beta,i)} |\beta_i| -\frac{1}{\sigma} (1-\alpha)\sum_{j=1}^m w_{r_g(\boldsymbol\beta,j)}\sqrt{p_j}\|\boldsymbol\beta^{(j)}\|_2 \right),
\end{equation}
where, as for ABSLOPE and BGSLOPE, the sorting occurs through the weights and rank functions. 

The \textit{Bayesian Sparse-group SLOPE} (BSGS) is given by the hierarchical model
\begin{align}
\pi\left(\boldsymbol\beta \mid \boldsymbol\delta,\boldsymbol\gamma, c_g, c_v, \sigma^2; \mathbf{w},\mathbf{v} \right) &\propto c_g^{\sum_{j=1}^m p_j\mathbbm{1}(\gamma_j = 1)} c_v^{\sum_{i=1}^p \mathbbm{1}(\delta_i = 1)}\prod_{i=1}^p \exp \left\{-\frac{1}{\sigma} |\beta_i|\hat{a}_i \alpha v_{r_v(\hat{\mathbf{A}} \boldsymbol\beta, i)}\right\}\nonumber\\
&\quad\times \prod_{j=1}^m \exp \left\{-\frac{1}{\sigma} \|\hat{\mathbf{A}}^{(j)}\boldsymbol\beta^{(j)}\|_2\sqrt{p_j} (1-\alpha) w_{r_g(\hat{\mathbf{A}} \boldsymbol\beta, j)}\right\},\\ \label{eqn:bsgs_prior}
\pi(\boldsymbol\gamma \mid \theta_g) &=  \prod_{j=1}^m \theta_g^{\gamma_j}(1-\theta_g)^{1-\gamma_j}, \\
 \pi(\boldsymbol\delta \mid \boldsymbol\gamma, \theta_v) &= \prod_{j=1}^{m} \prod_{i \in G_j} \left[ \theta_v^{\delta_i} (1-\theta_v)^{1-\delta_i} \mathbbm{1}(\gamma_j = 1) + \mathbbm{1}(\gamma_j = 0, \delta_i = 0) \right],  \\
 \theta_g &\sim \text{Beta}(d_1,d_2), \quad \theta_v \sim \text{Beta}(e_1,e_2), \nonumber\\ \pi(\sigma^2) &\propto \sigma^{-2}, \quad c_g, c_v\sim\mathcal{U}[0,1],\nonumber
\end{align}
where $\boldsymbol\gamma\in\{0,1\}^m, \boldsymbol\delta\in\{0,1\}^p$ are the group/variable inclusion vectors, $\theta_g,\theta_v \in (0,1)$ are the group/variable mixture proportions (with $d_1,d_2, e_1,e_2>0$ prior hyperparameters) and $c_g,c_v \in (0,1)$ allow the signal group/variables to be penalized less. As with BGSLOPE, choosing $d_1 \ll d_2$ and $e_1 \ll e_2$ leads to sparser models. The weight matrix is denoted as $\hat{\mathbf{A}} =\text{diag}(\hat{a}_1,\ldots,\hat{a}_p)$ with elements, where $i \in G_j$,
    \begin{equation}
    \hat{a}_i = c_g\gamma_j(c_v\delta_i + 1- \delta_i) + (1-\gamma_j) = \begin{cases}
    c_g c_v & \text{if } \gamma_j = 1 \text{ and } \delta_i = 1, \\
    c_g & \text{if } \gamma_j = 1 \text{ and } \delta_i = 0, \\
     1 & \text{if } \gamma_j = 0. 
\end{cases}
    \end{equation}
The weight matrix penalizes the active groups less via $c_g$ and the active variables less by $c_gc_v$. Additionally, the hierarchical construction of $\boldsymbol\delta$ being dependent on $\boldsymbol\gamma$ enforces all variables in a non-active group to also be non-active. To see this, note that if $\gamma_j = 1$, the prior on $\boldsymbol\delta$ is an independent Bernoulli product (as for ABSLOPE and BGSLOPE), and if $\gamma_j = 0$, then the prior is a point-mass at zero.

It is possible to scale the signal strength ratio using $\alpha$ to have $c_g^{(1-\alpha)\sum_{j=1}^m p_j\mathbbm{1}(\gamma_j = 1)} c_v^{\alpha \sum_{i=1}^p \mathbbm{1}(\delta_i = 1)}$ in Equation \ref{eqn:bsgs_prior}, which allows the hyperparameters to have a direct interpretive connection to ABSLOPE and BGSLOPE. However, as this adds complexity and interpretation between models is not a key priority, it is not pursued further.

BSGS provides uncertainty quantification and feature importance at two levels, group and variable, through the inclusion parameters $\boldsymbol{\gamma}$ and $\boldsymbol{\delta}$, respectively. Moreover, the dual mixing proportions $\theta_g$ and $\theta_v$ enable BSGS to adapt simultaneously to the underlying sparsity structures at both levels. As seen in Figure \ref{fig:bsgs_graph}, BSGS has a more complex dependency graph, with twice as many parameters. Proposition \ref{propn:sgs_equiv} demonstrates that the posterior mode from the BSGS prior is equivalent to the estimates obtained via the SGS Laplace prior under the input $\vect{X}\hat{\vect{A}}^{-1}$. This also makes it equivalent to the frequentist model (the proof is given in Appendix \ref{appendix:bsgs_theory}). 

\begin{proposition}\label{propn:sgs_equiv}
If the coefficients $\mathbf{z} = (z_1, z_2, \ldots, z_p)$ follow the SGS Laplace prior (Equation \ref{eqn:bayesian_sgs_prior}), then $\boldsymbol{\beta} = \hat{\mathbf{A}}^{-1}\mathbf{z}$ follows the BSGS prior (Equation \ref{eqn:bsgs_prior}).
\end{proposition}

\subsection{Extensions}\label{section:extentions}
While developed for SLOPE, our methodology extends naturally to broader model classes. BGSLOPE and BSGS apply to any Ordered Weighted $\ell_1$ (OWL) model, such as OSCAR \citep{oscar}, by replacing the corresponding penalties. The BSGS framework further generalizes to any sparse-group model, including the sparse-group lasso \citep{Simon2013} and sparse-group SCAD \citep{Buch2024SparseSelection}, with efficiency gains available via sparse-group screening rules \citep{Feser2024DualVariant,Liang2022Sparsegl:Lasso,Ndiaye2016GAPLasso}.

\subsection{Stochastic Approximation EM (SAEM) algorithm}\label{section:saem}
Both BGSLOPE and BSGS are fitted using SAEM. A description of it for the BGSLOPE case is presented in this section. The aim is to estimate $\Theta=(\boldsymbol\beta,\sigma)$ in the presence of latent variables $\Lambda 
= (\boldsymbol\gamma, c, \theta)$. The employment of expectation–maximization (EM) algorithms for spike-and-slab models is a standard technique. For example, the SSL model is fitted via the EMVS algorithm \citep{Rockova2014EMVS:Selection}. 

There are two steps to an EM algorithm (for step $t$): 
\begin{enumerate}
    \item \textbf{Expectation step (E step)}: Calculate the expected value of the log-likelihood 
\begin{equation}
     Q(\Theta \mid \Theta_{[t]}) = \mathbb{E}_{\Lambda \sim \pi(\cdot \mid \mathbf{y}, \Theta_{[t]})} \left[ \log \pi(\mathbf{y}, \Lambda \mid \Theta_{[t]}) \right],
\end{equation}
where $\Theta_{[t]}$ denotes the value of the parameters at step $t$.
\item \textbf{Maximization step (M step)}: Find the parameters that maximize $Q$: $\Theta_{[t+1]} = \argmax_\Theta Q(\Theta \mid \Theta_{[t]}).$
\end{enumerate}
However, for SLOPE models, $Q$ is not tractable due to the sorting procedures. Instead, the \textit{Stochastic Approximation EM} (SAEM) algorithm \citep{Lavielle2014MixedApproach} is used, in which the E step is replaced by a simulation step, followed by a stochastic approximation step. These are given by
\begin{itemize}
    \item \textit{Simulation step}: Sample the latent variables $\Lambda$ from $\pi(\Lambda \mid \mathbf{y}, \Theta_{[t-1]})$ using a Gibbs sampler.
\item \textit{Stochastic approximation}: Calculate $Q$ as
\begin{equation}
      Q(\Theta \mid \Theta_{[t]}) = Q(\Theta \mid\Theta_{[t-1]}) + \eta_t \left( \log \pi(\mathbf{y}, \Lambda \mid \Theta_{[t]}) - Q(\Theta \mid\Theta_{[t-1]}) \right),
\end{equation}
where $\eta_t$ is the step size, which is chosen as a decreasing sequence $\eta_t = 1/t$ that guarantees almost sure convergence \citep{Delyon1999ConvergenceAlgorithm}. The M step remains the same.
\end{itemize}

SAEM estimates $\boldsymbol{\beta}$ and simulates $\boldsymbol{\gamma}$ to allow for variable selection. It has two loosely defined stages: first, a quick search to find the solution neighborhood, followed by a steady fine-tuning stage \citep{Lavielle2014MixedApproach} (these steps can be seen in the illustrative example in Section \ref{section:illustrative-example}). SAEM is generally not sensitive to initialization, as demonstrated in the study presented in Section \ref{section:sensitivity_analysis}, which is in agreement with findings in the literature \citep{Lavielle2014MixedApproach}. The stochastic approximation is applied only after the $20$th iteration: $\eta_t = 1$ if $t \leq 20$. 

To determine the final set of active variables/groups, median thresholding is applied to the inclusion parameters over the last $T=20$ iterations of SAEM: a variable $i\in [p]$ is active if $\mathbb{P}(\delta_i = 1 \mid \mathbf{y}) \approx \frac{1}{T}\sum_{t\in \mathcal{T}}(\delta_{[t]})_i > 0.5$ and a group $j\in [m]$ is active if $\mathbb{P}(\gamma_j = 1 \mid \mathbf{y})\approx \frac{1}{T}\sum_{t\in \mathcal{T}}(\gamma_{[t]})_j > 0.5$, where $\mathcal{T}$ denotes the index set of the last $T$ iterations. In BSGS, group thresholding is not necessary due to the hierarchical structure between $\boldsymbol\gamma$ and $\boldsymbol\delta$.

SAEM is described fully for BGSLOPE and BSGS in Algorithms \ref{alg:bgslope_saem_alg} and \ref{alg:bsgs_saem_alg}, respectively. The SAEM updates for BGSLOPE and BSGS are derived next. This is followed by the development of the BSGS-$\alpha$ model, which additionally updates the $\alpha$ parameter, and the SLOBE models, which approximate the conditional distribution by its conditional expectation to accelerate the algorithm.

\subsubsection{Updates for BGSLOPE}\label{section:gslope_saem_algorithm}
For BGSLOPE, the penalized log-likelihood, which is used to derive the maximization steps, is given by
\begin{align}
  \log \pi(\vect{y}, \Lambda \mid \Theta_{[t]})&=  \log \pi\left(\vect{y}, \vect\gamma, c ; \boldsymbol\beta, \theta, \sigma^2\right)+  \log \pi\left(\boldsymbol\beta \mid \vect\gamma, c, \sigma^2 ; \vect{w} \right)\label{eqn:gslope_likelihood}  \\
&= \log\pi\left(\vect{y} \mid \vect{X} ; \boldsymbol\beta, \sigma^2\right) +\log\pi(\vect\gamma \mid \theta)+\log\pi(\sigma^2)\nonumber\\
&\quad+\log \pi\left(\boldsymbol\beta \mid \vect\gamma, c, \sigma^2 ; \vect{w} \right)\nonumber \\
&=-(n + 2)\log\sigma - \frac{1}{2\sigma^2} \| \vect{y} - \vect{X}\boldsymbol\beta \|^2_2 + \sum_{j=1}^{m} \mathbbm{1}(\gamma_j = 1) \log \theta \nonumber\\
&\quad+ \sum_{j=1}^{m} \mathbbm{1}(\gamma_j = 0) \log(1 - \theta) +\log c \sum_{j=1}^m p_j \mathbbm{1}(\gamma_j=1)  \nonumber \\
&\quad-\frac{1}{\sigma} \sum_{j=1}^{m} \tilde{a}_{j} \sqrt{p_j}  w_{r_g(\vect{\tilde{A}}\boldsymbol\beta, j)}\|\boldsymbol\beta^{(j)}\|_2,\nonumber
\end{align}
where the prior term on $c$ is omitted since it is a standard uniform.

\paragraph{Simulation step.}
The following Gibbs sampler regime is used for the simulation step, for $j\in [m]$,
\begin{align}
     \gamma_j &\sim \operatorname{Bernoulli}\left(\frac{L_1}{L_1 + L_2}\right), \; \text{where},\label{eqn:bgslope_gamma_update} \\
     &\quad \;  L_1 = \theta c^{p_j} \exp\left\{\frac{-c}{\sigma}\|\vect\beta^{(j)}\|_2 \sqrt{p_j}w_{r_g(\tilde{\vect{A}}\vect\beta,j)}\right\},\nonumber\\ 
     &\quad \; L_2 = (1-\theta)\exp\left\{\frac{-1}{\sigma}\|\vect\beta^{(j)}\|_2 \sqrt{p_j}w_{r_g(\tilde{\vect{A}}\vect\beta,j)}\right\},\nonumber \\
    \theta&\sim\text{Beta}\left(d_1+\sum_{j=1}^m \mathbbm{1}(\gamma_j=1),d_2+\sum_{j=1}^m\mathbbm{1}(\gamma_j=0)\right),\label{eqn:bgslope_theta_update}\\
    c &\sim \text{Gamma}\left( 1 + \sum_{j=1}^{m}p_j \mathbbm{1}(\gamma_j = 1), \frac{1}{\sigma} \sum_{j=1}^{m} \|\vect\beta^{(j)}\|_2\sqrt{p_j} w_{r_g(\tilde{\vect{A}} \vect\beta, j)}\mathbbm{1}(\gamma_j=1)\right),\label{eqn:bgslope_c_update}\\
        &\quad \; \text{truncated at} \; [0,1]. \nonumber
\end{align}
The derivations of these are given in Appendix \ref{appendix:saem_bgslope_simulation}.

\paragraph{Stochastic approximation and maximization steps.} After simulating the latent variables, we use these to complete the maximization step and compute $\Theta_{[t+1]}$. For a general step size $\eta_t$, we update the model parameters as
\begin{equation}\label{eqn:mle_update}
    \vect\beta_{[t+1]} = \vect\beta_{[t]} + \eta_t(\vect\beta^\text{MLE}_{[t]}-  \vect\beta_{[t]}),\quad \sigma_{[t+1]} = \sigma_{[t]} + \eta_t(\sigma^\text{MLE}_{[t]}-\sigma_{[t]}),
\end{equation}
where $\vect\beta^\text{MLE}_{[t]}$ and $\sigma^\text{MLE}_{[t]}$ are the \textit{Maximum Likelihood Estimation} (MLE) estimators of the log-likelihood (Equation \ref{eqn:gslope_likelihood}), using samples of the latent variables at step $t$. Considering only the terms involving the parameters of interest in the log-likelihood, they are given by
\begin{align}
&\vect\beta^\text{MLE}_{[t]} = \argmin_{\vect{b}\in \mathbb{R}^p}\left\{\frac{1}{2n}\|\vect{y}-\vect{X}\vect{b}\|_2^2 + \frac{\sigma_{[t-1]}}{n}\sum_{j=1}^m(\tilde{a}_{j})_{[t]} \sqrt{p_j}w_{r_g(\tilde{\vect{A}}_{[t]}\vect{b},j)}  \|\vect{b}^{(j)}\|_2\right\},\label{eqn:gslope_opt_problem}\\
     &\sigma^\text{MLE}_{[t]}=\frac{K_2 + \sqrt{K_2^2 + 4K_1(n+2)}}{2(n+2)},\label{eqn:gslope_opt_problem_sigma}\\
   &\text{where}\; K_1=\|\vect{y}-\vect{X}\vect\beta_{[t]}\|_2^2, \; K_2=\sum_{j=1}^m (\tilde{a}_{j})_{[t]} \|\vect\beta_{[t]}^{(j)}\|_2 \sqrt{p_j}  w_{r_g(\tilde{\vect{A}}_{[t]}\vect\beta_{[t]}, j)}\nonumber.
\end{align}

\begin{remark}\label{remark:beta_update}
To calculate the update for $\vect\beta^\text{MLE}$, consider the  transformation $\vect{z} = \tilde{\vect{A}}\vect\beta^\text{MLE}$, so that Equation \ref{eqn:gslope_opt_problem} becomes
\begin{equation*}
\vect{z}_{[t]} = \argmin_{\vect{z}\in \mathbb{R}^p}  \Bigg\{ \frac{1}{2n} \| \vect{y} - \vect{X} \tilde{\vect{A}}_{[t]}^{-1} \vect{z} \|_2^2 
+ \frac{\sigma_{[t-1]}}{n} \sum_{j=1}^m \sqrt{p_j} w_{r_g(\vect{z},j)} \| \vect{z}^{(j)} \|_2 \Bigg\}.
\end{equation*}
Following Proposition \ref{propn:gslope_equiv}, the update reduces to a standard gSLOPE problem with input matrix $\vect{X}\tilde{\vect{A}}_{[t]}^{-1}$ and regularization parameter $\lambda = \sigma_{[t-1]}/n$, yielding $\vect{z}_{[t]}$, with MLE $\vect\beta^\text{MLE}_{[t]} = \tilde{\vect{A}}_{[t]}^{-1} \vect{z}_{[t]}$. Since $\tilde{\vect{A}}$ is diagonal, the inverse is inexpensive, and sorting the modified penalty sequence is handled via $\vect{z}$. The optimization is solved using Adaptive Three Operator Splitting (ATOS) \citep{Pedregosa2018AdaptiveSplitting}. See Appendix \ref{appendix:saem_bgslope_opt} for the derivation of the $\sigma$ update.
\end{remark}

\subsubsection{Updates for BSGS}\label{section:sgs_saem_algorithm}
For BSGS, we define the latent variables as $\Lambda' = (\boldsymbol\gamma,\boldsymbol\delta,c_g,c_v,\theta_g,\theta_v)$ and the model parameters are as for BGSLOPE: $\Theta = (\boldsymbol\beta,\sigma$). Therefore, the penalized log-likelihood, which is used to derive the maximization steps, is given by
\begin{align}
&\log \pi(\vect{y}, \Lambda' \mid \Theta_{[t]}) =\log\pi\left(\vect{y} \mid \vect{X} ; \boldsymbol\beta, \sigma^2\right) +\log\pi(\vect\gamma \mid \theta_g)+\log\pi(\vect\delta \mid \vect\gamma,\theta_v)+\log\pi(\sigma^2) \nonumber\\
 &\quad\quad\quad\quad\quad\quad\quad\quad\quad +\log \pi\left(\boldsymbol\beta \mid \vect\gamma, \vect{\delta}, c_g, c_v, \sigma^2 ; \vect{w}, \vect{v} \right)\label{eqn:sgs_likelihood}\\
 &=-(n + 2)\log\sigma -\frac{1}{2\sigma^2}\|\vect{y}-\vect{X}\boldsymbol\beta\|^2_2 + \sum_{j=1}^{m} \mathbbm{1}(\gamma_j = 1) \log \theta_g + \sum_{j=1}^{m} \mathbbm{1}(\gamma_j = 0) \log(1 - \theta_g)\nonumber \\
 &\quad+\sum_{j=1}^{m} \sum_{i \in G_j} \log \left[ \mathbbm{1}(\gamma_j = 1) \left( \theta_v^{\delta_i} (1-\theta_v)^{1-\delta_i} \right) + \mathbbm{1}(\gamma_j = 0, \delta_i = 0) \right]\nonumber\\
 &\quad+\log c_g  \sum_{j=1}^m p_j \mathbbm{1}(\gamma_j = 1)
+\log c_v  \sum_{i=1}^p \mathbbm{1}(\delta_i = 1) \nonumber \\
 &\quad-\sum_{i=1}^p \frac{1}{\sigma} |\beta_i|\hat{a}_i v_{r_v(\hat{\vect{A}} \boldsymbol\beta, i)} -\sum_{j=1}^m \frac{1}{\sigma} \|\hat{\vect{A}}^{(j)}\boldsymbol\beta^{(j)}\|_2\sqrt{p_j}  w_{r_g(\hat{\vect{A}} \boldsymbol\beta, j)}.\nonumber
\end{align}
\paragraph{Simulation step.} 
A Gibbs sampler is used for the simulation step (the derivations of these updates are provided in Appendix \ref{appendix:saem_bgslope_simulation}). For the inclusion parameters, we sample from Bernoulli distributions
\begin{align}
&\gamma_j\sim \text{Bernoulli}\left(\frac{L_1'}{L_1'+L_2'}\right), \; j \in [m], \; \text{where},\label{eqn:bsgs_gamma_update}\\
&\quad \;L_1'= \theta_g c_g^{p_j} c_v^{\sum_{i\in G_j}\mathbbm{1}\{\delta_{i} = 1\}} \exp \left\{-\frac{c_g}{\sigma}\|\vect\kappa^{(j)}\vect\beta^{(j)}\|_2\sqrt{p_j}  (1-\alpha) w_{r_g(\hat{\vect{A}} \vect\beta, j)}\right\}\label{eqn:l1_defn}\\
&\quad\quad \;\times \prod_{i \in G_j} \theta_v^{\delta_i} (1-\theta_v)^{1-\delta_i} \exp \left\{-\frac{c_g}{\sigma} |\beta_i| c_v^{\mathbbm{1}(\delta_i =1)} \alpha v_{r_v(\hat{\vect{A}} \vect\beta, i)}\right\},\nonumber\\
&\quad \; L_2' = (1-\theta_g)\exp \left\{-\frac{1}{\sigma}\|\vect\beta^{(j)}\|_2\sqrt{p_j} (1-\alpha) w_{r_g(\hat{\vect{A}} \vect\beta, j)}\right\} \label{eqn:l2_defn}\\
&\quad\quad \; \times \prod_{i\in G_j} \exp \left\{-\frac{1}{\sigma} |\beta_i| \alpha v_{r_v(\hat{\vect{A}} \vect\beta, i)}\right\},\nonumber \\
&\delta_i\sim\text{Bernoulli}\left(\frac{\tilde{L}_1}{\tilde{L}_1 + \tilde{L}_2}\right), \; i \in [p], \; \text{where},\label{eqn:bsgs_delta_update}\\
&\quad \; \tilde{L}_1 = \theta_v c_v \exp \left\{-\frac{1}{\sigma} |\beta_i|c_g c_v \alpha v_{r_v(\hat{\vect{A}} \vect\beta, i)}\right\}, \; \tilde{L}_2  = (1-\theta_v) \exp \left\{-\frac{1}{\sigma} |\beta_i|c_g \alpha v_{r_v(\hat{\vect{A}} \vect\beta, i)}\right\}. \nonumber
\end{align}
The mixing parameters are sampled via Beta distributions, given by
\begin{align}
\theta_g&\sim\text{Beta}\left(d_1+\sum_{j=1}^m \mathbbm{1}(\gamma_j=1),d_2+\sum_{j=1}^m\mathbbm{1}(\gamma_j=0)\right),\label{eqn:bsgs_theta_g_update}\\
\theta_v&\sim  \text{Beta}\Bigg(me_1-m + \sum_{j=1}^m\mathbbm{1}(\gamma_j = 1)\sum_{i \in G_j}\mathbbm{1}(\delta_i=1) + 1, \nonumber \\
&\quad\quad\quad\quad\quad\quad\quad\quad me_2-m +\sum_{j=1}^m\mathbbm{1}(\gamma_j = 1) \sum_{i \in G_j}\mathbbm{1}(\delta_i=0) + 1\Bigg).\label{eqn:bsgs_theta_v_update}
\end{align}
Finally, the signal strength ratios are sampled as
\begin{align}
&c_g \sim  \text{Gamma}\Bigg(1 + \sum_{j=1}^m p_j \mathbbm{1}(\gamma_j = 1),  \label{eqn:bsgs_c_g_update}\\
&\frac{1}{\sigma}\Bigg[\sum_{i=1}^p c_v|\beta_i| \alpha v_{r_v(\hat{\vect{A}} \vect\beta, i)}\mathbbm{1}(\delta_i=1)+ \sum_{j=1}^m \|\vect\kappa^{(j)}\vect\beta^{(j)}\|_2\sqrt{p_j} (1-\alpha) w_{r_g(\hat{\vect{A}} \vect\beta, j)}\mathbbm{1}(\gamma_j=1)\Bigg]\Bigg), \nonumber \\ &\quad \; \text{truncated at}\; [0,1], \nonumber\\
&c_v \;\text{sampled via Metropolis-Hastings (MH) with proposal} \;\text{Gamma}(2,2),\label{eqn:bsgs_c_v_update}\\
&\quad \;\text{where} \; \vect\kappa^{(j)} = \text{diag}(c_v^{\mathbbm{1}(\delta_i = 1)}), \forall  i \in G_j.\nonumber
\end{align}

\paragraph{Stochastic approximation and maximization steps.}
After the simulation step, the model parameters are given by Equation \ref{eqn:mle_update} with the MLE estimators
\begin{align}
&\vect\beta_{[t]}^\text{MLE} = \argmin_{\vect{b}\in \mathbb{R}^p} \Bigg\{ \frac{1}{2n}\|\vect{y}-\vect{X}\vect{b}\|_2^2+\frac{\sigma_{[t-1]}}{n}\Bigg(\alpha\sum_{i=1}^p|b_i|(\hat{a}_i)_{[t]} v_{r_v(\hat{\vect{A}}_{[t]} \vect{b}, i)} \nonumber\\
&\quad\quad\quad\quad\quad\quad\quad\quad\quad\quad\quad+(1-\alpha)\sum_{j=1}^m \|\hat{\vect{A}}^{(j)}_{[t]}\vect{b}^{(j)}\|_2\sqrt{p_j}  w_{r_g(\hat{\vect{A}}_{[t]} \vect{b}, j)}\Bigg)\Bigg\}, \label{eqn:sgs_opt_problem}\\
 &\sigma^\text{MLE}_{[t]} =\frac{K_2' + \sqrt{(K_2')^2 + 4K_1'(n+2)}}{2(n+2)},\; \text{where},\label{eqn:sgs_opt_problem_sigma}\\
&K_1'=\|\vect{y}-\vect{X}\vect\beta_{[t]}\|_2^2, \nonumber\\
&K_2'=\alpha\sum_{i=1}^p|(\beta_{i})_{[t]}|(\hat{a}_i)_{[t]} v_{r_v(\hat{\vect{A}}_{[t]} \vect\beta_{[t]}, i)} +(1-\alpha)\sum_{j=1}^m \|\hat{\vect{A}}^{(j)}_{[t]}\vect\beta_{[t]}^{(j)}\|_2\sqrt{p_j}  w_{r_g(\hat{\vect{A}}_{[t]} \vect\beta_{[t]}, j)}. \nonumber
\end{align}
The $\vect\beta$ update is the SGS optimization task. As with BGSLOPE, this is calculated via a transformation with input $\vect{X} \hat{\vect{A}}^{-1}_{[t]}$ and regularization parameter 
$\lambda = \sigma_{[t-1]}/n$; see Remark \ref{remark:beta_update}. It is performed using ATOS. See Appendix \ref{appendix:saem_bsgs_opt} for the derivation of the $\sigma$ update.

\subsection{BSGS-$\alpha$} 
It is also possible to learn the $\alpha$ model parameter for BSGS, which defines the balance between the two types of penalization, by placing a uniform prior on it. The posterior for $\alpha$ is not available in closed-form, so it is sampled via MH using a $\text{Beta}(10,0.5)$ proposal, which skews $\alpha$ towards $1$ (see Appendix \ref{appendix:bsgs-alpha-posterior}). We denote this model by BSGS-$\alpha$. To the best of our knowledge, there are no other instances in the literature of $\alpha$ being updated in a Bayesian model.

\subsubsection{SLOBE models}\label{section:slobe}
The SAEM algorithm for BGSLOPE and BSGS can be accelerated by approximating samples from the conditional distributions of the latent parameters using their conditional expectations. That is, $\Lambda_{[t]} \leftarrow \mathbb{E}[\Lambda_{[t]} \mid \vect{y},\Theta_{[t-1]}].$ Besides speeding up computation, this reduces latent variable variability, which helps diminish algorithmic noise in high dimensions.

ABSLOPE was modified in this way to form \textit{SLOBE} in \citet{Jiang2022AdaptiveData}. Here, we apply the SLOBE acceleration to BGSLOPE and BSGS, to form \textit{GSLOBE} and \textit{SGSLOBE}, respectively. The SAEM procedures for GSLOBE and SGSLOBE are as in Algorithms \ref{alg:bgslope_saem_alg} and \ref{alg:bsgs_saem_alg}, with the simulation step replaced by the conditional expectation approximations.

\paragraph{GSLOBE.}
Taking the conditional expectations of the simulation updates for BGSLOPE (Equations \ref{eqn:bgslope_gamma_update}--\ref{eqn:bgslope_c_update}):
\begin{enumerate}
    \item Approximate each $\gamma_j, j\in [m]$, by the active probability of a Bernoulli distribution, where $\vect{\gamma_{-j}}$ is the vector $\vect{\gamma}$ with the $j$th entry removed,
    \begin{align*}
&\mathbb{E}[\gamma_j = 1 \mid \vect{\gamma}_{-j}, c, \vect\beta, \sigma, \theta, \tilde{\vect{A}}] \\
&= \frac{\theta c^{p_j} \exp\left\{\frac{-c}{\sigma}\|\vect\beta^{(j)}\|_2 \sqrt{p_j}w_{r_g(\tilde{\vect{A}}\vect\beta,j)}\right\}}{(1-\theta)\exp\left\{\frac{-1}{\sigma}\|\vect\beta^{(j)}\|_2 \sqrt{p_j}w_{r_g(\tilde{\vect{A}}\vect\beta,j)}\right\}+\theta c^{p_j} \exp\left\{\frac{-c}{\sigma}\|\vect\beta^{(j)}\|_2 \sqrt{p_j}w_{r_g(\tilde{\vect{A}}\vect\beta,j)}\right\}}.
\end{align*}
\item Approximate $\theta$ by the mean of a Beta distribution
    \begin{equation*}
\mathbb{E}[\theta \mid \vect\gamma, \vect{y}, \vect\beta, \sigma, c, \tilde{\vect{A}}]
= \mathbb{E}[\theta \mid \vect\gamma, \vect\beta, \sigma, \tilde{\vect{A}}]
= \frac{d_1 + \sum_{j=1}^m \mathbbm{1}(\gamma_j = 1)}{d_1 + d_2 + m}.
\end{equation*}
    \item Approximate $c$ by the mean of a truncated Gamma distribution
\begin{align*}
&\mathbb{E}[c \mid \vect\gamma, \vect{y}, \vect\beta, \sigma, \theta, \tilde{\vect{A}}]
= \frac{\int_0^1 z^{x} \exp\{-x'z\} \, dz}{\int_0^1 z^{x - 1} \exp\{-x'z\} \, dz},\quad\text{where,}\\
&x =1 + \sum_{j=1}^{m}p_j \mathbbm{1}(\gamma_j = 1), \; x' = \frac{1}{\sigma} \sum_{j=1}^{m} \|\vect\beta^{(j)}\|_2\sqrt{p_j} w_{r_g(\tilde{\vect{A}} \vect\beta, j)}\mathbbm{1}(\gamma_j=1).
\end{align*}
\end{enumerate}
\paragraph{SGSLOBE.}
Taking the conditional expectations of the simulation updates for BSGS (Equations \ref{eqn:bsgs_gamma_update}--\ref{eqn:bsgs_c_g_update}; the update for $c_v$ remains the same via MH):
\begin{enumerate}
    \item Approximate the inclusion probabilities by the active probability of a Bernoulli distribution (where $L_1'$ and $L_2'$ are as described in Equations \ref{eqn:l1_defn} and \ref{eqn:l2_defn}), for each $j\in [m]$ and $i\in [p]$,
    \begin{align*}
       & \mathbb{E}[\gamma_j = 1 \mid \vect\gamma_{-j}, c_g,c_v, \vect\beta, \sigma, \theta_g,\theta_v,\vect\delta, \hat{\vect{A}}] =\frac{L_1'}{L_1' + L_2'},\\
       & \mathbb{E}[\delta_i = 1 \mid \vect\delta_{-i}, c_g,c_v, \vect\beta, \sigma, \theta_g,\theta_v,\vect\gamma, \hat{\vect{A}}] \\
&=\left(\frac{\theta_v c_v \exp \left\{-\frac{1}{\sigma} |\beta_i|c_g c_v \alpha v_{r_v(\hat{\vect{A}} \vect\beta, i)}\right\}}{\theta_v c_v \exp \left\{-\frac{1}{\sigma} |\beta_i|c_g c_v \alpha v_{r_v(\hat{\vect{A}} \vect\beta, i)}\right\}+(1-\theta_v) \exp \left\{-\frac{1}{\sigma} |\beta_i|c_g \alpha v_{r_v(\hat{\vect{A}} \vect\beta, i)}\right\}}\right).
    \end{align*}
    \item Approximate $\theta_g$ and $\theta_v$ by the means of Beta distributions
    \begin{align*}
        &\mathbb{E}[\theta_g \mid \vect\gamma, \vect{y}, \theta_v,\vect\beta, \sigma, c_v,c_g, \vect\delta, \hat{\vect{A}}]= \mathbb{E}[\theta_g \mid \vect\gamma, \vect\beta, \sigma, \hat{\vect{A}},\vect\delta]= \frac{d_1+\sum_{j=1}^m \mathbbm{1}(\gamma_j=1)}{d_1+d_2+m},\\
        &\mathbb{E}[\theta_v \mid \vect\gamma, \vect{y}, \theta_g,\vect\beta, \sigma, c_v,c_g, \vect\delta, \hat{\vect{A}}]\\
&= \mathbb{E}[\theta_v \mid \vect\gamma, \vect\beta, \sigma, \hat{\vect{A}},\vect\delta]
= \frac{me_1-m + x + 1}{me_1 + x + 2+ me_2-2m +x'}, \; \text{where}, \\
& x=\sum_{j=1}^m\mathbbm{1}(\gamma_j = 1)\sum_{i \in G_j}\mathbbm{1}(\delta_i=1), \; x' = \sum_{j=1}^m\mathbbm{1}(\gamma_j = 1) \sum_{i \in G_j}\mathbbm{1}(\delta_i=0).
    \end{align*}

    \item Approximate $c_g$ by the mean of a truncated Gamma distribution
\begin{align*}
&\mathbb{E}[c_g \mid \vect\gamma, \vect{y}, \vect\beta, \sigma, \theta, \hat{\vect{A}}]
= \frac{\int_0^1 z^{x} \exp\{-x'z\} \, dz}{\int_0^1 z^{x - 1} \exp\{-x'z\} \, dz},\quad\text{where,}\\
&x =1 + \sum_{j=1}^m p_j \mathbbm{1}(\gamma_j = 1),\\ 
&x' =\frac{1}{\sigma}\Bigg[\sum_{i=1}^p c_v|\beta_i| v_{r_v(\hat{\vect{A}} \vect\beta, i)}\mathbbm{1}(\delta_i=1)+ \sum_{j=1}^m \|\vect\kappa^{(j)}\vect\beta^{(j)}\|_2\sqrt{p_j}w_{r_g(\hat{\vect{A}} \vect\beta, j)}\mathbbm{1}(\gamma_j=1)\Bigg].
\end{align*}
\end{enumerate}
\subsection{Illustrative example}\label{section:illustrative-example}
BSGS and BGSLOPE are applied to a toy example to demonstrate the fitting process, with data generated following the baseline setup from Section \ref{section:sim_study_set_up}: a block-correlated multivariate Gaussian design is used to generate linear responses with group-structured sparsity.

BGSLOPE and BSGS were applied with the initial $\vect\beta$ values generated from a lasso model and prior hyperparameters for the $\theta$ distributions given by $d_1=d_2=0.01n = 4$ for BGSLOPE and $d_1 =e_1= 0.003n = 1.2, d_2 =e_2 = 0.015n = 6$ for BSGS (these initializations are explored in Section \ref{section:sensitivity_analysis}). BSGS used $\alpha = 0.95$ and other parameters were set as described in Table \ref{tbl:appendix_model_data_simulation}.
\begin{figure}[H]
    \centering
    \includegraphics[width=1\linewidth]{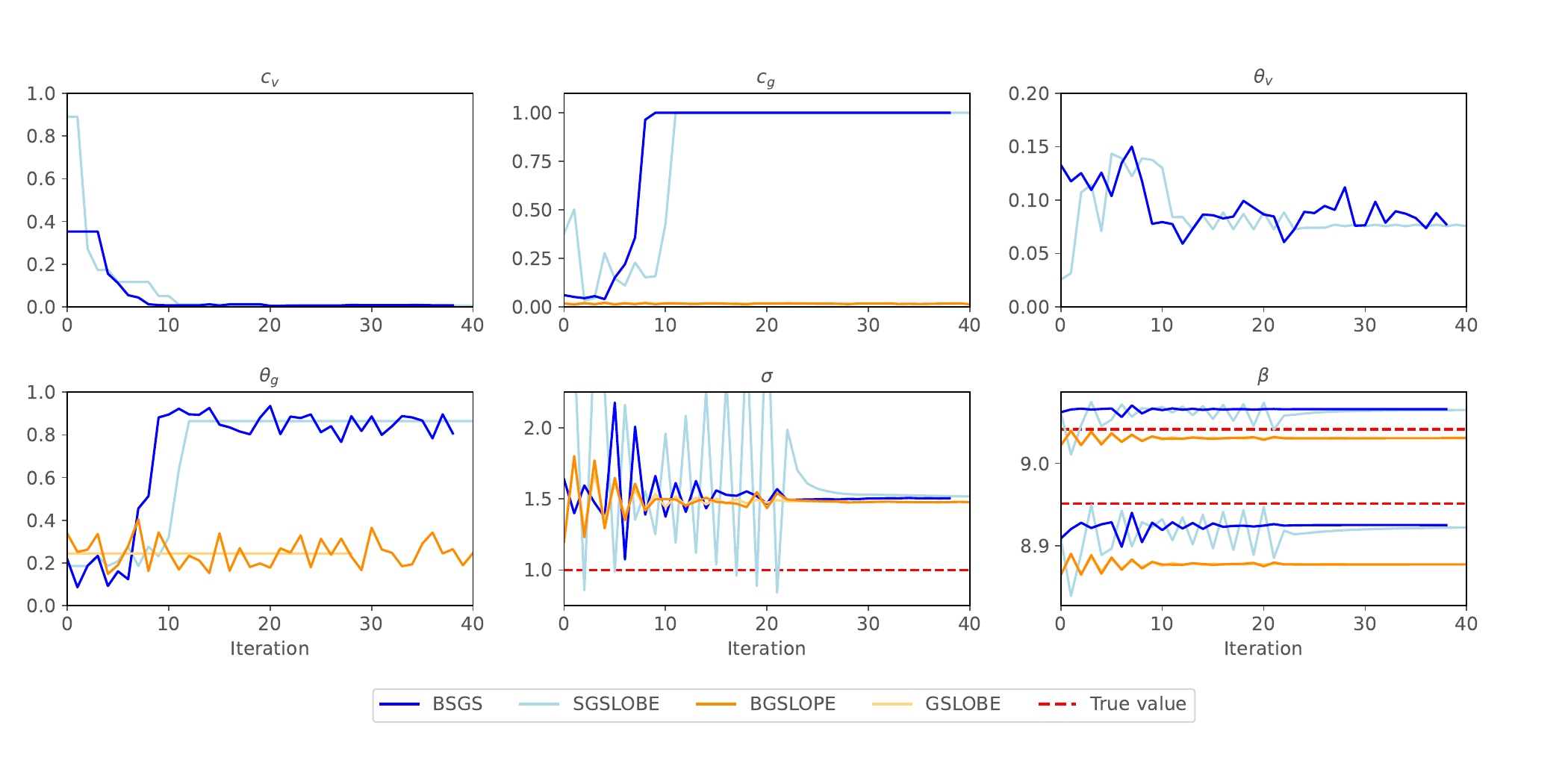}
    \caption[Bayesian parameters for illustrative example]{Bayesian latent variables and model parameters for BSGS, SGSLOBE, BGSLOPE, and GSLOBE, applied to the toy example, shown for the first $40$ iterations. The true values for $\sigma,\vect\beta$ are shown in red. Two active $\vect\beta$ values are shown, and the $c,\theta$ values for BGSLOPE and GSLOBE are placed into the $c_g,\theta_g$ plots. BSGS converged in $39$ iterations, SGSLOBE in $37$, BGSLOPE in $145$, and GSLOBE in $29$.}
    \label{fig:illustrative-example}
\end{figure}
Figure \ref{fig:illustrative-example} shows that the parameters very quickly enter into a period of stability. The reduced variability of the SLOBE models is highlighted well in the plot for $\theta_g$, where the conditional mean is clearly seen. BSGS and BGSLOPE appear equally accurate in recovering the true parameter values. Figure \ref{fig:illustrative-example-likelihood} displays the likelihood traces for each model, which shows that the SLOBE models achieve almost identical likelihood values.

SAEM has two phases \citep{Lavielle2014MixedApproach}: the first finds a neighborhood quickly (visible around iteration $8$ for BSGS via rapid jumps), and the second converges to the maximum, analogously to gradient descent (iterations $10$--$20$ for $\vect{\beta}$ and $\sigma$). The reduced fluctuations after $t=20$ follow from the stochastic approximation activating at that iteration. 

Figure \ref{fig:illustrative-example-alpha} shows trace plots for BSGS-$\alpha$, which converged to $\alpha = 0.68$ but with greater variability and slower convergence (202 iterations) than BSGS. For BGSLOPE, GSLOBE closely matched the full model. SGSLOBE deviated more frequently from BSGS due to its additional estimated parameters, with further divergence observed for BSGS-$\alpha$.   

\subsection{Initializations and sensitivity analysis}
\label{section:sensitivity_analysis}
Setting the initial parameter values for a Bayesian model is an important part of obtaining an accurate final model. The parameters for BGSLOPE and BSGS are initialized as
\begin{itemize}
\item $\sigma^0= \frac{\|\vect{y}-\vect{X}\vect\beta_{[0]}\|_2}{\sqrt{n-|\hat{S}_v^0|}}$, where $|\hat{S}_v^0|$ denotes the number of non-zero variables from $\vect\beta_{[0]}$. For BGSLOPE, this is taken as the number of non-zero groups instead, $|\hat{S}_g^0|$. This is a widely used estimator of the noise \citep{Dicker2014VarianceModels,Fan2012VarianceRegression,Sun2012ScaledRegression}, and is also used for ABSLOPE.
    \item $\theta.$ Using the mean of the posterior Beta distribution: 
    \begin{itemize}
        \item BGSLOPE: $\theta^0 = \frac{d_1+|\hat{S}_g^0|}{d_1+d_2+m}$.
        \item BSGS: $\theta_v^0= \frac{me_1-m + |\hat{S}_v^0| + 1}{2+ me_1+me_2-2m + |\hat{S}_v^0| +\sum_{j\in \hat{S}^0_g} |\hat{S}_0^v \cap G_j|}$,  where $\hat{S}_0^v \cap G_j$ are the non-zero variables in group $j$, and $\theta_g^0=\frac{d_1+|\hat{S}_g^0|}{d_1+d_2+m}$.
    \end{itemize}
    \item $c.$ Using the means of the posterior Gamma distributions: 
    \begin{itemize}
        \item BGSLOPE: $c^0 = \min\left\{ \frac{\sigma^0}{w_{r(\vect\beta_{[0]}, m)}}\frac{1+ \sum_{j\in \hat{S}_g^0} p_j} {\sum_{j=1}^{m} \|(\vect\beta_{[0]})^{(j)}\|_2\sqrt{p_j}},1\right\}$.
        \item BSGS: $c^0_g = \min\left\{ \frac{\sigma^0\left(1 + \sum_{j\in \hat{S}_g^0}p_j\right) }{v_{r_v(\vect\beta_{[0]}, p)}\sum_{i\in \hat{S}_v^0} c_v|(\beta_{[0]})_i| + w_{r_g( \vect\beta_{[0]}, m)}\sum_{j\in \hat{S}_g^0} \|\vect\kappa^{(j)}(\vect\beta_{[0]})^{(j)}\|_2\sqrt{p_j}},1\right\}$, and $c_v \sim \mathcal{U}[0,1]$.
    \end{itemize}
\end{itemize}

The impact of the choice of the initializations for $\vect\beta$ and the prior hyperparameters for $\theta$ is explored. In ABSLOPE, $\vect\beta_{[0]}$ is initialized using the lasso. Other initializations considered in this section include ridge, group lasso, and elastic net. SGL was also tested but provided no improvement over the lasso while adding substantial computational cost.

In ABSLOPE, the following schemes are considered for the Beta prior hyperparameters: i). $d_1 = 0.01n, d_2 = 0.01n$,
ii). $d_1 = 2p, d_2 = 1- 2/p$, iii). $d_1 = 1, d_2 = p$. For the group-based models, the following schemes were considered: 
\begin{itemize}
    \item \textit{Scheme 1.} BGSLOPE: $d_1 = 0.01n, d_2 = 0.01n$, BSGS: $d_1 = e_1 = 0.003n, d_2 = e_2 = 0.015n$.
    \item \textit{Scheme 2.} BGSLOPE: $d_1 = 2/m, d_2 = 1-2/m$, BSGS: $d_1 = e_1 = 1, d_2 = e_2 = 1$.
    \item \textit{Scheme 3.} BGSLOPE: $d_1 = 1, d_2 = m$, BSGS: $d_1 = e_1 = 1, d_2 = e_2 = p$.
    \item \textit{Scheme 4.} BGSLOPE: $d_1 = m, d_2 = 1$, BSGS: $d_1 = e_1 = p, d_2 = e_2 = 1$.
\end{itemize}

A simulation study compared these initializations and schemes, assessing the stability of BGSLOPE and BSGS on a Gaussian design matrix and a continuous response. The baseline setup from Section \ref{section:sim_study_set_up} was used, varying the correlation ($\rho_w$), signal strength ($s$), and dimensionality ($p$). The results are averaged across these three cases, with each case using $100$ repetitions. BSGS used $\alpha = 0.95$ and other parameters were set as described in Table \ref{tbl:appendix_model_data_simulation}.

Figure \ref{fig:study-1-init-main-text} shows the mean squared error (MSE) of the fitted coefficients $\vect\beta$ to the true values for BSGS and BSGS-$\alpha$ under different $\vect\beta$ initializations and $\theta$ priors. Figure \ref{fig:study-1-init-main-text-sigma} shows the mean absolute error (MAE) of $\sigma$. All initializations of BSGS showed strong robustness, with only minor deviations for elastic net under $\rho = 0.6$. SGSLOBE, BGSLOPE, and GSLOBE (Figures \ref{fig:study-1-init-sgslobe} and \ref{fig:study-1-init-bgslope}) were equally robust. The only unstable cases were for BSGS-$\alpha$ (Figure \ref{fig:study-1-init-main-text}) and SGSLOBE-$\alpha$ (Figure \ref{fig:study-1-init-sgslobe}), indicating that estimating $\alpha$ increases sensitivity to initialization. On average, BSGS-$\alpha$ selected $\alpha = 0.70$ and SGSLOBE-$\alpha$ selected $\alpha = 0.75$ (Figure \ref{fig:alpha-hist}).

\begin{figure}[h]
    \centering
\includegraphics[width=1\linewidth]{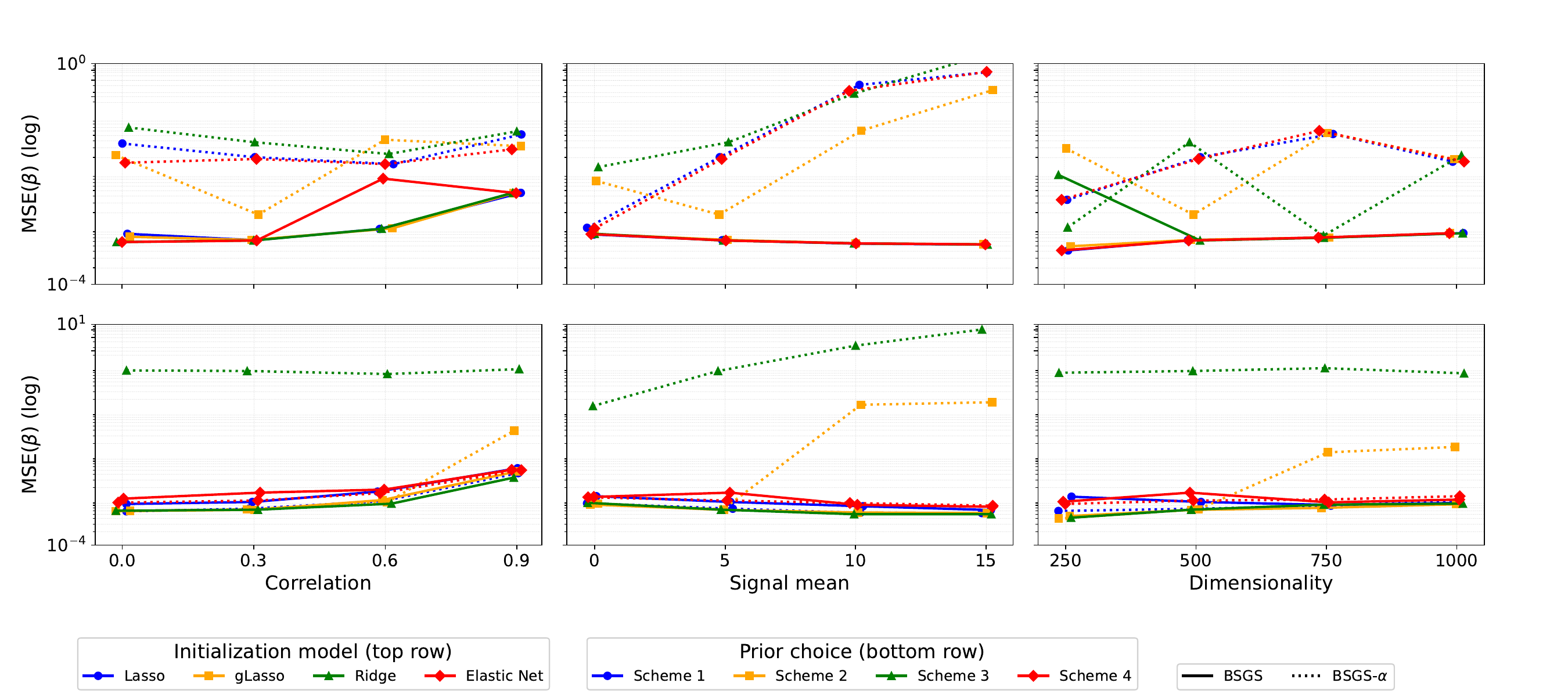}
    \caption[Bias of BSGS under different $\vect\beta$ initialization models]{$\text{MSE}(\vect\beta)$ for BSGS and BSGS-$\alpha$ under different $\vect\beta$ initialization models (top row) and Beta prior choices (bottom row), with a small amount of jitter added to allow the differences to be seen.}
    \label{fig:study-1-init-main-text}
\end{figure}

No initialization consistently outperformed the others (Table \ref{tbl:init-study-beta-lasso}), and while elastic net converged slightly faster on average (Table \ref{tbl:init-study-mean-num-it}), differences were minimal. For the remainder of the manuscript, the lasso initialization for $\vect\beta$ and Scheme 1 for the Beta priors are used. The lasso is preferred for its simplicity and alignment with ABSLOPE. Scheme 1 offers the best balance in parameter bias relative to other schemes (Table \ref{tbl:init-study-beta}). Overall, performance is insensitive to these choices, and we do not expect them to meaningfully impact subsequent analyses.

\begin{remark}[Limitations of Scheme 1]
Under Scheme 1 for BGSLOPE, $\mathbb{E}[\theta] = 0.5$ independently of $n$, imposing no prior preference for sparsity or density. This is reasonable when no domain knowledge is available. However, as $n \to \infty$,
\begin{equation*}                       \lim_{n\rightarrow \infty}\mathbb{E}[\theta \mid \vect{\gamma}, \vect{\beta}, 
  \sigma, \tilde{\vect{A}}] = \lim_{n\rightarrow \infty}\frac{0.01n +           
  \sum_{j=1}^m \mathbbm{1}(\gamma_j = 1)}{0.02n + m} = 0.5,        
  \end{equation*}                       so the posterior mean reverts to the prior, and the data contributes no
  information. The same holds for $\theta_g$ and $\theta_v$ in BSGS. In practice, this is not critical since BGSLOPE targets high-dimensional genetic settings where large $n$ is uncommon, and on the Trust-experts dataset ($n=9759$) (Section \ref{section:real_data}), BGSLOPE outperformed competing methods, suggesting limited
  practical impact.
  \end{remark}

\section{Two-step Orthogonal}\label{section:tso}
We propose a new two-step procedure, \textit{Two-step Orthogonal} (TSO), that transforms a general setting into an orthogonal one, thereby enabling FDR control. Appendix \ref{section:two-step} provides a background of two-step models. The Gram-Schmidt procedure can orthogonalize a matrix but requires low-dimensional data. Assuming that $|S_v| \leq n$, TSO proceeds as follows: 
\begin{enumerate}
    \item Compute $\hat{S}_v$, such that $|\hat{S}_v|\leq n$, using the lasso with CV (picking the 1se model).
    \item Orthogonalize $\vect{X}_{\hat{S}}$ using Gram-Schmidt to generate the orthogonalized design matrix $\tilde{\vect{X}}_{\hat{S}}$.
    \item Fit a SLOPE model using $\tilde{\vect{X}}_{\hat{S}}$ and $\lambda = 1/n$, where the $1/n$ factor is from the loss function.
\end{enumerate} 
For gSLOPE and SGS, the procedure is identical, replacing step 3 with optimizing using the gSLOPE and SGS models instead.

\citet{Bogdan2015SLOPEAdaptiveOptimization} derived penalty sequences for SLOPE under low-dimensional Gaussian designs, termed the \textit{Gaussian sequence} (shown in Figure \ref{fig:bh_gaussian_seq}). The sequence is given by
\begin{equation*}
    v_i^\text{GA} = v_i^\text{BH} \sqrt{1 + f(i-1)\sum_{j<i} (v^\text{GA}_j)^2},\;\text{where} \; f(i) = 1/(n-i-1).
\end{equation*} 
The sequence penalizes small coefficients more heavily to reduce false discoveries at some cost to power, but collapses to the lasso when $p>n$ \citep{Larsson2020a}, limiting its use. Since step 1 of TSO yields a low-dimensional setting, the Gaussian sequence is applied, and step 2 is omitted. All TSO-SLOPE applications in this manuscript use the Gaussian sequence unless otherwise stated (TSO-BH denotes TSO with the BH sequence). Penalty sequences for gSLOPE and SGS remain unchanged.

\begin{wrapfigure}{r}{0.45\linewidth}  
    \centering
\includegraphics[width=\linewidth]{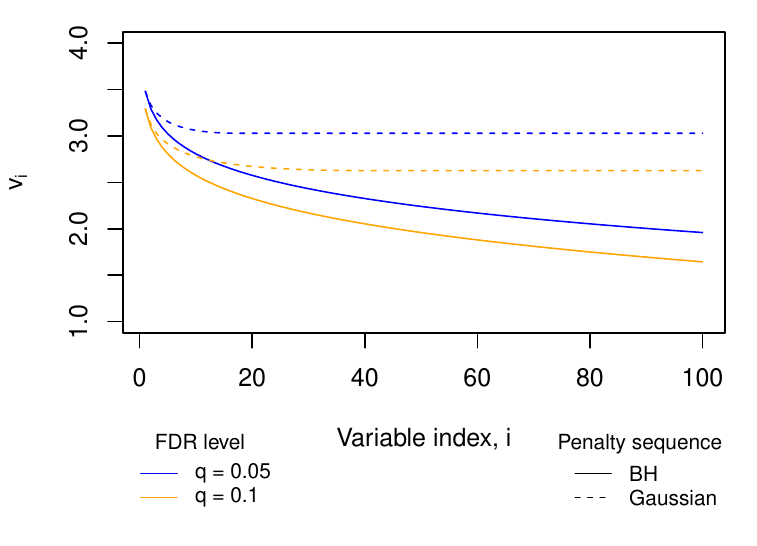}
    \caption[BH and Gaussian sequences for SLOPE]{BH and Gaussian SLOPE sequences for different values of $q$ for $p=100, n= 500$.}
    \label{fig:bh_gaussian_seq}
\end{wrapfigure}

Future work should more carefully consider the step 1 model, as it determines TSO's power. The lasso was chosen for its simplicity and well-understood variable selection properties, and it guarantees $|\hat{S}_v| \leq n$ \citep{Tibshirani1996RegressionLasso}. Ideally, this step would minimize false negatives (type II error), though \citet{10.1214/16-AOS1521} notes this is generally unattainable for the lasso, even if asymptotically it recovers most true signals with few false positives. Replacements to the lasso would ideally be designed for type II error control. This would lead to a model that combines type I and type II error control procedures. For further results on lasso consistency, see \citep{4839045,NEURIPS2020_e7db14e1,JMLR:v7:zhao06a}. Our empirical results show that the lasso is a sensible choice, enabling TSO to control FDR with fast runtime (Section \ref{section:simulation_study}).

\section{Synthetic study}\label{section:simulation_study}
We compare model selection approaches for SLOPE, gSLOPE, and SGS, focusing on FDR control and benchmarking our Bayesian approaches against competitive alternatives (the approaches are summarized in Table \ref{tbl:fp_model_summary}). $10$-fold cross-validation was implemented using strong screening rules \citep{Feser2025StrongModels,Larsson2020a,10.1111/j.2517-6161.1974.tb00994.x}. The oracle model refers to a frequentist model fit with known noise level.

\subsection{Setup}\label{section:sim_study_set_up}
A multivariate Gaussian design matrix $\vect{X}\sim \vect{\mathcal{N}}_p(\vect{0},\vect{\Sigma})\in \mathbb{R}^{400 \times p}$ was generated with block correlation structure, such that within-group correlation was $\Sigma_{i_1,i_2} = \rho_w$, for each pair $i_1,i_2 \in G_j, \forall j\in [m]$, for groups of sizes $[3,25]$. For variables not in the same group, across-group correlation was set to $\Sigma_{i_1,i_2} = \rho_a$ if $i_1 \in G_j$ and $i_2 \in G_k$ such that $j \neq k$.

The response was generated using a linear model $\mathbf{y} = \mathbf{X}\boldsymbol{\beta} + \vect{\mathcal{N}}_n(\vect{0},\sigma^2\mathbf{I}_n)$. A proportion $\xi_g = 0.2$ of groups were set to active, with a proportion $\xi_v$ of variables within each active group set to active. Signals were drawn as $\beta_i \sim \mathcal{N}(s, 10), i \in S_v$. The following data-generating parameters were varied:
\begin{itemize}
    \item (Variable) Sparsity proportion: $\xi_v \in \{0.1, 0.3, 0.5,0.7\}$.
    \item Signal strength: $s\in\{0,5,10,15\}$.
    \item Dimensionality: $p \in \{200,500,750,1000\}$.
    \item Noise: $\sigma \in \{0,1,2,3\}$.
    \item Correlations: $\rho_w \in \{0,0.3,0.6,0.9\}$ and $\rho_a \in \{0,0.1,0.2,0.3\}$.
\end{itemize}
The baseline parameters were set to $\xi_v = 0.3, s = 5,p =500, \sigma = 1, \rho_w = 0.3$, and $\rho_a = 0$ when other parameters were varied (see Table \ref{tbl:appendix_model_data_simulation}). This forms six simulation settings, with each case in a setting repeated $100$ times. The target FDR is set to $0.1$, so BGSLOPE and BSGS used $q_v = q_g = 0.1$. All SGS models used $\alpha = 0.95$ and all Bayesian models used lasso initialization and Scheme 1 (Section \ref{section:sensitivity_analysis}). All other parameters were set as in Table \ref{tbl:appendix_model_data_simulation}.

\subsection{Metrics}
 A test response $\vect{y}_\text{test} = \vect{X}_\text{train}\vect{\beta}_\text{train} + \vect{\mathcal{N}}_n(\vect{0},\sigma^2\vect{I}_n)$ was generated for predictive evaluation. Metrics used are: FDR, power, out-of-sample (OOS) $\ell_2$ prediction error, $\text{F}_1$ score, $\text{MSE}(\vect{\beta})$, and $\text{MAE}(\sigma)$. Methods not estimating the noise are excluded from $\text{MAE}(\sigma)$ plots, and $\text{MSE}(\vect{\beta})$ is not considered critical since penalized regression coefficients are inherently biased and can be debiased via OLS. For readability, only the best-performing models are shown in the main text, with the full results appearing in Appendix \ref{appendix:model_selection_synthetic_study}. Figures \ref{fig:all-f1-plot} and \ref{fig:all-sigma-plot} show $\text{F}_1$ scores and $\text{MAE}(\sigma)$ across all cases, with variable metrics reported for SLOPE and SGS and group metrics for SGS.     

\subsubsection{Impact of signal}
Figure \ref{fig:impact-of-signal-fdr-main} shows that the Bayesian models maintain FDR well below the 0.1 target while achieving the highest power across all signal strengths and sparsity levels. Amongst the non-Bayesian methods, only TSO controls FDR throughout, but with lower power. BGSLOPE achieves near-perfect recovery, and most models are broadly robust to signal changes, with the Bayesian methods particularly so.
\begin{figure}[H]
    \centering
\includegraphics[width=1\linewidth]{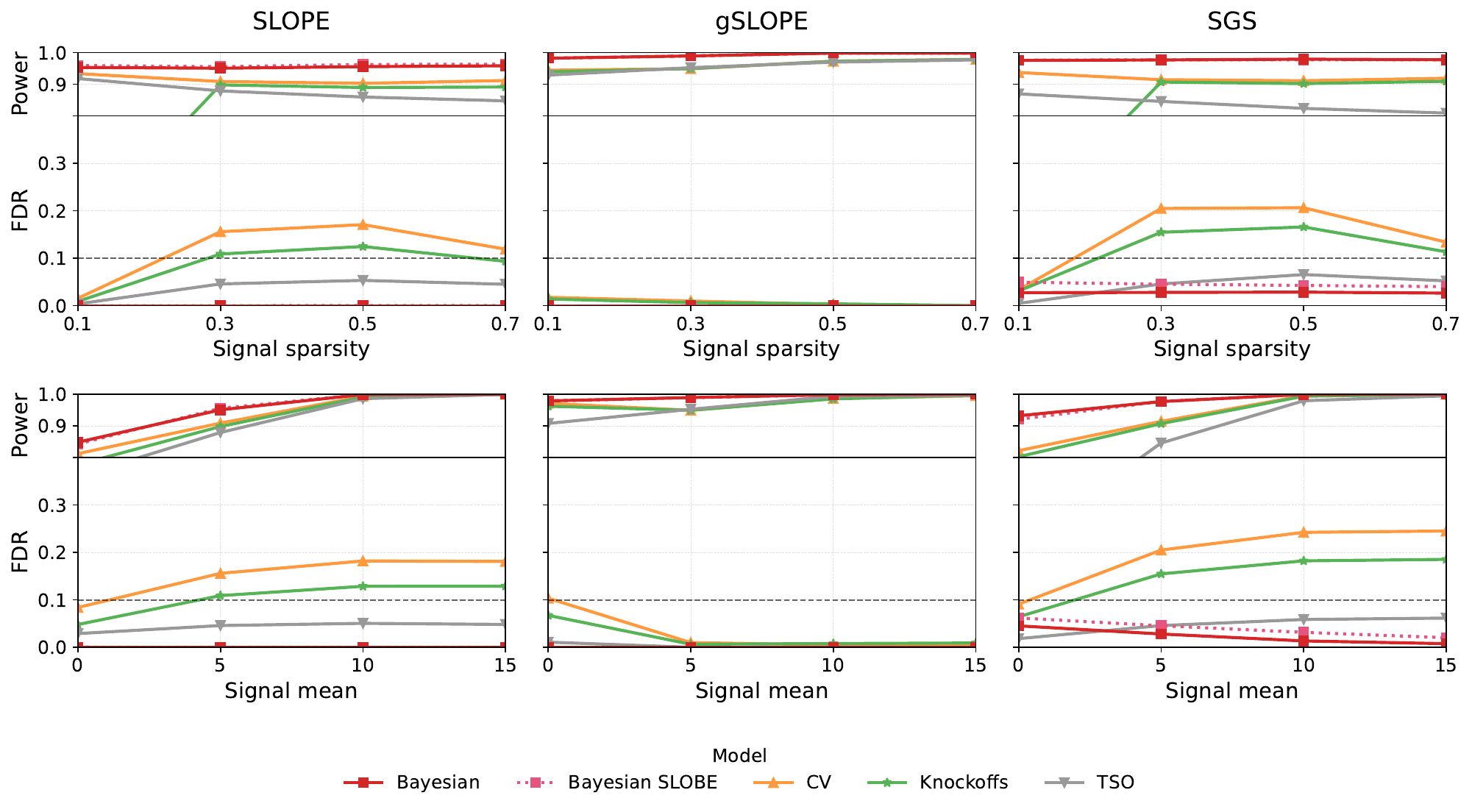}
    \caption{FDR (bottom plots) and power (top plots) for the best performing model selection approaches, as functions of the sparsity proportion (top row) and signal strength (bottom row), split into the type of model (SLOPE, gSLOPE, SGS).}
    \label{fig:impact-of-signal-fdr-main}
\end{figure}

Beyond FDR control, Figure \ref{fig:impact-of-signal-oos} shows that Bayesian models achieve the lowest OOS error under changing signals. They maintain this performance while producing the least biased $\vect{\beta}$ estimates (Figure \ref{fig:impact-of-signal-mse-beta}) and the highest $\text{F}_1$ scores, indicating superior selection (Figure \ref{fig:all-f1-plot}). Figure \ref{fig:impact-of-signal-other-methods} shows that BSGS-$\alpha$ also controls FDR but loses power under highly saturated signals, while all other methods fail to control FDR across all model types.

\begin{figure}[H]
    \centering
\includegraphics[width=1\linewidth]{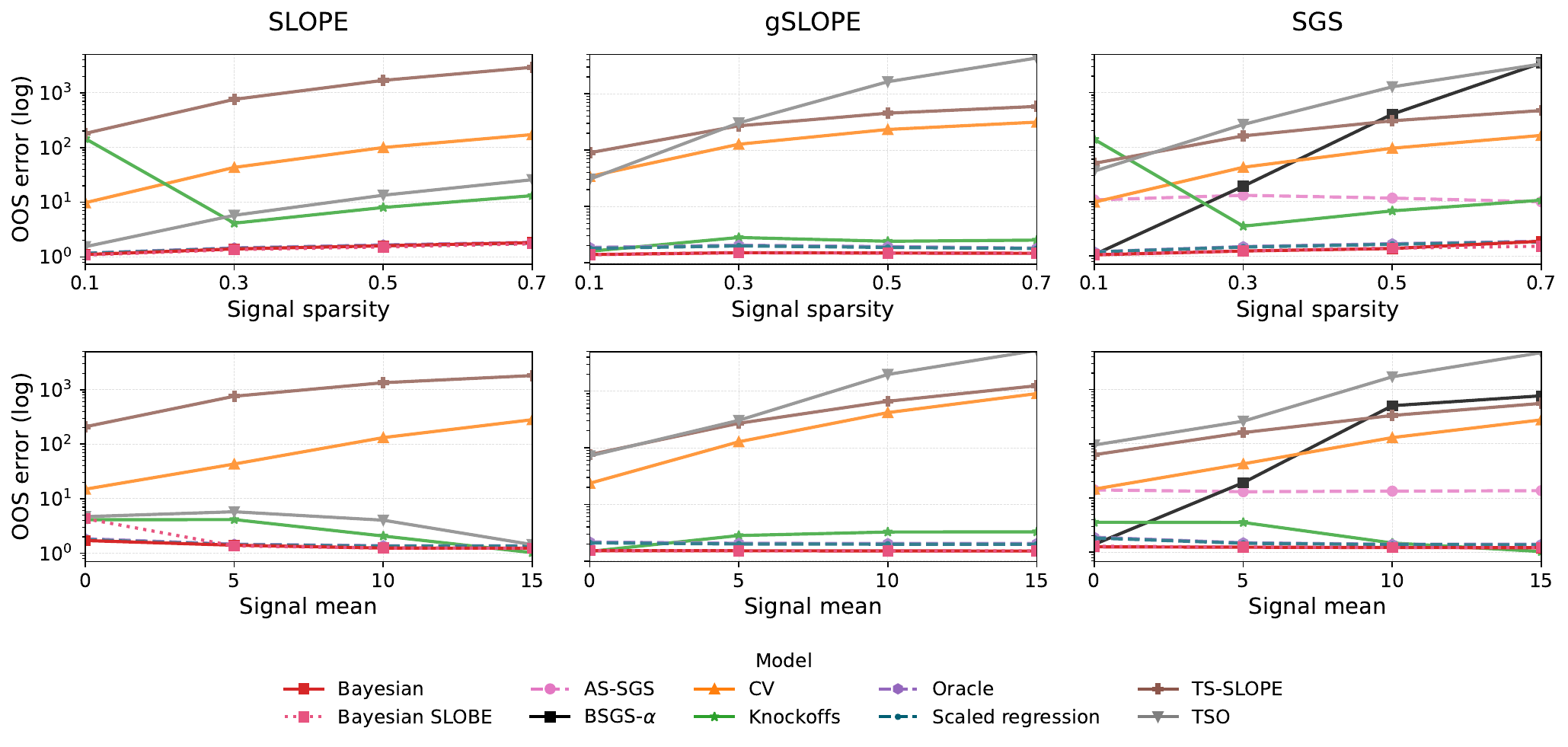}
    \caption{OOS error (log scale) for all model selection approaches, as a function of the sparsity proportion (top row) and signal strength (bottom row), split into the type of model (SLOPE, gSLOPE, SGS).}
    \label{fig:impact-of-signal-oos}
\end{figure}
\begin{figure}[H]
    \centering
\includegraphics[width=1\linewidth]{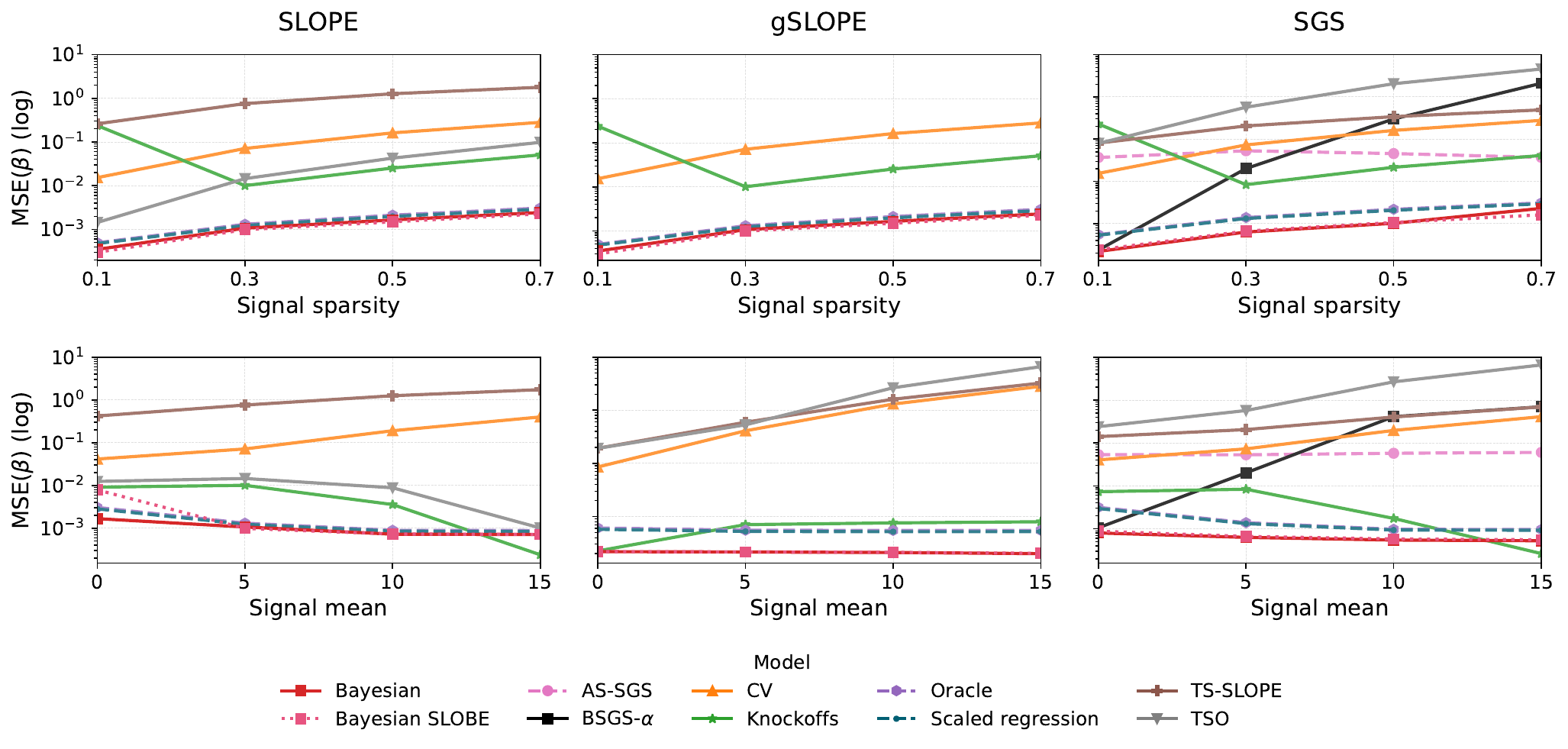}
    \caption{MSE($\boldsymbol\beta$) (log scale) for all model selection approaches, as a function of the sparsity proportion (top row) and signal strength (bottom row), split into the type of model (SLOPE, gSLOPE, SGS).}
    \label{fig:impact-of-signal-mse-beta}
\end{figure}

\subsubsection{Impact of data-generating parameters}
The impact of changing the dimensionality and noise parameters in the data-generating process is assessed in Figure \ref{fig:data-gen-main}. Under increasing dimensionality, the Bayesian methods can clearly control the FDR while maintaining the highest power, showing robustness to dimensionality. TSO is the only other approach that has somewhat satisfactory FDR control, although it is not found to control the FDR above $p = 750$. 

\begin{figure}[H]
    \centering
\includegraphics[width=1\linewidth]{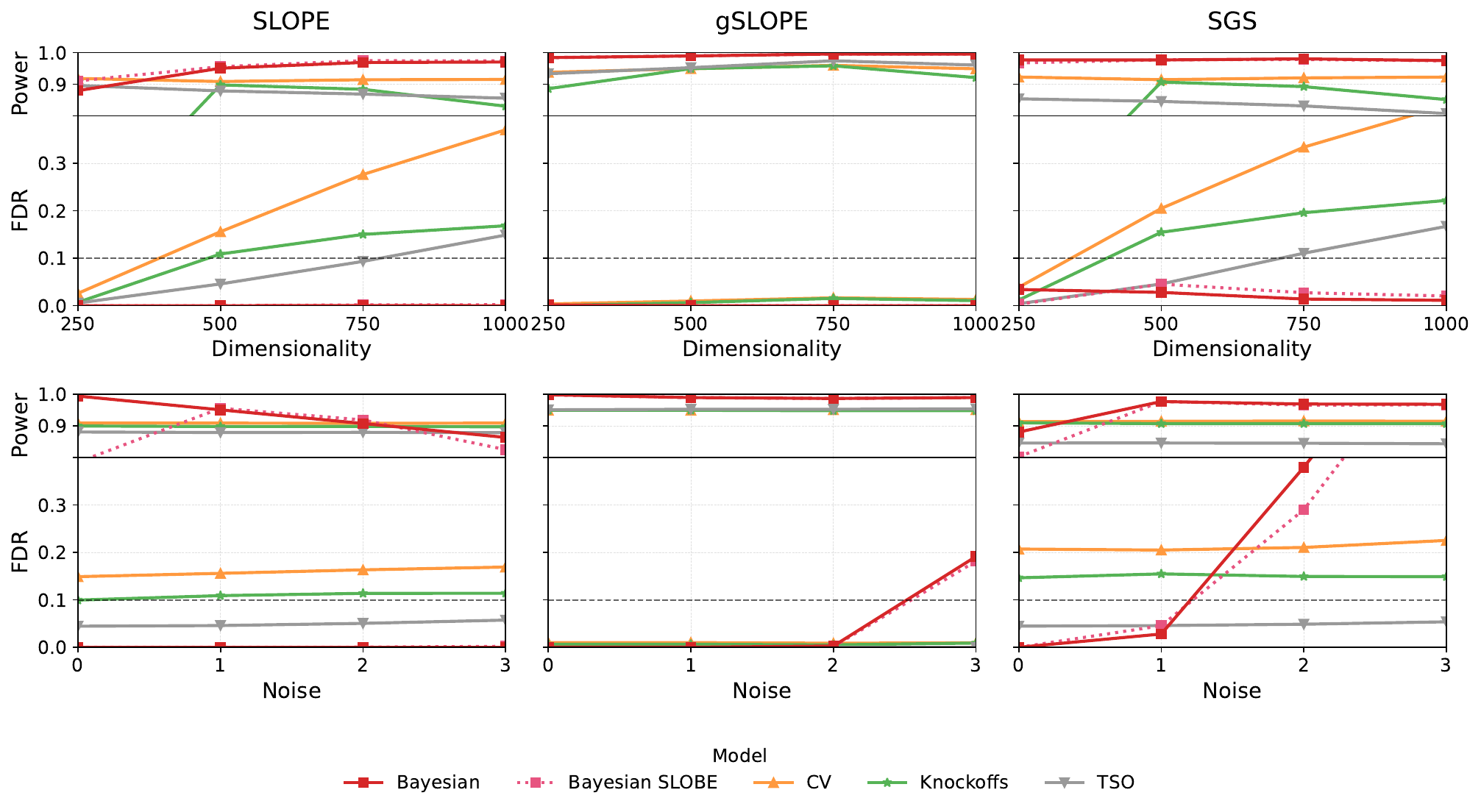}
    \caption{FDR (bottom plots) and power (top plots) for the best performing model selection approaches, as functions of the dimensionality (top row) and noise (bottom row), split into the type of model (SLOPE, gSLOPE, SGS).}
    \label{fig:data-gen-main}
\end{figure}

Under increased noise, BGSLOPE and BSGS show the only instances of failing to control the FDR: BGSLOPE at $\sigma = 3$ and BSGS for $\sigma \geq 2$, though both retain high power. Other methods remain stable, with TSO controlling FDR across all noise levels. While Bayesian methods are more sensitive to noise due to direct estimation, BGSLOPE and BSGS still estimate noise well, suggesting their FDR issues stem from difficulty separating signal from noise rather than poor estimation. In contrast, ABSLOPE fails to scale noise estimates appropriately (Figure \ref{fig:data-gen-mae-sigma}), leading to reduced power but consistent FDR (Figure \ref{fig:data-gen-main}). Thus, as noise increases, ABSLOPE becomes more conservative, whereas BSGS and BGSLOPE become more permissive.

\begin{figure}[H]
    \centering
\includegraphics[width=1\linewidth]{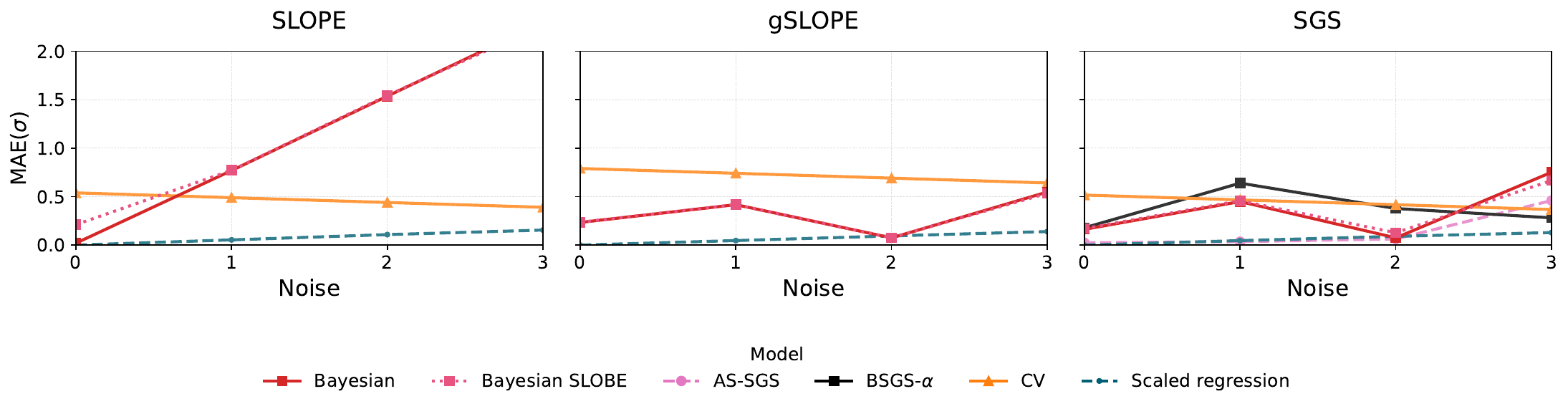}
    \caption{$\text{MAE}(\sigma)$ as a function of the noise for all model selection approaches that estimate the noise, split into the type of model (SLOPE, gSLOPE, SGS).}
    \label{fig:data-gen-mae-sigma}
\end{figure}
The high FDR of Bayesian models under large noise highlights both a challenge and an opportunity: they tend to select overly saturated models. This can be mitigated in BSGS by adjusting the priors on $\theta_g$ and $\theta_v$ to favor sparsity. For instance, setting $\theta_g,\theta_v \sim \text{Beta}(1,1000)$ restores FDR control (Figure \ref{fig:study-1-bsgs-extreme}). This illustrates the flexibility of Bayesian methods, enabling practitioners to tune the power--FDR trade-off by incorporating prior knowledge of a problem setting.

\begin{figure}[H]
    \centering
\includegraphics[width=1\linewidth]{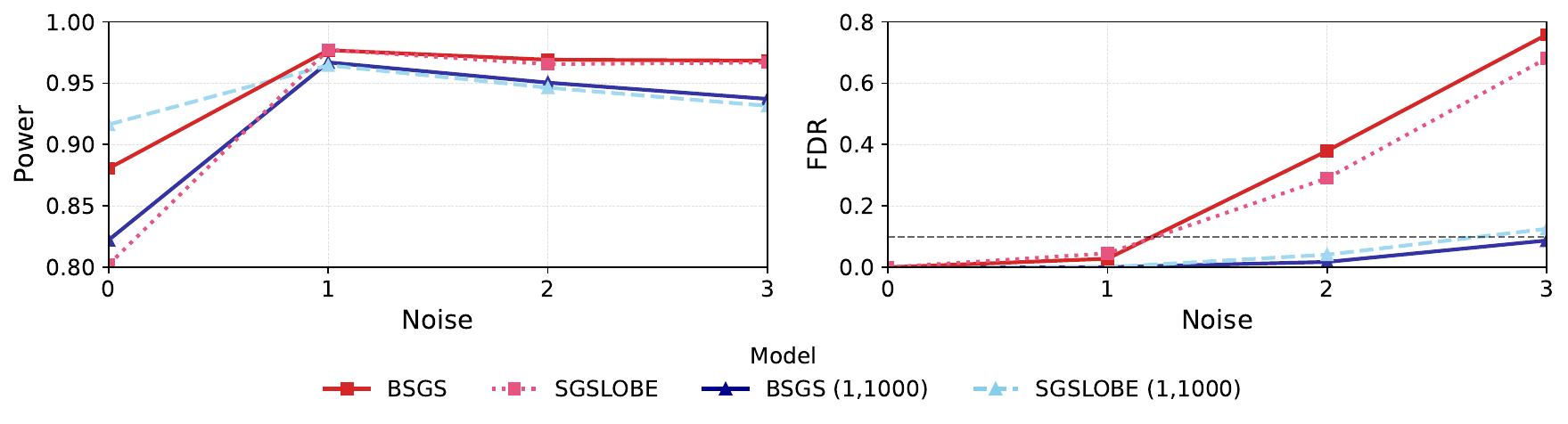}
    \caption{Power and FDR of BSGS under the Scheme 1 priors (BSGS) and $\theta_g,\theta_v \sim \text{Beta}(1,1000)$ priors (BSGS (1,1000)). The comparison is also shown for SGSLOBE.}
    \label{fig:study-1-bsgs-extreme}
\end{figure}
Figure \ref{fig:noise-dim-other} shows that BSGS-$\alpha$ overcomes BSGS’s weakness under high noise, maintaining FDR control, while TS-SLOPE is the only other method that controls FDR in any setting. Figures \ref{fig:noise-dim-mse-beta} and \ref{fig:noise-pred} indicate that Bayesian methods and scaled regression achieve the best $\text{MSE}(\vect\beta)$ and OOS error. Figure \ref{fig:study-1-dim-runtime} shows the runtime increasing with dimensionality for all methods. While Bayesian approaches are computationally expensive, they scale well, whereas TSO exhibits the most favorable scaling overall.

\subsubsection{Correlation}
Our Bayesian methods control the FDR across all levels of across-group correlation while achieving the highest power (Figure \ref{fig:corr-main}). Knockoffs also show promise, controlling FDR for all correlations with SLOPE and for $\rho_a \geq 0.1$ with SGS, while TSO does so for $\rho_a \leq 0.1$, though both fall short in power. For within-group correlation, Bayesian methods again maintain FDR control, except for BSGS at $\rho_w = 0.9$. BGSLOPE is largely unaffected due to its group-wise selection, and ABSLOPE is similarly robust as it ignores group structure.

As in other cases, the FDR level and power are generally very robust across correlations for the Bayesian methods. The other competitive approaches considered here become sensitive and often lose FDR control at higher correlation levels.

Figure \ref{fig:corr-other} shows that BSGS-$\alpha$ maintains FDR control and high power under both correlation settings. Like BSGS, it exhibits increased FDR at $\rho_a = 0.9$, but without violating control. TS-SLOPE and TS-GSLOPE control the FDR, but not TS-SGS. Figures \ref{fig:corr-mse-beta} and \ref{fig:corr-pred} show that Bayesian methods achieve the best $\text{MSE}(\vect\beta)$ and OOS error.

Noise estimation reveals that Bayesian methods have consistent misestimation across correlations (Figure \ref{fig:all-sigma-plot}). In contrast, CV shows a roughly linear increase in error as correlation grows, while the scaled regression methods perform well up to $\rho_{a} \leq 0.1$, after which their errors increase dramatically. However, AS-SGS achieves highly accurate noise estimation across all correlation values.
\begin{figure}[H]
    \centering
\includegraphics[width=1\linewidth]{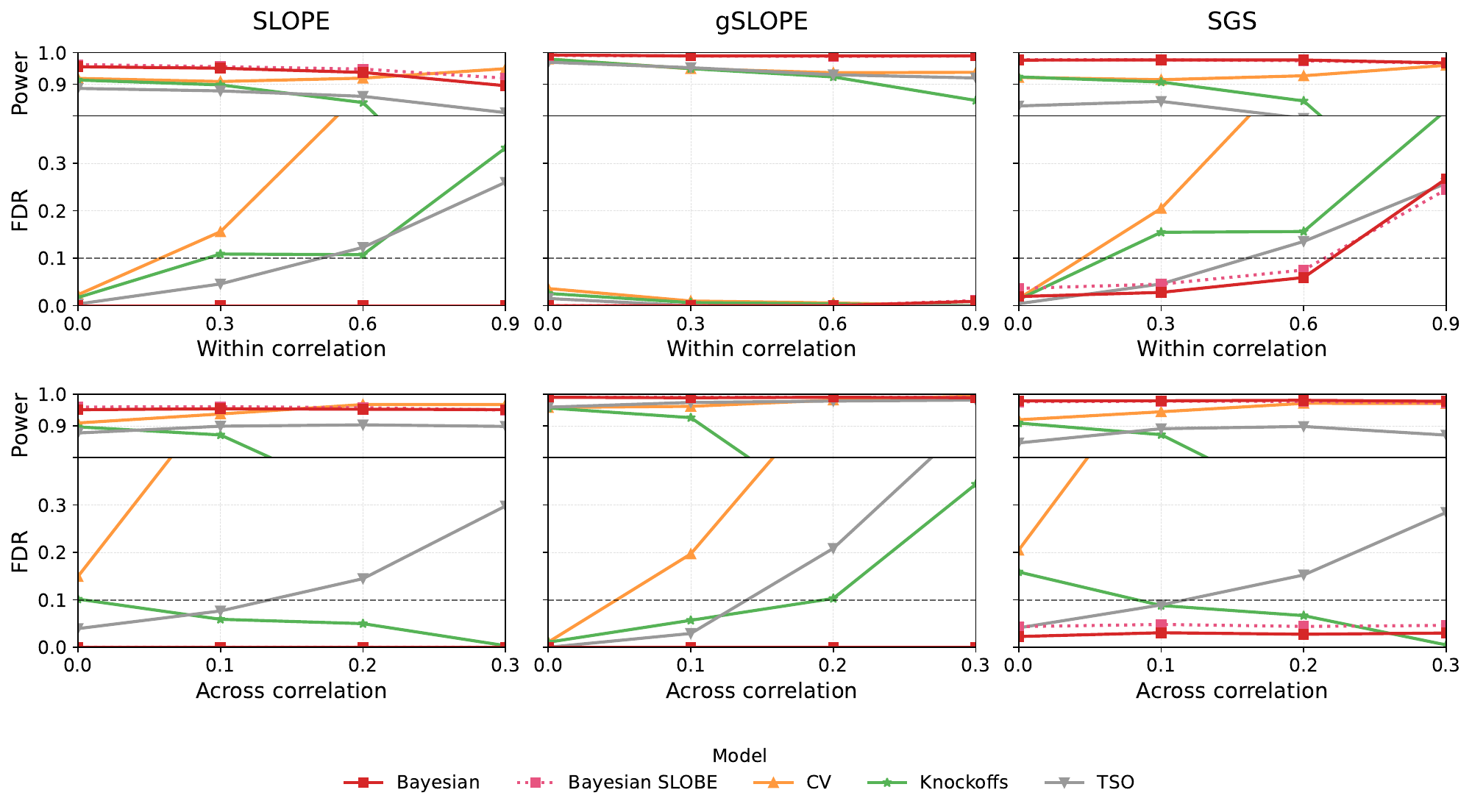}
    \caption{FDR (bottom plots) and power (top plots) for the best performing model selection approaches, as functions of within-group correlation (top row) and across-group correlation (bottom row), split into the type of model (SLOPE, gSLOPE, SGS).}
    \label{fig:corr-main}
\end{figure}
\subsubsection{Summary of results}
Averaged across the six simulation scenarios considered in this section, our group-based Bayesian methods consistently controlled the FDR and achieved the highest power among all methods that maintained FDR control, yielding the strongest $\text{F}_1$ scores (Table \ref{tbl:full-summary-var}). Combined with their strong predictive performance and competitive runtimes, this made them the best overall methods across the SLOPE models.

TSO was the only other method to control FDR across all model types, although this came with a modest reduction in power, likely due to type II errors introduced during the lasso screening step (Section \ref{section:tso}). TSO also exhibited relatively poor predictive performance. In contrast, using the Gaussian sequence for SLOPE improved predictive accuracy while preserving FDR control and increasing power, highlighting the potential value of developing analogous sequences for gSLOPE and SGS.

Although the Bayesian methods generally produced reliable noise estimates, they were not always the most accurate. For SLOPE, CV yielded the most accurate noise estimates on average, while for SGS, AS-SGS estimated the noise level with exceptional precision. However, this accuracy did not translate into strong overall performance, as AS-SGS exhibited extremely high FDR levels.

Among the SGS approaches, only BSGS-$\alpha$ and TSO maintained bi-level FDR control, though both experienced a slight loss of power relative to BSGS. Notably, BSGS-$\alpha$ controlled the FDR in every experiment. Nevertheless, BSGS achieved more accurate estimation of both $\vect{\beta}$ and $\sigma$, leading to superior predictive performance (Table \ref{tbl:full-summary-var}).

With respect to the group-level metrics, BSGS did not control the FDR on average but achieved the highest power (Table \ref{tbl:full-summary-grp}). It provided stronger control of variable-level FDR than group-level FDR, whereas CV exhibited the opposite pattern (Figure \ref{fig:study-1-cv-bsgs-comparison}). This difference arises from BSGS's hierarchical spike-and-slab structure, which selects groups before variables, whereas CV performs selection at both levels simultaneously. Since variable-level control yields more informative associations, it is generally the more desirable property.


\begin{table}[H]
\centering
\resizebox{\textwidth}{!}{%
\begin{tabular}{llccccccc}
\toprule
\textbf{Type} & \textbf{Model} & \textbf{FDR} & \textbf{Power} & \textbf{$\text{F}_1$} & \textbf{MSE($\boldsymbol\beta$)} & \textbf{MAE($\sigma$)} & \textbf{OOS} & \textbf{Time (s)} \\
\midrule
\multirow[c]{11}{*}{SLOPE} 
& ABSLOPE & $\underset{\scriptscriptstyle \textcolor{gray}{(9 \times 10^{-6})}}{\mathbf{9 \times 10^{-6}}}$ & $\underset{\scriptscriptstyle \textcolor{gray}{(4 \times 10^{-3})}}{0.94}$ & $\underset{\scriptscriptstyle \textcolor{gray}{(2 \times 10^{-3})}}{\mathbf{0.97}}$ & $\underset{\scriptscriptstyle \textcolor{gray}{(1 \times 10^{-4})}}{\mathbf{2 \times 10^{-3}}}$ & $\underset{\scriptscriptstyle \textcolor{gray}{(2\times 10^{-3})}}{0.81}$ & $\underset{\scriptscriptstyle \textcolor{gray}{(0.06)}}{\mathbf{1.84}}$ & $\underset{\scriptscriptstyle \textcolor{gray}{(98.38)}}{3339.91}$\\
&SLOBE & $\underset{\scriptscriptstyle \textcolor{gray}{(2 \times 10^{-4})}}{\mathbf{3 \times 10^{-4}}}$ & $\underset{\scriptscriptstyle \textcolor{gray}{(0.01)}}{0.94}$ & $\underset{\scriptscriptstyle \textcolor{gray}{(0.01)}}{0.96}$ & $\underset{\scriptscriptstyle \textcolor{gray}{(0.01)}}{0.02}$ & $\underset{\scriptscriptstyle \textcolor{gray}{(4 \times 10^{-3})}}{0.82}$ & $\underset{\scriptscriptstyle \textcolor{gray}{(10.24)}}{23.66}$ & $\underset{\scriptscriptstyle \textcolor{gray}{(144.47)}}{2700.49}$\\
&CV & $\underset{\scriptscriptstyle \textcolor{gray}{(0.01)}}{0.27}$ & $\underset{\scriptscriptstyle \textcolor{gray}{(0.01)}}{0.92}$ & $\underset{\scriptscriptstyle \textcolor{gray}{(0.01)}}{0.78}$ & $\underset{\scriptscriptstyle \textcolor{gray}{(0.01)}}{0.18}$ & $\underset{\scriptscriptstyle \textcolor{gray}{(0.02)}}{\mathbf{0.70}}$ & $\underset{\scriptscriptstyle \textcolor{gray}{(3.85)}}{81.62}$ & $\underset{\scriptscriptstyle \textcolor{gray}{(0.20)}}{4.77}$\\
&Knockoffs & $\underset{\scriptscriptstyle \textcolor{gray}{(0.01)}}{\mathbf{0.10}}$ & $\underset{\scriptscriptstyle \textcolor{gray}{(0.01)}}{0.79}$ & $\underset{\scriptscriptstyle \textcolor{gray}{(0.01)}}{0.80}$ & $\underset{\scriptscriptstyle \textcolor{gray}{(0.02)}}{0.23}$ & -- & $\underset{\scriptscriptstyle \textcolor{gray}{(46.71)}}{362.54}$ & $\underset{\scriptscriptstyle \textcolor{gray}{(0.89)}}{15.12}$\\
&Oracle & $\underset{\scriptscriptstyle \textcolor{gray}{(0.01)}}{0.54}$ & $\underset{\scriptscriptstyle \textcolor{gray}{(3 \times 10^{-3})}}{0.87}$ & $\underset{\scriptscriptstyle \textcolor{gray}{(0.01)}}{0.58}$ & $\underset{\scriptscriptstyle \textcolor{gray}{(0.01)}}{0.48}$ & -- & $\underset{\scriptscriptstyle \textcolor{gray}{(19.06)}}{629.36}$ & $\underset{\scriptscriptstyle \textcolor{gray}{(3.24)}}{101.17}$\\
&Scaled & $\underset{\scriptscriptstyle \textcolor{gray}{(0.01)}}{0.47}$ & $\underset{\scriptscriptstyle \textcolor{gray}{(2 \times 10^{-3})}}{\mathbf{0.98}}$ & $\underset{\scriptscriptstyle \textcolor{gray}{(0.01)}}{0.67}$ & $\underset{\scriptscriptstyle \textcolor{gray}{(5 \times 10^{-3})}}{0.08}$ & $\underset{\scriptscriptstyle \textcolor{gray}{(0.21)}}{4.30}$ & $\underset{\scriptscriptstyle \textcolor{gray}{(2.55)}}{38.55}$ & $\underset{\scriptscriptstyle \textcolor{gray}{(11.11)}}{194.14}$\\
&TS-SLOPE & $\underset{\scriptscriptstyle \textcolor{gray}{(2 \times 10^{-4})}}{\mathbf{4 \times 10^{-4}}}$ & $\underset{\scriptscriptstyle \textcolor{gray}{(0.01)}}{0.61}$ & $\underset{\scriptscriptstyle \textcolor{gray}{(0.01)}}{0.73}$ & $\underset{\scriptscriptstyle \textcolor{gray}{(0.03)}}{0.96}$ & -- & $\underset{\scriptscriptstyle \textcolor{gray}{(68.33)}}{1510.52}$ & $\underset{\scriptscriptstyle \textcolor{gray}{(1.24)}}{26.06}$\\
&TSO & $\underset{\scriptscriptstyle \textcolor{gray}{(5\times 10^{-3})}}{\mathbf{0.08}}$ & $\underset{\scriptscriptstyle \textcolor{gray}{(0.01)}}{0.88}$ & $\underset{\scriptscriptstyle \textcolor{gray}{(5\times10^{-3})}}{0.90}$ & $\underset{\scriptscriptstyle \textcolor{gray}{(2\times 10^{-2})}}{0.02}$ & -- & $\underset{\scriptscriptstyle \textcolor{gray}{(0.64)}}{7.10}$ & $\underset{\scriptscriptstyle \textcolor{gray}{(0.08)}}{1.56}$\\
&TSO-BH & $\underset{\scriptscriptstyle \textcolor{gray}{(0.01)}}{\mathbf{0.08}}$ & $\underset{\scriptscriptstyle \textcolor{gray}{(0.01)}}{0.83}$ & $\underset{\scriptscriptstyle \textcolor{gray}{(0.01)}}{0.87}$ & $\underset{\scriptscriptstyle \textcolor{gray}{(0.12)}}{1.99}$ & -- & $\underset{\scriptscriptstyle \textcolor{gray}{(137.13)}}{1713.95}$ & $\underset{\scriptscriptstyle \textcolor{gray}{(0.01)}}{\mathbf{0.53}}$\\
\midrule
\multirow[c]{8}{*}{gSLOPE} 
&BGSLOPE & $\underset{\scriptscriptstyle \textcolor{gray}{(1 \times 10^{-3})}}{\mathbf{9\times 10^{-3}}}$ & $\underset{\scriptscriptstyle \textcolor{gray}{(4\times 10^{-3})}}{\mathbf{0.99}}$ & $\underset{\scriptscriptstyle \textcolor{gray}{(3\times 10^{-3})}}{\mathbf{0.99}}$ & $\underset{\scriptscriptstyle \textcolor{gray}{(3\times 10^{-4})}}{\mathbf{5 \times 10^{-3}}}$ & $\underset{\scriptscriptstyle \textcolor{gray}{(0.01)}}{\mathbf{0.43}}$ & $\underset{\scriptscriptstyle \textcolor{gray}{(0.05)}}{\mathbf{2.01}}$ & $\underset{\scriptscriptstyle \textcolor{gray}{(22.99)}}{511.21}$\\
&SLOBE & $\underset{\scriptscriptstyle \textcolor{gray}{(1\times 10^{-3})}}{\mathbf{8 \times 10^{-3}}}$ & $\underset{\scriptscriptstyle \textcolor{gray}{(4\times 10^{-3})}}{\mathbf{0.99}}$ & $\underset{\scriptscriptstyle \textcolor{gray}{(3\times 10^{-3})}}{\mathbf{0.99}}$ & $\underset{\scriptscriptstyle \textcolor{gray}{(3\times 10^{-4})}}{\mathbf{5\times 10^{-3}}}$ & $\underset{\scriptscriptstyle \textcolor{gray}{(0.01)}}{\mathbf{0.43}}$ & $\underset{\scriptscriptstyle \textcolor{gray}{(0.05)}}{\mathbf{2.01}}$ & $\underset{\scriptscriptstyle \textcolor{gray}{(13.56)}}{321.30}$\\
&CV & $\underset{\scriptscriptstyle \textcolor{gray}{(5\times 10^{-3})}}{\mathbf{0.08}}$ & $\underset{\scriptscriptstyle \textcolor{gray}{(0.01)}}{0.96}$ & $\underset{\scriptscriptstyle \textcolor{gray}{(0.01)}}{0.92}$ & $\underset{\scriptscriptstyle \textcolor{gray}{(0.03)}}{0.68}$ & $\underset{\scriptscriptstyle \textcolor{gray}{(0.03)}}{1.12}$ & $\underset{\scriptscriptstyle \textcolor{gray}{(9.51)}}{220.46}$ & $\underset{\scriptscriptstyle \textcolor{gray}{(20.64)}}{750.76}$\\
&Knockoffs & $\underset{\scriptscriptstyle \textcolor{gray}{(0.01)}}{\mathbf{0.03}}$ & $\underset{\scriptscriptstyle \textcolor{gray}{(0.01)}}{0.90}$ & $\underset{\scriptscriptstyle \textcolor{gray}{(0.01)}}{0.92}$ & $\underset{\scriptscriptstyle \textcolor{gray}{(0.02)}}{0.18}$ & -- & $\underset{\scriptscriptstyle \textcolor{gray}{(31.66)}}{305.01}$ & $\underset{\scriptscriptstyle \textcolor{gray}{(52.47)}}{1311.43}$\\
&Oracle & $\underset{\scriptscriptstyle \textcolor{gray}{(0.01)}}{0.58}$ & $\underset{\scriptscriptstyle \textcolor{gray}{(4\times 10^{-3})}}{0.96}$ & $\underset{\scriptscriptstyle \textcolor{gray}{(0.01)}}{0.55}$ & $\underset{\scriptscriptstyle \textcolor{gray}{(0.01)}}{0.49}$ & -- & $\underset{\scriptscriptstyle \textcolor{gray}{(18.87)}}{617.95}$ & $\underset{\scriptscriptstyle \textcolor{gray}{(1.73)}}{33.59}$\\
&Scaled & $\underset{\scriptscriptstyle \textcolor{gray}{(0.01)}}{0.56}$ & $\underset{\scriptscriptstyle \textcolor{gray}{(2\times 10^{-3})}}{\mathbf{0.99}}$ & $\underset{\scriptscriptstyle \textcolor{gray}{(0.01)}}{0.59}$ & $\underset{\scriptscriptstyle \textcolor{gray}{(0.01)}}{0.10}$ & $\underset{\scriptscriptstyle \textcolor{gray}{(0.21)}}{3.45}$ & $\underset{\scriptscriptstyle \textcolor{gray}{(2.12)}}{29.56}$ & $\underset{\scriptscriptstyle \textcolor{gray}{(12.30)}}{77.60}$\\
&TS-GSLOPE & $\underset{\scriptscriptstyle \textcolor{gray}{(1\times 10^{-3})}}{\mathbf{6\times 10^{-3}}}$ & $\underset{\scriptscriptstyle \textcolor{gray}{(0.01)}}{0.89}$ & $\underset{\scriptscriptstyle \textcolor{gray}{(0.01)}}{0.94}$ & $\underset{\scriptscriptstyle \textcolor{gray}{(0.03)}}{0.87}$ & -- & $\underset{\scriptscriptstyle \textcolor{gray}{(27.51)}}{589.73}$ & $\underset{\scriptscriptstyle \textcolor{gray}{(18.33)}}{507.42}$\\
&TSO & $\underset{\scriptscriptstyle \textcolor{gray}{(2\times 10^{-3})}}{\mathbf{0.03}}$ & $\underset{\scriptscriptstyle \textcolor{gray}{(0.01)}}{0.96}$ & $\underset{\scriptscriptstyle \textcolor{gray}{(0.01)}}{0.96}$ & $\underset{\scriptscriptstyle \textcolor{gray}{(0.12)}}{1.95}$ & -- & $\underset{\scriptscriptstyle \textcolor{gray}{(159.04)}}{1968.43}$ & $\underset{\scriptscriptstyle \textcolor{gray}{(0.04)}}{\mathbf{1.04}}$\\
\midrule
\multirow{11}{*}{SGS} 
&AS-SGS & $\underset{\scriptscriptstyle \textcolor{gray}{(0.01)}}{0.70}$ & $\underset{\scriptscriptstyle \textcolor{gray}{(1\times 10^{-3})}}{\mathbf{0.99}}$ & $\underset{\scriptscriptstyle \textcolor{gray}{(0.01)}}{0.42}$ & $\underset{\scriptscriptstyle \textcolor{gray}{(3\times 10^{-3})}}{0.10}$ & $\underset{\scriptscriptstyle \textcolor{gray}{(0.01)}}{\mathbf{0.06}}$ & $\underset{\scriptscriptstyle \textcolor{gray}{(0.40)}}{17.29}$ & $\underset{\scriptscriptstyle \textcolor{gray}{(25.84)}}{206.35}$ \\
&BSGS & $\underset{\scriptscriptstyle \textcolor{gray}{(4\times 10^{-3})}}{\mathbf{0.08}}$ & $\underset{\scriptscriptstyle \textcolor{gray}{(4\times 10^{-3})}}{0.97}$ & $\underset{\scriptscriptstyle \textcolor{gray}{(3\times 10^{-3})}}{0.93}$ & $\underset{\scriptscriptstyle \textcolor{gray}{(1\times 10^{-5})}}{5\times 10^{-3}}$ & $\underset{\scriptscriptstyle \textcolor{gray}{(4\times 10^{-3})}}{0.44}$ & $\underset{\scriptscriptstyle \textcolor{gray}{(0.76)}}{3.96}$ & $\underset{\scriptscriptstyle \textcolor{gray}{(14.76)}}{385.80}$ \\
&SGSLOBE & $\underset{\scriptscriptstyle \textcolor{gray}{(4\times 10^{-3})}}{\mathbf{0.09}}$ & $\underset{\scriptscriptstyle \textcolor{gray}{(3\times 10^{-3})}}{0.97}$ & $\underset{\scriptscriptstyle \textcolor{gray}{(3\times 10^{-3})}}{0.93}$ & $\underset{\scriptscriptstyle \textcolor{gray}{(2\times 10^{-4})}}{\mathbf{3\times 10^{-3}}}$ & $\underset{\scriptscriptstyle \textcolor{gray}{(5\times 10^{-3})}}{0.44}$ & $\underset{\scriptscriptstyle \textcolor{gray}{(0.13)}}{\mathbf{2.40}}$ & $\underset{\scriptscriptstyle \textcolor{gray}{(16.44)}}{428.13}$ \\
&BSGS-$\alpha$ & $\underset{\scriptscriptstyle \textcolor{gray}{(1 \times 10^{-3})}}{\mathbf{8\times 10^{-3}}}$ & $\underset{\scriptscriptstyle \textcolor{gray}{(0.01)}}{0.92}$ & $\underset{\scriptscriptstyle \textcolor{gray}{(0.01)}}{0.94}$ & $\underset{\scriptscriptstyle \textcolor{gray}{(0.05)}}{0.17}$ & $\underset{\scriptscriptstyle \textcolor{gray}{(0.01)}}{0.61}$ & $\underset{\scriptscriptstyle \textcolor{gray}{(68.00)}}{250.39}$ & $\underset{\scriptscriptstyle \textcolor{gray}{(17.82)}}{537.29}$ \\
&SGSLOBE-$\alpha$ & $\underset{\scriptscriptstyle \textcolor{gray}{(2\times 10^{-3})}}{\mathbf{0.02}}$ & $\underset{\scriptscriptstyle \textcolor{gray}{(4\times 10^{-3})}}{0.95}$ & $\underset{\scriptscriptstyle \textcolor{gray}{(3\times 10^{-3})}}{\mathbf{0.96}}$ & $\underset{\scriptscriptstyle \textcolor{gray}{(1 \times 10^{-3})}}{9 \times 10^{-3}}$ & $\underset{\scriptscriptstyle \textcolor{gray}{(0.01)}}{0.58}$ & $\underset{\scriptscriptstyle \textcolor{gray}{(1.25)}}{7.57}$ & $\underset{\scriptscriptstyle \textcolor{gray}{(22.19)}}{688.77}$ \\
&CV & $\underset{\scriptscriptstyle \textcolor{gray}{(0.01)}}{0.30}$ & $\underset{\scriptscriptstyle \textcolor{gray}{(0.01)}}{0.93}$ & $\underset{\scriptscriptstyle \textcolor{gray}{(0.01)}}{0.75}$ & $\underset{\scriptscriptstyle \textcolor{gray}{(0.01)}}{0.18}$ & $\underset{\scriptscriptstyle \textcolor{gray}{(0.02)}}{0.66}$ & $\underset{\scriptscriptstyle \textcolor{gray}{(3.83)}}{80.33}$ & $\underset{\scriptscriptstyle \textcolor{gray}{(37.58)}}{1061.34}$ \\
&Knockoffs & $\underset{\scriptscriptstyle \textcolor{gray}{(0.01)}}{0.14}$ & $\underset{\scriptscriptstyle \textcolor{gray}{(0.01)}}{0.80}$ & $\underset{\scriptscriptstyle \textcolor{gray}{(0.01)}}{0.78}$ & $\underset{\scriptscriptstyle \textcolor{gray}{(0.02)}}{0.23}$ & -- & $\underset{\scriptscriptstyle \textcolor{gray}{(46.12)}}{357.97}$ & $\underset{\scriptscriptstyle \textcolor{gray}{(65.59)}}{1607.69}$ \\
&Oracle & $\underset{\scriptscriptstyle \textcolor{gray}{(0.01)}}{0.51}$ & $\underset{\scriptscriptstyle \textcolor{gray}{(3\times 10^{-3})}}{0.87}$ & $\underset{\scriptscriptstyle \textcolor{gray}{(0.01)}}{0.60}$ & $\underset{\scriptscriptstyle \textcolor{gray}{(0.01)}}{0.49}$ & -- & $\underset{\scriptscriptstyle \textcolor{gray}{(19.07)}}{629.27}$ & $\underset{\scriptscriptstyle \textcolor{gray}{(1.24)}}{26.21}$ \\
&Scaled & $\underset{\scriptscriptstyle \textcolor{gray}{(0.01)}}{0.45}$ & $\underset{\scriptscriptstyle \textcolor{gray}{(2\times 10^{-3})}}{0.98}$ & $\underset{\scriptscriptstyle \textcolor{gray}{(0.01)}}{0.67}$ & $\underset{\scriptscriptstyle \textcolor{gray}{(0.01)}}{0.10}$ & $\underset{\scriptscriptstyle \textcolor{gray}{(0.17)}}{5.55}$ & $\underset{\scriptscriptstyle \textcolor{gray}{(2.79)}}{54.43}$ & $\underset{\scriptscriptstyle \textcolor{gray}{(1.78)}}{46.39}$ \\
&TS-SGS & $\underset{\scriptscriptstyle \textcolor{gray}{(0.01)}}{0.17}$ & $\underset{\scriptscriptstyle \textcolor{gray}{(0.01)}}{0.80}$ & $\underset{\scriptscriptstyle \textcolor{gray}{(0.01)}}{0.84}$ & $\underset{\scriptscriptstyle \textcolor{gray}{(0.01)}}{0.34}$ & -- & $\underset{\scriptscriptstyle \textcolor{gray}{(22.17)}}{425.05}$ & $\underset{\scriptscriptstyle \textcolor{gray}{(29.89)}}{694.72}$ \\
&TSO & $\underset{\scriptscriptstyle \textcolor{gray}{(0.01)}}{\mathbf{0.08}}$ & $\underset{\scriptscriptstyle \textcolor{gray}{(0.01)}}{0.83}$ & $\underset{\scriptscriptstyle \textcolor{gray}{(0.01)}}{0.87}$ & $\underset{\scriptscriptstyle \textcolor{gray}{(0.12)}}{1.97}$ & -- & $\underset{\scriptscriptstyle \textcolor{gray}{(134.68)}}{1671.26}$ & $\underset{\scriptscriptstyle \textcolor{gray}{(0.04)}}{\mathbf{0.98}}$ \\
\bottomrule
\end{tabular}
}
\caption{Key metrics averaged across the six simulation cases considered, split into model type (SLOPE, gSLOPE, SGS), shown with standard errors. The metrics for SLOPE and SGS correspond to variable metrics, while those for gSLOPE correspond to group metrics. The best performing model for each metric within each model type is highlighted in \textbf{bold} (aside from FDR, for which any that have $\text{FDR} \leq 0.1$ are in bold).}
\label{tbl:full-summary-var}
\end{table}

\subsubsection{Additional insights}

\paragraph{Impact of $\alpha$ on SGS models.}
Figure \ref{fig:study-2-case-2} shows that BSGS is not sensitive to the choice of $\alpha$, controlling FDR for all $\alpha \in [0.05, 0.95]$, while maintaining high power; a property shared only by TSO. BSGS demonstrates the best selection performance across all values of $\alpha$ and exhibits low bias. BSGS’s adaptivity allows it to perform well across $\alpha$, unlike CV. TSO’s consistent performance stems from its initial lasso step, which filters irrelevant variables and mitigates the limitations of group-only fitting at small $\alpha$.

Scaled regression closely matches the oracle due to highly accurate noise estimation ($\text{MAE}(\sigma) = 0.03$). Overall, for most models, FDR decreases as $\alpha \to 0.95$, reflecting the shift away from group-only behavior, where variable-level FDR control is harder to achieve.


\begin{figure}[H]
    \centering
\includegraphics[width=1\linewidth]{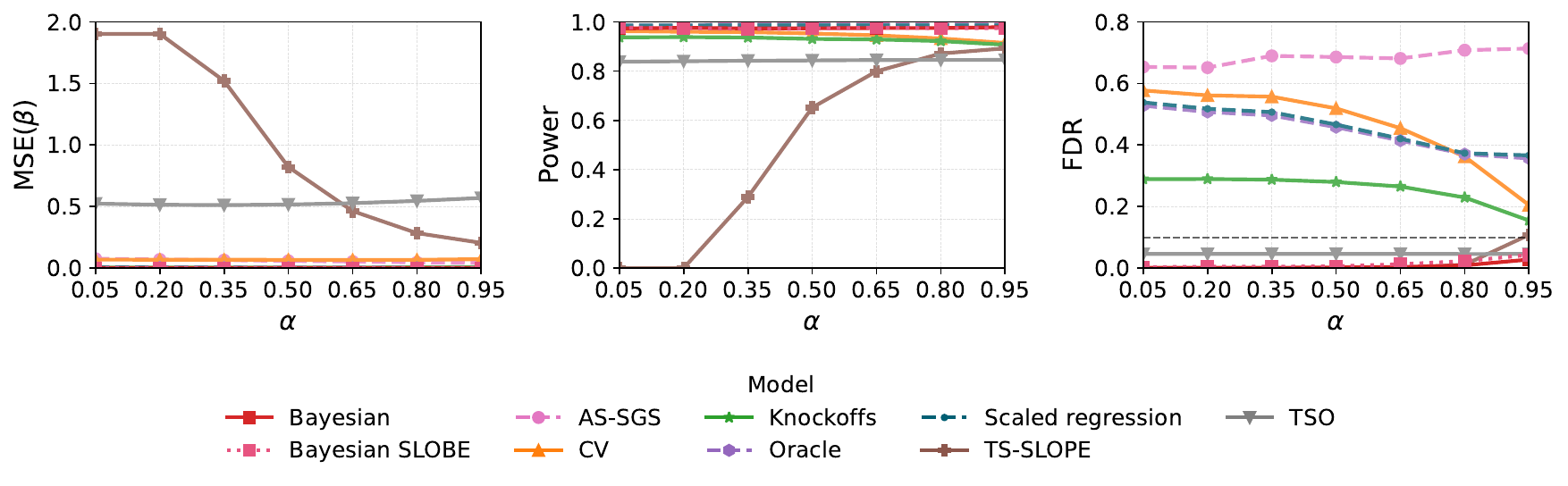}
    \caption{MSE($\boldsymbol\beta$), variable Power, and variable FDR for all SGS model selection approaches as a function of $\alpha$ under the baseline setting (Table \ref{tbl:appendix_model_data_simulation}).}
    \label{fig:study-2-case-2}
\end{figure}

\paragraph{Equal group size.} 
The approaches were also applied to an equal-sized grouping structure, varying the size of each group $p_j \in \{5,10,15,20\}$ (Figure \ref{fig:study-2-case-1}). Under all group sizes, the Bayesian models (including BSGS-$\alpha$) and TSO are found to control the FDR, while the Bayesian models achieve higher power. No other methods were able to control the FDR for all model types.
\section{Real data study}\label{section:real_data}
In this section, we apply the model selection approaches that showed the most promise in the synthetic study to real data: the Bayesian methods, CV, and TSO. 

\subsection{Setup}
\paragraph{Datasets.}
The models are applied to seven real-world datasets: Cancer, Colitis, BRCA1, Carbotax, Rhee, Scheetz, and Trust-experts. The datasets have varied characteristics to test the models under different scenarios; five of the datasets are high-dimensional, four different grouping structure approaches are taken, and two of the datasets are semi-synthetic. All datasets have continuous responses, so linear models are fitted. Further information on the datasets is provided in Appendix \ref{appendix:real_data_info}.

The Cancer and Colitis datasets originally have binary responses. To enable evaluation of the FDR on real data and allow us to use this dataset with a linear model, a synthetic response was generated for each using the linear model: $\vect{y}_\text{semi-syn} = \vect{X}_\text{real}\vect{\beta}_\text{semi-syn} + \vect{\mathcal{N}}_n(\vect{0}, \vect{I}_n)$, with group sparsity set to $0.2$, variable sparsity set to $0.15$, and signal strength set to $5$.

Each dataset was divided into five 80/20 train/test partitions, and the final results were averaged over these splits. To assess predictive performance, we use the Normalized MSE (NMSE) of the prediction, defined by 
\begin{equation*}
    \text{NMSE} = \text{MSE}(\vect{y}_\text{test},\vect{X}_\text{test}\hat{\vect{\beta}}_\text{train})/\text{Var}(\vect{y}_\text{test}).
\end{equation*} 
This yields a dimensionless metric ranging from 0 to 1, where $\text{NMSE} = 1$ corresponds to predicting the mean of $\vect{y}$, while values below $1$ indicate improved predictive performance.

\paragraph{Model hyperparameters.}
The models were applied in the same way as for the synthetic study, with the following minor changes. For the SGS models, two modifications were made. First, $\alpha = 0.99$ is used, following \citet{Feser2023Sparse-groupFDR-control}, who observed that this value yields stronger predictive performance than the default $\alpha = 0.95$ on real data. Second, anticipating weaker signals in the real data, the priors on the mixture proportions were adjusted to $\theta_g,\theta_v \sim \text{Beta}(5,1)$, allowing BSGS to select denser models.

\subsection{Results}
Table \ref{tbl:real-data-pred-scores} reports the NMSE scores across all datasets. For SLOPE and SGS, the Bayesian approach attains the lowest NMSE on all datasets except the semi-synthetic Cancer dataset, where TSO performs best. For gSLOPE, BGSLOPE yields the best performance on all datasets except Colitis and BRCA1, where CV performs best. Notably, TSO performs very poorly on Colitis for gSLOPE and SGS, but not for SLOPE, mirroring the synthetic study where combining TSO with the Gaussian sequence boosts predictive performance.

\begin{table}[H]
\centering
\resizebox{\textwidth}{!}{%
\begin{tabular}{lrrrrrrrrr}
    \toprule
        \textbf{Dataset}       &  \multicolumn{2}{r}{\textbf{SLOPE}}  & \multicolumn{3}{r}{\textbf{gSLOPE}} & \multicolumn{3}{r}{\textbf{SGS}} \\
    \cmidrule(lr){2-4}  \cmidrule(lr){5-7} 
    \cmidrule(lr){8-10}
&Bayesian & CV&TSO & Bayesian & CV&TSO & Bayesian & CV&TSO   \\
    \midrule
  Cancer & $\underset{\scriptscriptstyle \textcolor{gray}{(0.111)}}{0.803}$ & $\underset{\scriptscriptstyle \textcolor{gray}{(0.094)}}{0.479}$ & $\underset{\scriptscriptstyle \textcolor{gray}{(0.053)}}{\underline{\vect{0.235}}}$ & $\underset{\scriptscriptstyle \textcolor{gray}{(0.031)}}{\vect{0.258}}$ & $\underset{\scriptscriptstyle \textcolor{gray}{(0.088)}}{0.390}$ & $\underset{\scriptscriptstyle \textcolor{gray}{(0.121)}}{0.361}$ & $\underset{\scriptscriptstyle \textcolor{gray}{(0.039)}}{0.291}$  & $\underset{\scriptscriptstyle \textcolor{gray}{(0.095)}}{0.437}$ & $\underset{\scriptscriptstyle \textcolor{gray}{(0.057)}}{\vect{0.287}}$\\ 
  Colitis & $\underset{\scriptscriptstyle \textcolor{gray}{(0.052)}}{\mathbf{0.419}}$ & $\underset{\scriptscriptstyle \textcolor{gray}{(0.131)}}{0.891}$ & $\underset{\scriptscriptstyle \textcolor{gray}{(0.139)}}{0.898}$ & $\underset{\scriptscriptstyle \textcolor{gray}{(0.105)}}{0.759}$ & $\underset{\scriptscriptstyle \textcolor{gray}{(0.095)}}{\mathbf{0.700}}$ & $\underset{\scriptscriptstyle \textcolor{gray}{(151.254)}}{458.813}$ & $\underset{\scriptscriptstyle \textcolor{gray}{(0.052)}}{\underline{\mathbf{0.401}}}$ & $\underset{\scriptscriptstyle \textcolor{gray}{(0.105)}}{0.640}$ & $\underset{\scriptscriptstyle \textcolor{gray}{(450.085)}}{801.935}$ \\ \arrayrulecolor{gray!40}\midrule\arrayrulecolor{black}
  BRCA1 & $\underset{\scriptscriptstyle \textcolor{gray}{(0.043)}}{\underline{\vect{0.445}}}$ & $\underset{\scriptscriptstyle \textcolor{gray}{(0.062)}}{0.486}$ & $\underset{\scriptscriptstyle \textcolor{gray}{(0.048)}}{0.555}$ & $\underset{\scriptscriptstyle \textcolor{gray}{(0.094)}}{0.884}$ & $\underset{\scriptscriptstyle \textcolor{gray}{(0.067)}}{\vect{0.766}}$ & $\underset{\scriptscriptstyle \textcolor{gray}{(0.056)}}{1.028}$ & $\underset{\scriptscriptstyle \textcolor{gray}{(0.043)}}{\vect{0.467}}$ & $\underset{\scriptscriptstyle \textcolor{gray}{(0.060)}}{0.498}$ & $\underset{\scriptscriptstyle \textcolor{gray}{(0.056)}}{1.028}$ \\ 
 Carbotax & $\underset{\scriptscriptstyle \textcolor{gray}{(0.175)}}{\vect{0.838}}$ & $\underset{\scriptscriptstyle \textcolor{gray}{(0.209)}}{0.923}$ & $\underset{\scriptscriptstyle \textcolor{gray}{(0.223)}}{1.000}$ & $\underset{\scriptscriptstyle \textcolor{gray}{(0.156)}}{\underline{\vect{0.742}}}$ & $\underset{\scriptscriptstyle \textcolor{gray}{(0.209)}}{0.961}$ & $\underset{\scriptscriptstyle \textcolor{gray}{(0.223)}}{1.001}$ & $\underset{\scriptscriptstyle \textcolor{gray}{(0.195)}}{\vect{0.878}}$ & $\underset{\scriptscriptstyle \textcolor{gray}{(0.216)}}{0.955}$ & $\underset{\scriptscriptstyle \textcolor{gray}{(0.225)}}{0.997}$ \\ 
Rhee & $\underset{\scriptscriptstyle \textcolor{gray}{(0.007)}}{\underline{\vect{0.138}}}$ & $\underset{\scriptscriptstyle \textcolor{gray}{(0.022)}}{0.371}$ & $\underset{\scriptscriptstyle \textcolor{gray}{(0.022)}}{0.450}$ & $\underset{\scriptscriptstyle \textcolor{gray}{(0.008)}}{\vect{0.140}}$ & $\underset{\scriptscriptstyle \textcolor{gray}{(0.020)}}{0.348}$ & $\underset{\scriptscriptstyle \textcolor{gray}{(0.037)}}{0.970}$ & $\underset{\scriptscriptstyle \textcolor{gray}{(0.007)}}{\underline{\vect{0.138}}}$ & $\underset{\scriptscriptstyle \textcolor{gray}{(0.023)}}{0.328}$ & $\underset{\scriptscriptstyle \textcolor{gray}{(0.037)}}{0.970}$ \\ 
  Scheetz & $\underset{\scriptscriptstyle \textcolor{gray}{(0.102)}}{\underline{\vect{0.442}}}$ & $\underset{\scriptscriptstyle \textcolor{gray}{(0.107)}}{0.478}$ & $\underset{\scriptscriptstyle \textcolor{gray}{(0.388)}}{1.026}$ & $\underset{\scriptscriptstyle \textcolor{gray}{(0.112)}}{\vect{0.479}}$ & $\underset{\scriptscriptstyle \textcolor{gray}{(0.229)}}{0.694}$ & $\underset{\scriptscriptstyle \textcolor{gray}{(0.388)}}{1.026}$ & $\underset{\scriptscriptstyle \textcolor{gray}{(0.236)}}{\vect{0.715}}$ & $\underset{\scriptscriptstyle \textcolor{gray}{(0.270)}}{0.744}$ & $\underset{\scriptscriptstyle \textcolor{gray}{(0.388)}}{1.026}$ \\ 
  Trust-experts & $\underset{\scriptscriptstyle \textcolor{gray}{(0.009)}}{\underline{\vect{0.345}}}$ & $\underset{\scriptscriptstyle \textcolor{gray}{(0.009)}}{0.420}$ & $\underset{\scriptscriptstyle \textcolor{gray}{(0.009)}}{0.399}$ & $\underset{\scriptscriptstyle \textcolor{gray}{(0.009)}}{\underline{\vect{0.345}}}$ & $\underset{\scriptscriptstyle \textcolor{gray}{(0.009)}}{0.352}$ & $\underset{\scriptscriptstyle \textcolor{gray}{(0.015)}}{0.847}$ & $\underset{\scriptscriptstyle \textcolor{gray}{(0.009)}}{\underline{\vect{0.345}}}$ & $\underset{\scriptscriptstyle \textcolor{gray}{(0.009)}}{0.352}$ & $\underset{\scriptscriptstyle \textcolor{gray}{(0.015)}}{0.935}$ \\ 
    \bottomrule
  \end{tabular}
  }
  \caption[NMSE for approaches on the real data]{NMSE for the Bayesian methods, CV, and TSO for each model type (SLOPE, gSLOPE, SGS), with standard errors shown in \textcolor{gray}{grey}. The best performing method for each data within the model type is highlighted in \textbf{bold}. The best method overall for a dataset is \underline{underlined}.}
\label{tbl:real-data-pred-scores}
\end{table}
Overall, the results show that Bayesian methods consistently provide superior predictive performance on real data, echoing the synthetic experiments and confirming Bayesian model selection as the most effective approach for prediction. Among SLOPE variants, SLOPE performed best on four datasets, SGS on three, and gSLOPE on two, though comparing them was not the primary goal.

Table \ref{tbl:real-data-pred-scores-variant} reports the NMSE for Bayesian models versus SLOBE variants. For SLOPE and gSLOPE, the simplified SLOBE sampling substantially reduces predictive performance on most datasets. In contrast, SGSLOBE often outperforms SGS, indicating that reduced variability benefits the complex BSGS sampling regime. In fact, SGSLOBE is the top-performing method for five datasets.

Table \ref{tbl:real-data-pred-scores-bsgs} presents the NMSE for several additional BSGS variants, to assess which hyperparameter configurations yield the best predictive performance. BSGS with priors $\theta_g,\theta_v \sim \text{Beta}(5,1)$ outperforms the default prior variant, reflecting the benefit of accounting for higher noise in real data. Models with $\alpha = 0.99$ consistently outperform $\alpha = 0.95$, supporting $\alpha = 0.99$ as the default for real data. BSGS-$\alpha$ with the same $\text{Beta}(5,1)$ priors performs comparably to BSGS at $\alpha = 0.99$.

\section{Discussion}
SLOPE models are known to control the FDR under orthogonal designs. In this work, we examined model selection for SLOPE under more general settings, aiming to restore FDR control in practice through empirical evaluation of new and existing approaches.

Our main methodological contribution is the development of group-based Bayesian SLOPE models, BGSLOPE and BSGS, which use spike-and-slab priors equivalent to gSLOPE and SGS, respectively. The Bayesian framework allows the models to learn the noise alongside the regression parameters, while providing uncertainty quantification, feature importance, and adaptivity to the underlying sparsity structure. Notably, BSGS introduces the first continuous spike-and-slab framework for sparse-group models, extending these benefits to both variable and group levels. We also propose the Two-step Orthogonal (TSO) approach, which transforms a general setting into an orthogonal one, allowing SLOPE's FDR properties to be used.

This manuscript presents a comprehensive comparison of model selection approaches for SLOPE models using synthetic and real data. In synthetic experiments, the Bayesian methods consistently outperform alternatives, achieving strong FDR control, predictive accuracy, and parameter recovery, while retaining the highest power among methods that control FDR. TSO performs well as a fast alternative, particularly for SLOPE, where the Gaussian sequence leads to strong predictive performance. In real data, the Bayesian models again deliver the best predictive performance, surpassing competing methods on most datasets.

Our Bayesian methodology eliminates the need to tune $\lambda$ by learning it via the noise, but introduces additional hyperparameters compared to the frequentist counterparts. For BSGS, we need to set $(\alpha, q_v, q_g, d_1, d_2, e_1, e_2)$ and a $\vect{\beta}$ initialization, though $\alpha$ can be updated using BSGS-$\alpha$. Section \ref{section:sensitivity_analysis} shows that the models are generally robust to prior hyperparameters, while still allowing the incorporation of domain knowledge to improve performance. In practice, tuning hyperparameters to encourage sparsity or density can meaningfully improve FDR control and predictive accuracy, illustrating the flexibility of the Bayesian framework in adapting to different noise levels and signal structures.


We conclude that under general settings, SLOPE models are best applied within the Bayesian framework developed in this manuscript. This recommendation is supported by their strong empirical performance on synthetic and real data: the Bayesian methods effectively control FDR, maintain high power, produce less biased estimates, and deliver accurate predictions.

\bibliography{ref}
\bibliographystyle{plainnat}    

\appendix
\section{SLOPE models}
\subsection{Weight choices}\label{appendix:slope_weights}
\subsubsection{Group SLOPE (gSLOPE)}
The maximum criterion in the gSLOPE weights (Equation \ref{eqn:gslope_pen_max}) can be relaxed to formulate the gSLOPE mean sequence \citep{Brzyski2019GroupPredictors} 
 \begin{equation}\label{eqn:gslope_pen_mean}
	w_j^\text{mean} = \overline{F}^{-1}_{\chi_{p_j}} \left(1-\frac{q_g \cdot j}{m}\right), \; \text{where} \; \overline{F}_{\chi_{p_j}}(x):= \frac{1}{m}\sum_{k=1}^{m}F_{\chi_{p_k}}(\sqrt{p_k}x).
\end{equation}
\subsubsection{Sparse-group SLOPE (SGS)}
Bi-level FDR control was proven for SGS in \citet{Feser2023Sparse-groupFDR-control} under orthogonal $\mathbf{X}$ using the following sequences
\begin{align*}
	&v_i^\text{max} = \max_{k=1,\dots,m} \left\{\frac{1}{\alpha} F_\mathcal{N}^{-1} \left(1-\frac{q_v \cdot i}{2p}\right) -   \frac{1}{3\alpha}(1-\alpha) a_k w_k\right\}, \; i=1,\dots,p, \\
	&w_j^\text{max} =\max_{k=1,\dots,m}\left\{\frac{F_\text{FN}^{-1}(1-\frac{q_g \cdot j}{m})-\alpha \sum_{l \in G_k}v_l }{(1-\alpha) p_k}\right\}, \; j=1,\dots,m,
\end{align*}
where  $F_{\text{FN}}$ is the CDF of a folded Gaussian distribution and $F_{\chi_{p_j}}$ is the CDF of a $\chi$ distribution with $p_j$ degrees of freedom. The term $a_j$ is generally unknown, with the estimator $\hat{a}_j = \lfloor\alpha p_j\rfloor$ used in practice \citep{Feser2023Sparse-groupFDR-control}. As with gSLOPE, it is also feasible to relax the sequences to form the variable and group SGS mean sequences
\begin{align}
	&v_i^\text{mean} = \overline{F}_\mathcal{N}^{-1}\left(1-\frac{q_v \cdot i}{2p}\right), \; \text{where}\; \overline{F}_\mathcal{N}(x) := \frac{1}{m}\sum_{k=1}^{m} F_\mathcal{N}\left(\alpha x +  \frac{1}{3}(1-\alpha) a_k w_k\right),\label{eqn:sgs_var_pen_mean} \\
	&w_j^\text{mean} = \overline{F}_\text{FN}^{-1}\left(1-\frac{q_g \cdot j}{m}\right), \; \text{where}\; \overline{F}_\text{FN}(x) := \frac{1}{m}\sum_{k=1}^{m} F_\text{FN}\left((1-\alpha) p_k x + \alpha \sum_{l \in G_k} v_l\right).\label{eqn:sgs_grp_pen_mean}
\end{align}
\section{Group-based Bayesian SLOPE models}
\subsection{Bayesian gSLOPE (BGSLOPE)}
\subsubsection{Theory}\label{appendix:bgslope_theory}
\begin{proof}[Proof of Proposition \ref{propn:gslope_equiv}]
Assume we have a variable $\vect{z}=(z_1,\ldots,z_p)$ with the gSLOPE prior
\begin{equation}\label{eqn:bayesian_gslope_prior_2}
    \pi(\vect\beta \mid \sigma^2, \vect{w}) \propto \prod_{j=1}^m \exp\left\{-\frac{1}{\sigma} w_{r_g(\vect\beta,j)}\sqrt{p_j}\|\vect{\beta}^{(j)}\|_2 \right\}.
\end{equation}
Now, define $\vect\beta = \vect{\tilde{A}}^{-1}\vect{z}$ such that $z_i = \beta_i \tilde{a}_j$ for $i\in G_j$. By the transformation of variables, we have that the prior for $\vect\beta$ is given by
\begin{equation*}
    \pi\left(\vect\beta \mid \vect\gamma, c, \sigma^2 ; \vect{w} \right) \propto \left|\text{det}(d\vect{z}/d \vect\beta)\right| \pi\left(\vect{z} \mid \vect\gamma, c, \sigma^2 ; \vect{w} \right).
\end{equation*}
Now, we have that $(d\vect{z}/d\vect\beta)_i = \tilde{a}_j$, and as $\vect{\tilde{A}}$ is diagonal, 
\begin{equation*}
    \text{det}(d\vect{z}/d\vect\beta) = \text{det}(\tilde{\vect{A}}) = \prod_{j=1}^m \tilde{a}^{p_j}_j=c^{\sum_{j=1}^m p_j \mathbbm{1}(\gamma_j = 1)}.
\end{equation*}
The final equality follows from the definition of $\vect{\tilde{A}}$. Therefore,
\begin{align*}
        \pi\left(\vect\beta \mid \vect\gamma, c, \sigma^2 ; \vect{w} \right)& \propto \left|\text{det}\left(\frac{d\vect{z}}{d\vect\beta}\right)\right| \pi\left(\vect{z} \mid \vect\gamma, c, \sigma^2 ; \vect{w} \right)\\
        & = c^{\sum_{j=1}^m p_j \mathbbm{1}(\gamma_j = 1)}\prod_{j=1}^m \exp\left\{-\frac{1}{\sigma} w_{r_g(\tilde{\vect{A}}\vect\beta,j)}\sqrt{p_j}\|\vect{z}^{(j)}\|_2 \right\}\\
              & = c^{\sum_{j=1}^m p_j \mathbbm{1}(\gamma_j = 1)}\prod_{j=1}^m \exp\left\{-\frac{1}{\sigma}\tilde{a}_{j} w_{r_g(\tilde{\vect{A}}\vect\beta,j)}\sqrt{p_j}\|\vect{\beta}^{(j)}\|_2 \right\},
\end{align*}
which corresponds to the BGSLOPE prior. The final step follows as the entries in $\vect{\tilde{A}}$ are identical within each group, so that it can be treated as a scaling term.
\end{proof}

\begin{theorem}\label{thm:gslope_generalising_constant}
The gSLOPE normalizing constant is given by
\begin{align*}
C^{-1} = m! \prod_{g=1}^{m} &\frac{2\pi^{p_g/2}}{\Gamma(p_g/2)} \\
&\times\sum_{\vect{k} \in \mathbb{K}} \prod_{j=1}^{m-1} \frac{(p_j + k_{j-1} - 1)! ( \sigma^{-1} \sqrt{p_{j}} \sum_{l=1}^{j}w_l )^{k_j}(p_m + k_{m-1} - 1)!}{k_j! ( \sigma^{-1} \sqrt{p_{j}} \sum_{l=1}^{j}w_l  )^{p_j + k_{j-1}}(\sigma^{-1} \sqrt{p_{m}} \sum_{l=1}^{m}w_l  )^{p_m + k_{m-1}}},
\end{align*}
where $\mathbb{K} = \{ \vect{k} \in \mathbb{Z}^{m-1} : k_1 \in [0, p_1 - 1], \; k_2 \in [0, p_2 + k_1 - 1], \; \ldots, \; k_{m-1} \in [0, p_{m-1} + k_{m-2} - 1] \},$ with $k_0 = 0$.
\end{theorem}
\begin{proof}
The proof strategy is to transform the problem into spherical coordinates, expressing the integrals via Gamma functions, and thus exploiting their exponential representation.

The normalizing constant is given by the integral \citep{Sepehri2016TheSLOPE}
\begin{equation*}
I = \int e^{- \sigma^{-1} \sum_{g=1}^{m} p_g w_g \|\vect{\beta}^{(g)}\|_2} d\vect\beta.
\end{equation*}
We restrict ourselves to the setting $\sqrt{p_1}\|\vect{\beta}^{(1)}\|_2 \geq \dotsb \geq \sqrt{p_m}\|\vect{\beta}^{(m)}\|_2$. There are $m!$ possible permutations of this ordering, so that
\begin{equation*}
I = m!\int_{\sqrt{p_1}\|\vect{\beta}^{(1)}\|_2 \geq \dotsb \geq \sqrt{p_m}\|\vect{\beta}^{(m)}\|_2} e^{- \sigma^{-1} \sum_{g=1}^{m} p_g w_g \|\vect{\beta}^{(g)}\|_2} d\vect{\beta}^{(1)}\dotsb d\vect{\beta}^{(m)}.
\end{equation*}
We can rewrite the exponential term as 
\begin{equation*}
e^{-\sigma^{-1} w_m \sqrt{p_m}\|\vect\beta^{(m)}\|_2}e^{-\sigma^{-1} \sum_{g=1}^{m-1} w_g \sqrt{p_g}\|\vect{\beta}^{(g)}\|_2},
\end{equation*}
as $\sqrt{p_1}\|\vect\beta^{(1)}\|_2 \geq \dotsb \geq \sqrt{p_m}\|\vect\beta^{(m)}\|_2$ and by the fact that the penalty sequence is decreasing. Therefore, the integral becomes
\begin{align*}
I = m!\int_{\sqrt{p_1}\|\vect\beta^{(1)}\|_2 \geq \dotsb \geq \sqrt{p_m}\|\vect\beta^{(m)}\|_2} &e^{-\sigma^{-1} w_m \sqrt{p_m}\|\vect\beta^{(m)}\|_2}\\
&\times e^{-\sigma^{-1} \sum_{g=1}^{m-1} w_g \sqrt{p_g}\|\vect{\beta}^{(g)}\|_2} d\vect\beta^{(1)}\dotsb d\vect\beta^{(m)}.
\end{align*}
Continuing the pattern, we get nested integrals
\begin{align*}
I = m!\int_{\mathcal{R}^{p_m}|_{m-1}} e^{-\sigma^{-1} w_m \sqrt{p_m}\|\vect\beta^{(m)}\|_2} &\int_{\mathcal{R}^{p_{m-1}}|_{m-2}} e^{-\sigma^{-1} w_m \sqrt{p_m}\|\vect\beta^{(m)}\|_2} \\
& \dotsb \int_{\mathcal{R}^{p_2}|_{1}} e^{-\sigma^{-1} w_1 \sqrt{p_1}\|\vect\beta^{(1)}\|_2}d\vect\beta^{(1)}\dotsb d\vect\beta^{(m)},
\end{align*}
where $\mathcal{R}^{p_g}|_{g-1}$ is the space of all possible vectors of $\vect{\beta}^{(g)}$ restricted to $\sqrt{p_{g-1}}\|\vect\beta^{(g-1)}\|_2 \geq \sqrt{p_g}\|\vect{\beta}^{(g)}\|_2$.

To find this, spherical coordinates are applied. Using the change of variables $r = \|\vect{\beta}^{(g)}\|_2$ allows each vector to be represented as $\vect{\beta}^{(g)} = r\vect{u}$, where $r$ is the radius (magnitude) of the vector and $\vect{u}\in\mathbb{R}^{m_g}$ is a unit vector (giving direction). These two components allow any vector to be expressed in spherical form. Using these coordinates, the volume element transforms as (decomposed into radial and surface area components)
\begin{equation*}
d\vect{\beta}^{(g)} = r_g^{p_g - 1}dr_g d\Omega_g. 
\end{equation*}
Therefore, the first nested integral (rightmost one), with the restriction that $r_1\sqrt{p_1} \geq r_2 \sqrt{p_2} \implies r_1 \geq r_2\sqrt{p_2}/\sqrt{p_1}$, is given by
\begin{equation*}
I_1 =\int_{r_2\sqrt{p_2}/\sqrt{p_1}}^\infty r_1^{p_1 - 1} e^{-\sigma^{-1} w_1 \sqrt{p_1}r_1}dr_1\int_{\mathbb{S}^{p_1 - 1}}\Omega_1,
\end{equation*}
where $\mathbb{S}$ is the space of surface areas to integrate over. These can be dealt with separately as they do not rely on $r$ (and therefore any sorting). We have
\begin{equation*}
I_1^S= \int_{\mathbb{S}^{p_1 - 1}} d\Omega_1 = \frac{2 \pi^{p_1 / 2}}{\Gamma(p_1 / 2)},
\end{equation*}
where $I^\mathbb{S}_1$ refers to the first surface area integral and $\Gamma$ denotes a complete gamma function.  Considering all groups, we have that the directional part of the integral will contribute (by their independence from each other)
\begin{equation*}
   I^S = \int_{\mathbb{S}^{p_m - 1}} \dotsb \int_{\mathbb{S}^{p_2 - 1}} \int_{\mathbb{S}^{p_1 - 1}} d\Omega_1 d\Omega_2\dotsb d\Omega_m = \prod_{g=1}^m \frac{2\pi^{p_g/2}}{\Gamma(p_g/2)}.
\end{equation*}
As these can be dealt with separately, the problem is that of calculating $C^{-1} = m!I^r_m I^S$, where $I^r_m$ is the calculation of all $m$ nested radial integrals; it is a function of all previous integrals $I_m^r = f(I^r_1,\ldots,I^r_{m-1})$.

Now, the radial part of the integral is in the form of a (scaled) upper incomplete gamma function \citep{abramowitz1972handbook}
\begin{equation*}
 \int_{x}^\infty t^{a- 1} e^{-bt}dt = \frac{\tilde\Gamma\left(a,bx\right)}{b^a},
\end{equation*}
where $\tilde\Gamma$ denotes an upper incomplete gamma function. Applied to the first radial integral gives, denoted $I_1^r$,
\begin{align*}
   I^r_1 = \int_{r_2\sqrt{p_2}/\sqrt{p_1}}^\infty r_1^{p_1 - 1} e^{-\sigma^{-1} w_1 \sqrt{p_1}r_1}dr_1 &= \frac{\tilde\Gamma\left(p_1,\frac{w_{1}\sqrt{p_{1}}}{\sigma}\cdot \frac{r_{2}\sqrt{p_2}}{\sqrt{p_1}}\right)}{\left(\sigma^{-1}w_{1}\sqrt{p_{1}}\right)^{p_{1}}}\\ 
   &= \frac{\tilde\Gamma\left(p_1,w_1r_2\sqrt{p_2}\sigma^{-1}\right)}{\left(\sigma^{-1}w_{1}\sqrt{p_{1}}\right)^{p_{1}}}.
\end{align*}
Using the following incomplete gamma function relation \citep{abramowitz1972handbook}
\begin{equation*}
    \tilde\Gamma(n,x) =(n-1)!e^{-x}\sum_{k=0}^{n-1}\frac{x^k}{k!},
\end{equation*} 
we have
\begin{align*}
     I_1^r &= \int_{r_2\sqrt{p_2}/\sqrt{p_1}}^\infty r_1^{p_1 - 1} e^{-\sigma^{-1}w_1 \sqrt{p_1}r_1 }dr_1=\frac{\tilde\Gamma\left(p_1,w_1r_2\sqrt{p_2}\sigma^{-1}\right)}{\left(\sigma^{-1}w_{1}\sqrt{p_{1}}\right)^{p_{1}}}\\ 
     &=\frac{(p_1-1)!e^{-\sigma^{-1}w_1r_2 \sqrt{p_2}}\sum_{k_1=0}^{p_1-1}(w_1\sqrt{p_2}\sigma^{-1})^{k_1} r_2^{k_1}/k_1!}{\left(\sigma^{-1}w_{1}\sqrt{p_{1}}\right)^{p_{1}}}.
\end{align*}
Plugging this term into the second integrand and ignoring the directional terms
\begin{align*}
    I_2^r &= \int_{r_3\sqrt{p_3}/\sqrt{p_2}}^\infty r_2^{p_2 - 1} e^{-\sigma^{-1} w_2 \sqrt{p_2}r_2} I_1^r dr_2,\\
    & = \frac{(p_1-1)!}{\left(\sigma^{-1}w_{1}\sqrt{p_{1}}\right)^{p_{1}}}\int_{r_3\sqrt{p_3}/\sqrt{p_2}}^\infty r_2^{p_2 - 1} e^{-\sigma^{-1} w_2 \sqrt{p_2}r_2}e^{-\sigma^{-1}w_1r_2 \sqrt{p_2}}\\
&\quad\quad\quad\quad\quad\quad\quad\quad\quad\quad\quad\quad\quad\quad\times\sum_{k_1=0}^{p_1-1}(w_1\sqrt{p_2}\sigma^{-1})^{k_1} r_2^{k_1}/k_1! dr_2\\
    & = \sum_{k_1=0}^{p_1-1}\frac{(p_1-1)!(w_1\sqrt{p_2}\sigma^{-1})^{k_1}}{k_1!\left(\sigma^{-1}w_{1}\sqrt{p_{1}}\right)^{p_{1}}}\int_{r_3\sqrt{p_3}/\sqrt{p_2}}^\infty r_2^{k_1} r_2^{p_2 - 1} e^{-\sigma^{-1} w_2 \sqrt{p_2}r_2}e^{-\sigma^{-1}w_1r_2 \sqrt{p_2}} dr_2\\
      & = \sum_{k_1=0}^{p_1-1}\frac{(p_1-1)!(w_1\sqrt{p_2}\sigma^{-1})^{k_1}}{k_1!\left(\sigma^{-1}w_{1}\sqrt{p_{1}}\right)^{p_{1}}}\int_{r_3\sqrt{p_3}/\sqrt{p_2}}^\infty r_2^{p_2 + k_1 - 1} e^{-\sigma^{-1} \sqrt{p_2}r_2 (w_1 + w_2)} dr_2
\end{align*}
where we have swapped the integral and summation by Fubini's theorem. Again, by the incomplete gamma relation, the integral term becomes
\begin{equation*}
\int_{r_3\sqrt{p_3}/\sqrt{p_2}}^\infty r_2^{p_2 + k_1 - 1} e^{-\sigma^{-1} \sqrt{p_2}r_2 (w_1 + w_2)} dr_2 = \frac{\Gamma\left(p_2+k_1,(w_1+w_2)r_3\sqrt{p_3}\sigma^{-1}\right)}{\left(\sigma^{-1}(w_1+w_2)\sqrt{p_{2}}\right)^{p_{2}+k_1}}.
\end{equation*}
Therefore,
\begin{equation*}
    I^r_2 =  \sum_{k_1=0}^{p_1-1}\frac{(p_1-1)!(w_1\sqrt{p_2}\sigma^{-1})^{k_1}}{k_1!\left(\sigma^{-1}w_{1}\sqrt{p_{1}}\right)^{p_{1}}} \times \frac{\tilde\Gamma\left(p_2+k_1,(w_1+w_2)r_3\sqrt{p_3}\sigma^{-1}\right)}{\left(\sigma^{-1}(w_1+w_2)\sqrt{p_{2}}\right)^{p_{2}+k_1}}.
\end{equation*}
Each time we evaluate an integrand, a summation term is generated. Therefore, generalizing this up until the $m-1$ integrand, we have
\begin{align*}
 I_{m-1}^r =    \sum_{k_1=0}^{p_1-1}\dotsb &\sum_{k_{m-2}=0}^{p_{m-2}+k_{m-3}-1}\frac{(p_1-1)!(\sigma^{-1}\sqrt{p_2}w_1)^{k_1}}{k_1!\left(\sigma^{-1}\sqrt{p_{1}}w_{1}\right)^{p_{1}}}\\
& \times \dotsb \times \frac{(p_{m-2}+k_{m-3}-1)!(\sigma^{-1}\sqrt{p_{m-1}}\sum_{l=1}^{m-2}w_l)^{k_{m-2}}}{k_{m-2}!\left(\sigma^{-1}\sqrt{p_{m-2}}\sum_{l=1}^{m-2}w_l\right)^{p_{m-2}+k_{m-3}}}  \\
     &\times \frac{\tilde\Gamma\left(p_{m-1}+k_{m-2},\sigma^{-1}\sqrt{p_{m}}r_{m}\sum_{l=1}^{m-1}w_l\right)}{\left(\sigma^{-1}\sqrt{p_{m-1}}\sum_{l=1}^{m-1}w_l\right)^{p_{m-1}+k_{m-2}}}.
\end{align*}
The $m$th and final integrand is given by (ignoring the constant terms before the integral)
\begin{align*}
    I_m^r &= \int_{0}^\infty r_m^{p_m - 1} e^{-\sigma^{-1} \sqrt{p_m}r_m  w_m} I_{m-1}^r dr_m\\
    &\propto \int_{0}^\infty r_m^{p_m - 1} e^{-\sigma^{-1} \sqrt{p_m}r_m  w_m} \tilde\Gamma\left(p_{m-1}+k_{m-2},\sigma^{-1}\sqrt{p_{m}}r_{m}\sum_{l=1}^{m-1}w_l\right) dr_m,\\
    &=\sum_{k_{m-1}=0}^{p_{m-1}+k_{m-2}-1}\frac{(p_{m-1}+k_{m-2}-1)!}{k_{m-1}!} \left(\sigma^{-1}\sqrt{p_{m}}\sum_{l=1}^{m-1}w_l\right)^{k_{m-1}} \\
& \;\;\;\;\;\;\;\;\;\;\;\;\;\;\;\;\;\;\;\;\;\;\;\;\;\;\;\;\;   \times\int_{0}^\infty r_m^{p_m - 1} e^{-\sigma^{-1} \sqrt{p_m}r_m w_m} e^{-\sigma^{-1} \sqrt{p_m}r_m\sum_{l=1}^{m-1}w_l} r_m^{k_{m-1}} dr_m \\
     &=\sum_{k_{m-1}=0}^{p_{m-1}+k_{m-2}-1}\frac{(p_{m-1}+k_{m-2}-1)!}{k_{m-1}!} \left(\sigma^{-1}\sqrt{p_{m}}\sum_{l=1}^{m-1}w_l\right)^{k_{m-1}} \\  &\;\;\;\;\;\;\;\;\;\;\;\;\;\;\;\;\;\;\;\;\;\;\;\;\;\;\;\;\; \times\int_{0}^\infty r_m^{p_m + k_{m-1}- 1} e^{-\sigma^{-1} \sqrt{p_m}r_m\sum_{l=1}^{m}w_l} dr_m\\
     & =\sum_{k_{m-1}=0}^{p_{m-1}+k_{m-2}-1}\frac{(p_{m-1}+k_{m-2}-1)!}{k_{m-1}!} \left(\sigma^{-1}\sqrt{p_{m}}\sum_{l=1}^{m-1}w_l\right)^{k_{m-1}}\\
&\;\;\;\;\;\;\;\;\;\;\;\;\;\;\;\;\;\;\;\;\;\;\;\;\;\;\;\;\; \times \frac{\Gamma(p_{m}+k_{m-1})}{\left(\sigma^{-1}\sqrt{p_{m}}\sum_{l=1}^{m}w_l\right)^{p_{m}+k_{m-1}}}\\ 
     &=\sum_{k_{m-1}=0}^{p_{m-1}+k_{m-2}-1}\frac{(p_{m-1}+k_{m-2}-1)!}{k_{m-1}!} \left(\sigma^{-1}\sqrt{p_{m}}\sum_{l=1}^{m-1}w_l\right)^{k_{m-1}}\\
&\;\;\;\;\;\;\;\;\;\;\;\;\;\;\;\;\;\;\;\;\;\;\;\;\;\;\;\;\; \times  \frac{(p_{m}+k_{m-1} -1)!}{\left(\sigma^{-1}\sqrt{p_{m}}\sum_{l=1}^{m}w_l\right)^{p_{m}+k_{m-1}}},
\end{align*}
where the complete gamma function satisfies $\int_0^\infty r^{a-1}e^{-br}dr = \Gamma(a)/b^a$, justifying the second to last step, and $\Gamma(n) = (n-1)!$, if $n$ is a positive integer, which is the case for $p_m + k_{m-1}$, justifying the last step.

Putting this all together, including the directional integrand, $I^S$, gives
\begin{align*}
    C^{-1} =m! &\prod_{g=1}^m \frac{2\pi^{p_g/2}}{\Gamma(p_g/2)}\\
    &\times\sum_{k_1=0}^{p_1-1}\dotsb \sum_{k_{m-2}=0}^{p_{m-2}+k_{m-3}-1}\sum_{k_{m-1}=0}^{p_{m-1}+k_{m-2}-1}\Biggl[ \frac{(p_1-1)!(\sigma^{-1}\sqrt{p_2}w_1)^{k_1}}{k_1!(\sigma^{-1}\sqrt{p_{1}}w_{1})^{p_{1}}}\times \dotsb \\
    &\times \frac{(p_{m-1}+k_{m-2}-1)! (\sigma^{-1}\sqrt{p_{m}}\sum_{l=1}^{m-1}w_l)^{k_{m-1}}}{k_{m-1}!(\sigma^{-1}\sqrt{p_{m-1}}\sum_{l=1}^{m-1}w_l)^{p_{m-1}+k_{m-2}}}\\
    &\times \frac{(p_{m}+k_{m-1} -1)!}{(\sigma^{-1}\sqrt{p_{m}}\sum_{l=1}^{m}w_l)^{p_{m}+k_{m-1}}} \Biggr].
\end{align*}
This can be simplified to
\begin{align*}
C^{-1} = m! \prod_{g=1}^{m} &\frac{2\pi^{p_g/2}}{\Gamma(p_g/2)}\\
&\times\sum_{\vect{k} \in \mathbb{K}} \prod_{j=1}^{m-1} \frac{(p_j + k_{j-1} - 1)! ( \sigma^{-1} \sqrt{p_{j}} \sum_{l=1}^{j}w_l )^{k_j}(p_m + k_{m-1} - 1)!}{k_j! ( \sigma^{-1} \sqrt{p_{j}} \sum_{l=1}^{j}w_l  )^{p_j + k_{j-1}}(\sigma^{-1} \sqrt{p_{m}} \sum_{l=1}^{m}w_l  )^{p_m + k_{m-1}}},
\end{align*}
where $\mathbb{K} = \{ \vect{k} \in \mathbb{Z}^{m-1} : k_1 \in [0, p_1 - 1], \; k_2 \in [0, p_2 + k_1 - 1], \; \ldots, \; k_{m-1} \in [0, p_{m-1} + k_{m-2} - 1] \},$ with $k_0 = 0$.
\end{proof}
\subsection{Bayesian SGS (BSGS)}
\subsubsection{Theory}\label{appendix:bsgs_theory}
\begin{proof}[Proof of Proposition \ref{propn:sgs_equiv}]
We assume a random variable $\vect{z} = (z_1,\ldots,z_p)$ has an SGS Laplace prior, given by
\begin{equation*}
    \pi(\vect{z} \mid \sigma^2, \vect{w},\vect{v}) \propto \exp\left\{-\frac{1}{\sigma} \alpha\sum_{i=1}^p v_{r_v(\vect{z},i)} |z_i| -\frac{1}{\sigma} (1-\alpha)\sum_{j=1}^m w_{r_g(\vect{z},j)}\sqrt{p_j}\|\vect{z}^{(j)}\|_2 \right\}.
\end{equation*}
Now, define $\vect\beta = \hat{\vect{A}}^{-1}\vect{z}$ such that $z_i = \beta_i \hat{a}_i$ for $i\in G_j$. Then, by the transformation of variables, the prior distribution for $\vect\beta$ is given by
\begin{equation*}
    \pi\left(\vect\beta \mid c_v, c_g, \sigma^2 ;  \vect{w},\vect{v} \right) \propto \left|\text{det}(d\vect{z}/d\vect\beta)\right|\pi\left(\vect{z} \mid c_v, c_g, \sigma^2 ; \vect{w},\vect{v} \right). 
\end{equation*}
We have that $(d\vect{z}/d\vect\beta)_i = \hat{a}_i$, so that 
\begin{align*}
\text{det}(d\vect{z}/d\vect\beta) &= \prod_{i=1}^p \hat{a}_i = \prod_{j=1}^m c_g^{p_j\mathbbm{1}\left\{\gamma_j = 1\right\}}\prod_{i\in G_j}c_v^{\mathbbm{1}\left\{\delta_i = 1\right\}}\\
&= \prod_{j=1}^m c_g^{p_j \mathbbm{1}\{\gamma_j = 1\}} c_v^{\sum_{i\in G_j}\mathbbm{1}\{\delta_{i} = 1\}} =  c_g^{\sum_{j=1}^m p_j\mathbbm{1}(\gamma_j = 1)} c_v^{\sum_{i=1}^p \mathbbm{1}(\delta_i = 1)}.
\end{align*}
We can rewrite the prior on $z$ in terms of $\vect\beta$:
\begin{align*}
    &\pi\left(\vect{z} \mid c_v, c_g, \sigma^2 ;  \vect{w},\vect{v} \right) \\
    &\propto \exp\left\{-\frac{1}{\sigma} \alpha\sum_{i=1}^p v_{r_v(\vect{z},i)} |z_i| -\frac{1}{\sigma} (1-\alpha)\sum_{j=1}^m w_{r_g(\vect{z},j)}\sqrt{p_j}\|\vect{z}^{(j)}\|_2 \right\}\\
    &=\prod_{i=1}^p  \exp\left\{-\frac{1}{\sigma} \alpha v_{r_v(\vect{z},i)} |z_i|\right\} \prod_{j=1}^m  \exp\left\{-\frac{1}{\sigma} (1-\alpha) w_{r_g(\vect{z},j)}\sqrt{p_j}\|\vect{z}^{(j)}\|_2 \right\}\\
    &=\prod_{i=1}^p  \exp\left\{-\frac{1}{\sigma} \hat{a}_i\alpha v_{r_v(\hat{\vect{A}}\vect\beta,i)} |\beta_i|\right\}\prod_{j=1}^m  \exp\left\{-\frac{1}{\sigma} (1-\alpha)w_{r_g(\hat{\vect{A}}\vect\beta,j)}\sqrt{p_j}\|\hat{\vect{A}}^{(j)}\vect\beta^{(j)}\|_2 \right\}.
\end{align*}
Therefore, the prior for $\vect\beta$ is given by
\begin{align*}
\pi\left(\vect\beta \mid c_v, c_g, \sigma^2;  \vect{w},\vect{v}\right) &\propto c_g^{\sum_{j=1}^m p_j\mathbbm{1}(\gamma_j = 1)} c_v^{\sum_{i=1}^p \mathbbm{1}(\delta_i = 1)} \prod_{i=1}^p \exp \left\{-\frac{1}{\sigma} |\beta_i|\hat{a}_i \alpha v_{r_v(\hat{\vect{A}} \vect\beta, i)}\right\}\\
&\quad\quad\times \prod_{j=1}^m \exp \left\{-\frac{1}{\sigma} \|\hat{\vect{A}}^{(j)}\vect\beta^{(j)}\|_2\sqrt{p_j} (1-\alpha) w_{r_g(\hat{\vect{A}} \vect\beta, j)}\right\},
\end{align*}
corresponding to the BSGS prior.
\end{proof}
\subsection{SAEM for BGSLOPE}\label{appendix:saem_bgslope}
\subsubsection{Simulation step derivations}\label{appendix:saem_bgslope_simulation}
Here, the conditional distributions required for the Gibbs sampler are derived for each latent variable.
\paragraph{1. Inclusion vector ($\boldsymbol\gamma$).}
According to the dependency graph (Figure \ref{fig:bgslope_graph}), we simulate $\gamma_j$, for $j\in [m]$, via
\begin{equation*}
    \gamma_j \sim  \pi(\gamma_j \mid \vect{\gamma}_{-j}, c, \vect{y}, \vect\beta,\sigma, \theta)=\pi(\gamma_j \mid \vect{\gamma}_{-j}, c, \vect\beta,\sigma, \theta).
\end{equation*}
As $\gamma_j$ is binary, the posterior is a Bernoulli distribution with active probability
\begin{equation*}
    \mathbb{P}(\gamma_j = 1 \mid \vect{\gamma}_{-j}, c, \vect\beta, \sigma, \theta) = \frac{\mathbb{P}(\gamma_j = 1 \mid \theta)\pi(\vect\beta\mid\gamma_j = 1, \vect{\gamma}_{-j}, c, \sigma)}{\sum_{\gamma_j \in \{0,1\}} \mathbb{P}(\gamma_j \mid \theta) \pi(\vect\beta \mid \gamma_j, \vect{\gamma}_{-j}, c, \sigma)}.
\end{equation*}
This describes the posterior probability of a binary signal indicator for the $j$th group, using the prior $\mathbb{P}(\gamma_j = 1 \mid \theta) = \theta$ and the conditional likelihood of $\vect\beta$ given $\gamma_j = 1$ and $\gamma_j = 0$. 

We have that $\mathbb{P}(\gamma_j =1\mid \theta) = \theta$ and $\mathbb{P}(\gamma_j =0\mid \theta) = 1-\theta$ via the Bernoulli distribution on $\gamma_j$, and, 
\begin{align*}
&\pi(\vect\beta\mid \gamma_j = 1, \vect{\gamma}_{-j}, c, \sigma)= c^{p_j} \exp\left\{\frac{-c}{\sigma}\|\vect\beta^{(j)}\|_2 \sqrt{p_j}w_{r_g(\tilde{\vect{A}}^1\vect\beta,j)}\right\}\\
&\quad\quad\quad\quad\quad\quad\times c^{\sum_{k \in [m]\symbol{92}\{j\}}p_{k}\mathbbm{1}(\gamma_{k} = 1)} \prod_{k \in [m]\symbol{92}\{j\}}\exp\left\{\frac{-1}{\sigma}\tilde{a}_{k}\|\vect\beta^{(k)}\|_2 \sqrt{p_{k}}w_{r_g(\tilde{\vect{A}}^1\vect\beta,k)}\right\},\\
&\pi(\vect\beta\mid\gamma_j = 0, \vect{\gamma}_{-j}, c, \sigma)=  \exp\left\{\frac{-1}{\sigma}\|\vect\beta^{(j)}\|_2 \sqrt{p_j}w_{r_g(\tilde{\vect{A}}^0\vect\beta,j)}\right\}\\
&\quad\quad\quad\quad\quad\quad\times c^{\sum_{k \in [m]\symbol{92}\{j\}}p_{k}\mathbbm{1}(\gamma_k = 1)} \prod_{k \in [m]\symbol{92}\{j\}}\exp\left\{\frac{-1}{\sigma}\tilde{a}_{k}\|\vect\beta^{(k)}\|_2 \sqrt{p_{k}}w_{r_g(\tilde{\vect{A}}^0\vect\beta,k)}\right\},
 \end{align*}
where $\tilde{\vect{A}}^1$ and $\tilde{\vect{A}}^0$ are the same as $\tilde{\vect{A}}$ except for the $j$th group, for which $\tilde{a}_{j}^1=c$ and $\tilde{a}_{j}^0=1$. This leads to
\begin{align*}
     & \mathbb{P}(\gamma_j = 1 \mid \vect{\gamma}_{-j}, c, \vect\beta, \sigma, \theta) = \frac{\mathbb{P}(\gamma_j = 1 \mid \theta)\pi(\vect\beta\mid\gamma_j = 1, \vect{\gamma}_{-j}, c, \sigma)}{\sum_{\gamma_j \in \{0,1\}} \mathbb{P}(\gamma_j \mid \theta) \pi(\vect\beta \mid \gamma_j, \vect{\gamma}_{-j}, c, \sigma) }\\
      & = \left[\frac{\mathbb{P}(\gamma_j = 1 \mid \theta)\pi(\vect\beta\mid\gamma_j = 1, \vect{\gamma}_{-j}, c, \sigma) + \mathbb{P}(\gamma_j = 0 \mid \theta)\pi(\vect\beta\mid\gamma_j = 0, \vect{\gamma}_{-j}, c, \sigma)}{\mathbb{P}(\gamma_j = 1 \mid \theta)\pi(\vect\beta\mid\gamma_j = 1, \vect{\gamma}_{-j}, c, \sigma)}\right]^{-1}\\
     & = \left[1+\frac{\mathbb{P}(\gamma_j = 0 \mid \theta)\pi(\vect\beta\mid\gamma_j = 0, \vect{\gamma}_{-j}, c, \sigma)}{\mathbb{P}(\gamma_j = 1 \mid \theta)\pi(\vect\beta\mid\gamma_j = 1, \vect{\gamma}_{-j}, c, \sigma)}\right]^{-1}\\
      &= \left[1+\frac{(1-\theta)\pi(\vect\beta\mid\gamma_j = 0, \vect{\gamma}_{-j}, c, \sigma)}{\theta \pi(\vect\beta\mid\gamma_j = 1, \vect{\gamma}_{-j}, c, \sigma)}\right]^{-1},
\end{align*}
where the $c^{\sum_{k \in [m]\symbol{92}\{j\}}p_{k}\mathbbm{1}(\gamma_{k} = 1)}$ term present in both has been canceled, and 
\begin{align*}
   &\pi(\vect\beta\mid\gamma_j = 0, \vect{\gamma}_{-j}, c, \sigma)\\
   &= \exp\left\{\frac{-1}{\sigma}\|\vect\beta^{(j)}\|_2 \sqrt{p_j}w_{r_g(\tilde{\vect{A}}^0\vect\beta,j)}\right\} \prod_{k \in [m]\symbol{92}\{j\}}\exp\left\{\frac{-1}{\sigma}\tilde{a}_{(k)}\|\vect\beta^{(k)}\|_2 \sqrt{p_{k}}w_{r_g(\tilde{\vect{A}}^0\vect\beta,k)}\right\},\\
   &\pi(\vect\beta\mid\gamma_j = 1, \vect{\gamma}_{-j}, c, \sigma)\\
   &= c^{p_j} \exp\left\{\frac{-c}{\sigma}\|\vect\beta^{(j)}\|_2 \sqrt{p_j}w_{r_g(\tilde{\vect{A}}^1\vect\beta,j)}\right\} \prod_{k \in [m]\symbol{92}\{j\}}\exp\left\{\frac{-1}{\sigma}\tilde{a}_{(k)}\|\vect\beta^{(k)}\|_2 \sqrt{p_{k}}w_{r_g(\tilde{\vect{A}}^1\vect\beta,k)}\right\}.
\end{align*}
Sampling here requires updating an ordered list at each group $j\in [m]$, which is expensive, so $\tilde{\vect{A}}^0$ and $\tilde{\vect{A}}^1$ are approximated by using the previous estimate of $\tilde{\vect{A}}$. This partially retains information on $\gamma_j$ and stabilizes updates; a similar approach is taken for ABSLOPE. The approximation allows the $\prod_{k \in [m]\symbol{92}\{j\}}$ terms to cancel, so that
\begin{align*}
     & \mathbb{P}(\gamma_j = 1 \mid \vect{\gamma}_{-j}, c, \vect\beta, \sigma, \theta, \tilde{\vect{A}}) =\left[1+\frac{(1-\theta)\exp\left\{\frac{-1}{\sigma}\|\vect\beta^{(j)}\|_2 \sqrt{p_j}w_{r_g(\tilde{\vect{A}}\vect\beta,j)}\right\}}{\theta c^{p_j} \exp\left\{\frac{-c}{\sigma}\|\vect\beta^{(j)}\|_2 \sqrt{p_j}w_{r_g(\tilde{\vect{A}}\vect\beta,j)}\right\}}\right]^{-1}\\
      &= \frac{\theta c^{p_j} \exp\left\{\frac{-c}{\sigma}\|\vect\beta^{(j)}\|_2 \sqrt{p_j}w_{r_g(\tilde{\vect{A}}\vect\beta,j)}\right\}}{(1-\theta)\exp\left\{\frac{-1}{\sigma}\|\vect\beta^{(j)}\|_2 \sqrt{p_j}w_{r_g(\tilde{\vect{A}}\vect\beta,j)}\right\}+\theta c^{p_j} \exp\left\{\frac{-c}{\sigma}\|\vect\beta^{(j)}\|_2 \sqrt{p_j}w_{r_g(\tilde{\vect{A}}\vect\beta,j)}\right\}}.
\end{align*}     
This leads to posterior
\begin{align*}
 &\gamma_j \sim \operatorname{Bernoulli}\left(\frac{L_1}{L_1 + L_2}\right), \; \text{where},\\
     &L_1 = \theta c^{p_j} \exp\left\{\frac{-c}{\sigma}\|\vect\beta^{(j)}\|_2 \sqrt{p_j}w_{r_g(\tilde{\vect{A}}\vect\beta,j)}\right\},\\ 
     &L_2 = (1-\theta)\exp\left\{\frac{-1}{\sigma}\|\vect\beta^{(j)}\|_2 \sqrt{p_j}w_{r_g(\tilde{\vect{A}}\vect\beta,j)}\right\}.
\end{align*}
\paragraph{2. Mixture proportion ($\theta$).}
We have that
\begin{equation*}
    \theta \sim \pi(\theta \mid \vect\gamma,c,\vect{y},\vect\beta,\sigma,\tilde{\vect{A}}) = \pi(\theta\mid\vect\gamma,\vect\beta,\sigma,\tilde{\vect{A}})\propto \pi(\theta)\pi(\vect\gamma \mid \theta),
\end{equation*}
where $\pi(\theta)$ is a $\text{Beta}(d_1,d_2)$ distribution and $\pi(\vect\gamma \mid \theta)$ is a multivariate Bernoulli distribution. Therefore, by conjugacy, the posterior is given by a Beta distribution:
\begin{equation*}
\theta\sim\text{Beta}\left(d_1+\sum_{j=1}^m \mathbbm{1}(\gamma_j=1),d_2+\sum_{j=1}^m\mathbbm{1}(\gamma_j=0)\right).
\end{equation*}
\paragraph{3. Signal strength ratio ($c$).}
We sample $c$ from 
\begin{equation*}
  c  \sim \pi(c\mid\vect\gamma,\vect{y},\vect\beta,\sigma,\theta,\tilde{\vect{A}}) = \pi(c\mid\vect\gamma,\vect\beta,\sigma,\tilde{\vect{A}}) \propto \pi(c)\pi(\vect\beta \mid c,\vect\gamma,\sigma,\tilde{\vect{A}}).
\end{equation*}
As $c\sim \mathcal{U}[0,1]$, this reduces to the BGSLOPE prior on $\vect\beta$ for an active group (the prior for non-active groups does not contain $c$)
\begin{equation*}
       \pi(\vect\beta \mid c,\vect\gamma,\sigma,\tilde{\vect{A}})= c^{\sum_{j=1}^m p_j \mathbbm{1}\left(\gamma_j=1\right) }  \exp \left\{\sum_{j=1}^m-\frac{c}{\sigma}\|\vect\beta^{(j)}\|_2\sqrt{p_j}  w_{r_g(\tilde{\vect{A}} \vect\beta, j)}\mathbbm{1}(\gamma_j=1)\right\}.
\end{equation*}
This is in the form of a Gamma distribution, truncated to $[0,1]$ to satisfy the constraint on $c$,
\begin{equation*}
c\sim \text{Gamma}\left( 1 + \sum_{j=1}^{m}p_j \mathbbm{1}(\gamma_j = 1), \frac{1}{\sigma} \sum_{j=1}^{m} \|\vect\beta^{(j)}\|_2\sqrt{p_j} w_{r_g(\tilde{\vect{A}} \vect\beta, j)}\mathbbm{1}(\gamma_j=1)\right).
\end{equation*}

\subsubsection{Stochastic and maximization step derivations}\label{appendix:saem_bgslope_opt}
\paragraph{Update for $\sigma$.}
Considering only the terms of the log-likelihood (Equation \ref{eqn:gslope_likelihood}) containing $\sigma$
\begin{align*}
      \sigma^\text{MLE}_{[t]}&=\argmax_{\sigma \in \mathbb{R}_{>0}}\bigg\{-(n + 2)\log\sigma- \frac{1}{2\sigma^2} \|\vect{y} - \vect{X}\vect\beta_{[t]} \|^2_2\\
&\quad\quad\quad\quad\quad\;\;\;\quad\quad\quad\quad\quad\quad\quad\quad-\frac{1}{\sigma} \sum_{j=1}^{m}(\tilde{a}_{j})_{[t]} \sqrt{p_j}  w_{r_g(\tilde{\vect{A}}_{[t]}\vect\beta_{[t]},j)}\|\vect\beta^{(j)}_{[t]}\|_2\bigg\}\\
     &=\argmax_{\sigma \in \mathbb{R}_{>0}}\left\{ -(n + 2)\log\sigma- \frac{K_1}{2\sigma^2} - \frac{1}{\sigma}K_2\right\},
\end{align*}
where $K_1=\|\vect{y}-\vect{X}\vect\beta_{[t]}\|_2^2 $ and $K_2=\sum_{j=1}^m (\tilde{a}_{j})_{[t]} \|\vect\beta_{[t]}^{(j)}\|_2 \sqrt{p_j}  w_{r_g(\tilde{\vect{A}}_{[t]}\vect\beta_{[t]}, j)}$. Differentiating this with respect to $\sigma$ gives
\begin{equation*}
K_1 + K_2\sigma = (n+2)\sigma^2 \implies \sigma = \frac{K_2 \pm \sqrt{K_2^2 + 4K_1(n+2)}}{2(n+2)}.
\end{equation*}
The negative root is negative if
    \begin{align*}
        K_2 - \sqrt{K_2^2 + 4K_1(n+2)} < 0 &\implies K_2 < \sqrt{K_2^2 + 4K_1(n+2)} \\
         &\implies 4K_1(n+2) > 0,
    \end{align*}
which is true as $K_1>0$ (unless we have perfect $\vect{y}$ recovery, which leads to a model with no variance). As $\sigma>0$, the positive root is used:
      \begin{equation*}
        \sigma^\text{MLE}_{[t]} = \frac{K_2 + \sqrt{K_2^2 + 4K_1(n+2)}}{2(n+2)}.
    \end{equation*}
Omitting the gSLOPE penalty, the estimator simplifies to $\sigma_{[t]} = \sqrt{K_1/(n+2)}$, which matches the classical MLE formula of $\sigma$ when $\vect{\beta}$ is also estimated via MLE \citep{Jiang2022AdaptiveData}.

\subsubsection{Algorithm}
\begin{algorithm}[h]
   \caption{SAEM for BGSLOPE}
   \label{alg:bgslope_saem_alg}
\begin{algorithmic}[1]
   \STATE {\bfseries Input:} Initial values ($\vect\beta_{[0]}, \sigma_{[0]}, c_{[0]}, \theta_{[0]}$), maximum number of iterations $T > 0$, tolerance $\epsilon > 0$
\WHILE{$t \leq T$ \AND $\|\vect\beta_{[t+1]} - \vect\beta_{[t]} \|_2^2> \epsilon$}
\STATE \textbf{Simulation step:}
\STATE Sample $(\gamma_j)_{[t]}$ from Equation \ref{eqn:bgslope_gamma_update} for $j=1,\ldots,m$
\STATE Sample $\theta_{[t]}$ from Equation \ref{eqn:bgslope_theta_update}
\STATE Sample $c_{[t]}$ from Equation \ref{eqn:bgslope_c_update}
\STATE \textbf{Stochastic approximation step:}
\STATE Calculate $\vect\beta^\text{MLE}_{[t]}$ via Equation \ref{eqn:gslope_opt_problem}
\STATE Calculate $\sigma_{[t]}^\text{MLE}$ from Equation \ref{eqn:gslope_opt_problem_sigma}
\STATE Update parameters
\begin{equation*}
    \vect\beta_{[t+1]} = \vect\beta_{[t]} + \eta_t(\vect\beta^\text{MLE}_{[t]}-  \vect\beta_{[t]}),\quad \sigma_{[t+1]} = \sigma_{[t]} + \eta_t(\sigma^\text{MLE}_{[t]}-\sigma_{[t]})
\end{equation*}
\STATE Update step size $\eta_t =
\begin{cases}
1, & \text{if } t \leq 20, \\
\frac{1}{t-20}, & \text{if } t > 20
\end{cases}$
\STATE $t = t+1$
\ENDWHILE
 \STATE {\bfseries Output:} $(\vect\beta_{[t+1]}, \vect\gamma_{[t]})$, such that $\hat{\vect\gamma}_{[t]} \leftarrow \frac{1}{20} \sum_{k=t-19}^{t} \vect\gamma_{[k]}$ (if $t<20$ all iterations of $\vect\gamma$ are used), giving the fitted coefficients $\hat{\vect\beta}_{[t+1]} = \vect\beta_{[t+1]} \cdot \mathbbm{1}\left(\hat{\vect\gamma}_{[t]} > 0.5\right)$
\end{algorithmic}
\end{algorithm}

\subsection{SAEM for BSGS}
\subsubsection{Simulation step derivations}
\paragraph{1. Group inclusion vector ($\boldsymbol\gamma$).}
By Figure \ref{fig:bsgs_graph}, $\vect\gamma$ is simulated by, for $j\in [m]$,
\begin{equation*}
\gamma_j \sim \pi(\gamma_j \mid \vect\gamma_{-j},c_v,c_g,\vect{y},\vect\beta,\sigma,\theta_v,\theta_g,\vect\delta) =  \pi(\gamma_j \mid \vect\gamma_{-j},c_v,c_g,\vect\beta,\sigma,\theta_v,\theta_g,\vect\delta).
\end{equation*}
As $\gamma_j$ is binary, the posterior is given by a Bernoulli distribution with success probability
\begin{align*}
\mathbb{P}(\gamma_j = 1\mid  \vect\gamma_{-j},c_v,c_g,\vect\beta,&\sigma,\theta_v,\theta_g,\vect\delta)  \\
&=\frac{\mathbb{P}(\gamma_j = 1 \mid\theta_g)\pi(\vect\beta,\vect\delta \mid \gamma_j=1,\vect\gamma_{-j},c_v,c_g,\sigma,\theta_v,\theta_g)}{\sum_{\gamma_j \in \{0,1\}}\mathbb{P}(\gamma_j\mid \theta_g)\pi(\vect\beta,\vect\delta \mid \gamma_j,\vect\gamma_{-j},c_v,c_g,\sigma,\theta_v,\theta_g)}.
\end{align*}
This can be rewritten as
\begin{align}
\mathbb{P}(\gamma_j = 1 \mid  \vect\gamma_{-j},c_v,c_g,&\vect\beta,\sigma,\theta_v,\theta_g,\vect\delta) \label{eqn:gamma_update_bsgs_1}\\
&= \left[1+\frac{\mathbb{P}(\gamma_j = 0 \mid \theta_g)\pi(\vect\beta,\vect\delta \mid \gamma_j=0,\vect\gamma_{-j},c_v,c_g,\sigma,\theta_v,\theta_g)}{\mathbb{P}(\gamma_j = 1 \mid \theta_g)\pi(\vect\beta,\vect\delta \mid \gamma_j=1,\vect\gamma_{-j},c_v,c_g,\sigma,\theta_v,\theta_g)}\right]^{-1},\nonumber
\end{align}
where $\mathbb{P}(\gamma_j = 1 \mid \theta_g) = \theta_g$ and $\mathbb{P}(\gamma_j = 0 \mid \theta_g) = 1-\theta_g$. The joint conditional distribution of $\vect\beta$ and $\vect\delta$ can be decomposed into the active and non-active cases
\begin{align*}
   &\pi(\vect\beta,\vect\delta \mid \gamma_j=1,\vect\gamma_{-j},c_v,c_g,\sigma,\theta_v,\theta_g) = \pi(\vect\delta \mid \gamma_j = 1, \theta_v)\pi(\vect\beta \mid \gamma_j=1,\vect\gamma_{-j},c_v,c_g,\sigma,\vect\delta), \\
 &\pi(\vect\beta,\vect\delta \mid \gamma_j=0,\vect\gamma_{-j},c_v,c_g,\sigma,\theta_v,\theta_g) = \pi(\vect\delta \mid \gamma_j = 0, \theta_v)\pi(\vect\beta \mid \gamma_j=0,\vect\gamma_{-j},c_v,c_g,\sigma,\vect\delta).
\end{align*}
The first term is given by
\begin{align*}
    \pi(\vect\delta \mid\gamma_j = 1, \theta_v) &= \prod_{i \in G_j} \theta_v^{\delta_i} (1-\theta_v)^{1-\delta_i}\\
    &\;\times\prod_{k \in [m]\symbol{92}\{j\}} \prod_{i \in G_{k}} \left[ \theta_v^{\delta_i} (1-\theta_v)^{1-\delta_i} \mathbbm{1}(\gamma_{k} = 1) + \mathbbm{1}(\gamma_{k} = 0, \delta_i = 0) \right],\\
 \pi(\vect\delta\mid\gamma_j = 0, \theta_v) &=  \prod_{k \in [m]\symbol{92}\{j\}} \prod_{i \in G_{k}} \left[ \theta_v^{\delta_i} (1-\theta_v)^{1-\delta_i} \mathbbm{1}(\gamma_{k} = 1) + \mathbbm{1}(\gamma_{k} = 0, \delta_i = 0) \right].
\end{align*}
The second term is the prior of the regression coefficients, which under an active group is
\begin{align*}
     \pi(\vect\beta \mid \gamma_j=1,\vect\gamma_{-j},c_v,c_g,\sigma,\vect\delta)=& c_g^{p_j} c_v^{\sum_{i\in G_j}\mathbbm{1}\{\delta_{i} = 1\}}\\
     &\times\prod_{i\in G_j} \exp \left\{-\frac{1}{\sigma} |\beta_i| c_g c_v^{\mathbbm{1}(\delta_i =1)}v_{r_v(\hat{\vect{A}}^1 \vect\beta, i)}\right\}\\
     &\times \exp \left\{-\frac{1}{\sigma} c_g\|\vect{\kappa}^{(j)}\vect\beta^{(j)}\|_2\sqrt{p_j}  w_{r_g(\hat{\vect{A}}^1 \vect\beta, j)}\right\}\\
     &\times\prod_{k \in [m]\symbol{92}\{j\}} c_g^{p_{k} \mathbbm{1}\{\gamma_{k} = 1\}} c_v^{\sum_{i\in G_{k}}\mathbbm{1}\{\delta_{i} = 1\}}\\
     &\times\prod_{i\notin G_j} \exp \left\{-\frac{1}{\sigma} |\beta_i|\hat{a}_i^1 v_{r_v(\hat{\vect{A}}^1 \vect\beta, i)}\right\}\\
     &\times \prod_{k \in [m]\symbol{92}\{j\}} \exp \left\{-\frac{1}{\sigma} \|(\hat{\vect{A}}^1)^{(k)}\vect\beta^{(k)}\|_2\sqrt{p_{k}}  w_{r_g(\hat{\vect{A}}^1 \vect\beta, k)}\right\},
\end{align*}
where $\vect{\kappa}^{(j)} = \text{diag}(c_v^{\mathbbm{1}(\delta_i = 1)})$ for all $i \in G_j$. Matrices $\hat{\vect{A}}^0$ and $\hat{\vect{A}}^1$ have the same entries as $\hat{\vect{A}}$ except for all $i \in G_j$, where the entries are $\hat{a}^0_{i} = 1$ and $\hat{a}^1_{i} = c_gc_v^{\mathbbm{1}(\delta_i=1)}$. Under an inactive group, the prior is given by 
\begin{align*}
     \pi(\vect\beta \mid \gamma_j=0,\vect\gamma_{-j},c_v,c_g,\sigma,\vect\delta)&= \prod_{i\in G_j} \exp\left\{-\frac{1}{\sigma} |\beta_i| v_{r_v(\hat{\vect{A}}^0 \vect\beta, i)}\right\}\\
     &\times \exp \left\{-\frac{1}{\sigma}\|\vect\beta^{(j)}\|_2\sqrt{p_j}  w_{r_g(\hat{\vect{A}}^0 \vect\beta, j)}\right\}\\
     &\times\prod_{k \in [m]\symbol{92}\{j\}} c_g^{p_{k} \mathbbm{1}\{\gamma_{k} = 1\}} c_v^{\sum_{i\in G_{k}}\mathbbm{1}\{\delta_{i} = 1\}}\\
     &\times \prod_{i\notin G_j} \exp \left\{-\frac{1}{\sigma} |\beta_i|\hat{a}_i^0 v_{r_v(\hat{\vect{A}}^0 \vect\beta, i)}\right\}\\
     &\times \prod_{k \in [m]\symbol{92}\{j\}} \exp \left\{-\frac{1}{\sigma} \|(\hat{\vect{A}}^0)^{(k)}\vect\beta^{(k)}\|_2\sqrt{p_{k}}  w_{r_g(\hat{\vect{A}}^0 \vect\beta, k)}\right\}.
\end{align*}
Therefore, considering first the fraction from Equation \ref{eqn:gamma_update_bsgs_1}, we have
\begin{align*}
    \text{Numerator} &=(1-\theta_g) \prod_{k \in [m]\symbol{92}\{j\}} \prod_{i \in G_{k}} \left[ \theta_v^{\delta_i} (1-\theta_v)^{1-\delta_i} \mathbbm{1}(\gamma_{k} = 1) + \mathbbm{1}(\gamma_{k} = 0, \delta_i = 0) \right]\\
    &\quad\;\times \prod_{i\in G_j}\exp \left\{-\frac{1}{\sigma} |\beta_i| v_{r_v(\hat{\vect{A}}^0 \vect\beta, i)}\right\}\\
     &\quad\;\times \exp \left\{-\frac{1}{\sigma}\|\vect\beta^{(j)}\|_2\sqrt{p_j}  w_{r_g(\hat{\vect{A}}^0 \vect\beta, j)}\right\}\\
     &\quad\;\times\prod_{k \in [m]\symbol{92}\{j\}} c_g^{p_{k} \mathbbm{1}\{\gamma_{k} = 1\}} c_v^{\sum_{i\in G_{k}}\mathbbm{1}\{\delta_{i} = 1\}}\prod_{i\notin G_j} \exp \left\{-\frac{1}{\sigma} |\beta_i|\hat{a}_i^1 v_{r_v(\hat{\vect{A}}^0 \vect\beta, i)}\right\}\\
     &\quad\;\times \prod_{k \in [m]\symbol{92}\{j\}} \exp \left\{-\frac{1}{\sigma} \|(\hat{\vect{A}}^0)^{(k)}\vect\beta^{(k)}\|_2\sqrt{p_{k}}  w_{r_g(\hat{\vect{A}}^0 \vect\beta, k)}\right\},
   \end{align*}
   \begin{align*}
    \text{Denominator} &= \theta_g \prod_{i \in G_j} \theta_v^{\delta_i} (1-\theta_v)^{1-\delta_i}\\
    &\quad \;\times \prod_{k \in [m]\symbol{92}\{j\}} \prod_{i \in G_{k}} \left[ \theta_v^{\delta_i} (1-\theta_v)^{1-\delta_i} \mathbbm{1}(\gamma_{k} = 1) + \mathbbm{1}(\gamma_{k} = 0, \delta_i = 0) \right]\\
    & \quad\;\times c_g^{p_j} c_v^{\sum_{i\in G_j}\mathbbm{1}\{\delta_{i} = 1\}}\prod_{i\in G_j} \exp \left\{-\frac{1}{\sigma} |\beta_i| c_g c_v^{\mathbbm{1}(\delta_i =1)}v_{r_v(\hat{\vect{A}}^1 \vect\beta, i)}\right\}\\
     &\quad\;\times \exp \left\{-\frac{1}{\sigma} c_g\|\vect{\kappa}^{(j)}\vect\beta^{(j)}\|_2\sqrt{p_j}  w_{r_g(\hat{\vect{A}}^1 \vect\beta, j)}\right\}\\
     &\quad\;\times\prod_{k \in [m]\symbol{92}\{j\}} c_g^{p_{k} \mathbbm{1}\{\gamma_{k} = 1\}} c_v^{\sum_{i\in G_{k}}\mathbbm{1}\{\delta_{i} = 1\}}\prod_{i\notin G_j} \exp \left\{-\frac{1}{\sigma} |\beta_i|\hat{a}_i^1 v_{r_v(\hat{\vect{A}}^1 \vect\beta, i)}\right\}\\
     &\quad\;\times \prod_{k \in [m]\symbol{92}\{j\}} \exp \left\{-\frac{1}{\sigma} \|(\hat{\vect{A}}^1)^{(k)}\vect\beta^{(k)}\|_2\sqrt{p_{k}}  w_{r_g(\hat{\vect{A}}^1 \vect\beta, k)}\right\}.
\end{align*}
As with BGSLOPE (Appendix \ref{appendix:saem_bgslope_simulation}), we will use $\hat{\vect{A}}$ as an approximation for both $\hat{\vect{A}}^0$ and $\hat{\vect{A}}^1$. Therefore, several terms can be canceled to obtain
\begin{align*}
      \text{Numerator} &=(1-\theta_g)\times \prod_{i\in G_j} \exp \left\{-\frac{1}{\sigma} |\beta_i| v_{r_v(\hat{\vect{A}} \vect\beta, i)}\right\}\\
     &\quad \;\times \exp \left\{-\frac{1}{\sigma}\|\vect\beta^{(j)}\|_2\sqrt{p_j}  w_{r_g(\hat{\vect{A}} \vect\beta, j)}\right\},\\
    \text{Denominator} &= \theta_g \prod_{i \in G_j} \theta_v^{\delta_i} (1-\theta_v)^{1-\delta_i} \\
    & \quad \;\times c_g^{p_j} c_v^{\sum_{i\in G_j}\mathbbm{1}\{\delta_{i} = 1\}}\prod_{i\in G_j} \exp \left\{-\frac{1}{\sigma} |\beta_i| c_g c_v^{\mathbbm{1}(\delta_i =1)}v_{r_v(\hat{\vect{A}} \vect\beta, i)}\right\}\\
     &\quad \;\times \exp \left\{-\frac{1}{\sigma} c_g\|\vect{\kappa}^{(j)}\vect\beta^{(j)}\|_2\sqrt{p_j}  w_{r_g(\hat{\vect{A}} \vect\beta, j)}\right\}.
\end{align*}
Finally,
\begin{align*}
    &\mathbb{P}(\gamma_j = 1 \mid  \vect\gamma_{-j},c_v,c_g,\vect\beta,\sigma,\theta_v,\theta_g,\vect\delta) = \frac{L_1'}{L_1'+L_2'},\; \text{where},\\
    &L_1'=\theta_g c_g^{p_j} c_v^{\sum_{i\in G_j}\mathbbm{1}\{\delta_{i} = 1\}} \exp \left\{-\frac{c_g}{\sigma}\|\vect{\kappa}^{(j)}\vect\beta^{(j)}\|_2\sqrt{p_j}  w_{r_g(\hat{\vect{A}} \vect\beta, j)}\right\}\\
    &\quad \;\quad \quad \;\times \prod_{i \in G_j} \theta_v^{\delta_i} (1-\theta_v)^{1-\delta_i} \exp \left\{-\frac{c_g}{\sigma} |\beta_i| c_v^{\mathbbm{1}(\delta_i =1)}v_{r_v(\hat{\vect{A}} \vect\beta, i)}\right\},\\
    & L_2'=(1-\theta_g)\exp \left\{-\frac{1}{\sigma}\|\vect\beta^{(j)}\|_2\sqrt{p_j}  w_{r_g(\hat{\vect{A}} \vect\beta, j)}\right\} \prod_{i\in G_j} \exp \left\{-\frac{1}{\sigma} |\beta_i| v_{r_v(\hat{\vect{A}} \vect\beta, i)}\right\}.
\end{align*}
Therefore, $\gamma_j\sim \text{Bernoulli}\left(\frac{L_1'}{L_1'+L_2'}\right)$, for $j\in [m]$.
\paragraph{2. Variable inclusion vector ($\boldsymbol\delta$).}
For each $i\in [p]$, the BSGS dependency graph shows that $\delta_i$ is simulated by
\begin{equation*}
    \delta_i \sim \pi(\delta_i \mid \vect\delta_{-i},c_v,c_g,\vect{y},\vect\beta,\sigma,\theta_v,\theta_g,\vect\gamma) =  \pi(\delta_i \mid \vect\delta_{-i},c_v,c_g,\vect\beta,\sigma,\theta_v,\theta_g,\vect\gamma).
\end{equation*}
Similar to the derivation for the $\vect\gamma$ simulation step, this can be expressed as, denoting $j$ as the group that variable $i$ sits in,
\begin{align}
   \mathbb{P}&(\delta_i = 1 \mid  \vect\delta_{-i},c_v,c_g,\vect\beta,\sigma,\theta_v,\theta_g,\vect\gamma)\nonumber \\
   &= \frac{\mathbb{P}(\delta_i = 1 \mid \theta_v,\gamma_j)\pi(\vect\beta \mid \delta_i=1,\vect\delta_{-i},c_v,c_g,\sigma,\theta_v,\theta_g,\vect\gamma)}{\sum_{\delta_i \in \{0,1\}}\mathbb{P}(\delta_i\mid \theta_v,\gamma_j)\pi(\vect\beta \mid \delta_i,\vect\delta_{-i},c_v,c_g,\sigma,\theta_v,\theta_g,\vect\gamma)}\nonumber\\
   &=\left[1+\frac{\mathbb{P}(\delta_i = 0 \mid \theta_v,\gamma_j)\pi(\vect\beta\mid \delta_i=0,\vect\delta_{-i},c_v,c_g,\sigma,\theta_v,\theta_g,\vect\gamma)}{\mathbb{P}(\delta_i = 1 \mid \theta_v,\gamma_j)\pi(\vect\beta \mid \delta_i=1,\vect\delta_{-i},c_v,c_g,\sigma,\theta_v,\theta_g,\vect\gamma)}\right]^{-1}.\label{eqn:delta_eqn_1}
\end{align}
Now, the prior terms are given by $\mathbb{P}(\delta_i = 1 \mid \theta_v, \gamma_j) = \theta_v\mathbbm{1}(\gamma_j = 1)$ and $\mathbb{P}(\delta_i = 0 \mid \theta_v, \gamma_j) = 1-\theta_v\mathbbm{1}(\gamma_j = 1)$. Denoting  $\hat{\vect{A}}^0$ and $\hat{\vect{A}}^1$ as matrices with the same elements as $\hat{\vect{A}}$ except for variable $i$, which has entries $\hat{\vect{a}}^0_{i} = c_g^{\mathbbm{1}(\gamma_j = 1)}$ and $\hat{\vect{a}}^1_{i} = c_v c_g$. The second term, for an active variable, is given by 
 \begin{align*}
       \pi(\vect\beta \mid \delta_i=1,\vect\delta_{-i},c_v,c_g,\sigma,\theta_v,\theta_g,&\vect\gamma)= c_g^{p_j} c_v c_g^{\sum_{k \in [m]\symbol{92}\{j\}} p_{k}\mathbbm{1}(\gamma_{k} = 1)} c_v^{\sum_{k \in [p]\symbol{92}\{i\}} \mathbbm{1}(\delta_{k} = 1)}\\
       &\times\exp \left\{-\frac{1}{\sigma} |\beta_i|c_g c_v v_{r_v(\hat{\vect{A}}^1 \vect\beta, i)}\right\}\\
  &\times \exp \left\{-\frac{1}{\sigma} c_g\|\tilde{\vect{\kappa}}^{(j)}\vect\beta^{(j)}\|_2\sqrt{p_j}  w_{r_g(\hat{\vect{A}}^1 \vect\beta, j)}\right\}\\
       &\times \prod_{k \in [p]\symbol{92}\{i\}} \exp \left\{-\frac{1}{\sigma} |\beta_{k}|\hat{a}_{k}^1 v_{r_v(\hat{\vect{A}}^1 \vect\beta, k)}\right\}\\
&\times \prod_{k \in [m]\symbol{92}\{j\}} \exp \left\{-\frac{1}{\sigma} \|(\hat{\vect{A}}^1)^{(k)}\vect\beta^{(k)}\|_2\sqrt{p_{k}} w_{r_g(\hat{\vect{A}}^1 \vect\beta, k)}\right\},
    \end{align*}
where $\tilde{\vect{\kappa}}^{(j)} = \text{diag}(c_v^{\mathbbm{1}(\delta_k = 1)})$ for all $k \in G_j, k\neq i$, and the $i$th element is $c_v$. For an inactive variable
 \begin{align*}
       \pi(\vect\beta \mid \delta_i=0,\vect\delta_{-i},c_v,c_g,\sigma,\theta_v,\theta_g,\vect\gamma)=&  c_g^{\sum_{j=1}^m p_j\mathbbm{1}(\gamma_j = 1)} c_v^{\sum_{k \in [p]\symbol{92}\{i\}} \mathbbm{1}(\delta_{k} = 1)}\\
       &\times \exp \left\{-\frac{1}{\sigma} |\beta_i|c_g^{\mathbbm{1}(\gamma_j = 1)} v_{r_v(\hat{\vect{A}}^0 \vect\beta, i)}\right\}\\
       &\times \prod_{k \in [p]\symbol{92}\{i\}} \exp \left\{-\frac{1}{\sigma} |\beta_{k}|\hat{a}_{k}^0 v_{r_v(\hat{\vect{A}}^0 \vect\beta, k)}\right\}\\
&\times \prod_{j=1}^m \exp \left\{-\frac{1}{\sigma} \|(\hat{\vect{A}}^0)^{(j)}\vect\beta^{(j)}\|_2\sqrt{p_j}  w_{r_g(\hat{\vect{A}}^0 \vect\beta, j)}\right\}.
\end{align*}
A similar approximation as before is used to replace $\hat{\vect{A}}^0$ and $\hat{\vect{A}}^1$ with $\hat{\vect{A}}$. Then, considering the fraction in Equation \ref{eqn:delta_eqn_1}, the condition $\mathbbm{1}(\gamma_j = 1)$ is removed as $\vect\delta$ is sampled only in active groups, so that
\begin{align*}
\text{Numerator} &= ( 1-\theta_v) c_g^{\sum_{j=1}^m p_j\mathbbm{1}(\gamma_j = 1)} c_v^{\sum_{k \in [p]\symbol{92}\{i\}} \mathbbm{1}(\delta_{k} = 1)}\exp \left\{-\frac{1}{\sigma} |\beta_i|c_g^{\mathbbm{1}(\gamma_j = 1)} v_{r_v(\hat{\vect{A}} \vect\beta, i)}\right\}\\
       &\quad\times \prod_{k \in [p]\symbol{92}\{i\}} \exp \left\{-\frac{1}{\sigma} |\beta_{k}|\hat{a}_{k} v_{r_v(\hat{\vect{A}} \vect\beta, k)}\right\}\\
&\quad\times \prod_{j=1}^m \exp \left\{-\frac{1}{\sigma} \|\hat{\vect{A}}^{(j)}\vect\beta^{(j)}\|_2\sqrt{p_j}  w_{r_g(\hat{\vect{A}} \vect\beta, j)}\right\},
\end{align*}
\begin{align*}
\text{Denominator} &=\theta_v c_g^{p_j} c_v c_g^{\sum_{k \in [m]\symbol{92}\{j\}} p_{k}\mathbbm{1}(\gamma_{k} = 1)} c_v^{\sum_{k \in [p]\symbol{92}\{i\}} \mathbbm{1}(\delta_{k} = 1)}\exp \left\{-\frac{1}{\sigma} |\beta_i|c_g c_v v_{r_v(\hat{\vect{A}} \vect\beta, i)}\right\}\\
       &\quad\times \prod_{k \in [p]\symbol{92}\{i\}} \exp \left\{-\frac{1}{\sigma} |\beta_{k}|\hat{a}_{k} v_{r_v(\hat{\vect{A}} \vect\beta, k)}\right\}\\
&\quad\times \prod_{k \in [m]\symbol{92}\{j\}} \exp \left\{-\frac{1}{\sigma} \|\hat{\vect{A}}^{(k)}\vect\beta^{(k)}\|_2\sqrt{p_{k}} w_{r_g(\hat{\vect{A}} \vect\beta, k)}\right\}\\
&\quad\times \exp \left\{-\frac{1}{\sigma} c_g\|\tilde{\vect{\kappa}}^{(j)}\vect\beta^{(j)}\|_2\sqrt{p_j}  w_{r_g(\hat{\vect{A}} \vect\beta, j)}\right\}.
\end{align*}
After cancellations, and combining the group terms in the denominator, we obtain 
    \begin{align*}
          \text{Numerator} &= (1-\theta_v) c_g^{p_j\mathbbm{1}(\gamma_j = 1)}\exp \left\{-\frac{1}{\sigma} |\beta_i|c_g^{\mathbbm{1}(\gamma_j = 1)} v_{r_v(\hat{\vect{A}} \vect\beta, i)}\right\},\\
        \text{Denominator} &=\theta_v c_g^{p_j} c_v  \exp \left\{-\frac{1}{\sigma} |\beta_i|c_g c_v v_{r_v(\hat{\vect{A}} \vect\beta, i)}\right\}.
    \end{align*}
Therefore, for $i \in [p]$, 
\begin{align*}
&\delta_i\sim\text{Bernoulli}\left(\frac{\tilde{L}_1}{\tilde{L}_1 + \tilde{L}_2}\right), \; \text{where},\\
&\tilde{L}_1 = \theta_v c_v \exp \left\{-\frac{1}{\sigma} |\beta_i|c_g c_v \alpha v_{r_v(\hat{\vect{A}} \vect\beta, i)}\right\}, \; \tilde{L}_2  = (1-\theta_v) \exp \left\{-\frac{1}{\sigma} |\beta_i|c_g \alpha v_{r_v(\hat{\vect{A}} \vect\beta, i)}\right\}. 
\end{align*}
\paragraph{3. Group mixture proportion ($\theta_g$).}
 We have that 
\begin{equation*}
\theta_g \sim \pi(\theta  \mid \vect\gamma,c_v,c_g,\vect{y},\vect\delta,\vect\beta,\sigma,\hat{\vect{A}},\theta_v) = \pi(\theta_g \mid\vect\gamma,\vect\beta,\sigma,\hat{\vect{A}},\vect\delta)\propto \pi(\theta_g)\pi(\vect\gamma  \mid \theta_g),
\end{equation*}
such that $\pi(\theta_g)$ is a $\text{Beta}(d_1,d_2)$ distribution. Now, as $\pi(\vect\gamma \mid \theta_g)$ follows a Bernoulli distribution, the posterior is given by a Beta distribution 
\begin{equation*}
\theta_g\sim\text{Beta}\left(d_1+\sum_{j=1}^m \mathbbm{1}(\gamma_j=1),d_2+\sum_{j=1}^m\mathbbm{1}(\gamma_j=0)\right).
\end{equation*}
\paragraph{4. Variable mixture proportion ($\theta_v$).}
Due to the dependency between $\vect\delta$ and $\vect\gamma$, we have 
\begin{equation*}
    \theta_v \sim \pi(\theta_v \mid \vect\gamma,c_v,c_g,\vect{y},\vect\delta,\vect\beta,\sigma,\hat{\vect{A}},\theta_g) = \pi(\theta_v\mid\vect\gamma,\vect\beta,\sigma,\hat{\vect{A}},\vect\delta)\propto \pi(\theta_v)\pi(\vect\delta \mid\vect\gamma, \theta_v).
\end{equation*}
As $\pi(\theta_v)$ is a $\text{Beta}(e_1,e_2)$ distribution, the posterior is given by
\begin{align*}
\pi(\theta_v)\pi(\vect\delta \mid\vect\gamma, \theta_v)&= \theta_v^{e_1-1} (1-\theta_v)^{e_2-1} \\
&\quad\;\times\prod_{j=1}^{m} \prod_{i \in G_j} \left[ \theta_v^{\delta_i} (1-\theta_v)^{1-\delta_i} \mathbbm{1}(\gamma_j = 1) + \mathbbm{1}(\gamma_j = 0, \delta_i = 0) \right].
\end{align*}
Considering the two cases of whether the group is active
\begin{align*}
     \text{Non-active:} \quad  \pi(\theta_v)\pi(\vect\delta \mid\gamma_j = 0, \theta_v)&=\theta_v^{e_1-1} (1-\theta_v)^{e_2-1}\\
        \text{Active:}\quad  \pi(\theta_v)\pi(\vect\delta \mid\gamma_j = 1, \theta_v)&=\theta_v^{e_1-1} (1-\theta_v)^{e_2-1}\prod_{i \in G_j} \left[ \theta_v^{\delta_i} (1-\theta_v)^{1-\delta_i}\right]\\
         &=\theta_v^{e_1-1 + \sum_{i \in G_j}\mathbbm{1}(\delta_i=1)} (1-\theta_v)^{e_2-1 + \sum_{i \in G_j}\mathbbm{1}(\delta_i=0)}.
\end{align*}
Combining the two cases yields
 \begin{align*}
& \pi(\theta_v)\pi(\vect\delta \mid\vect\gamma, \theta_v)\\
 &=\prod_{j=1}^m \theta_v^{e_1-1 + \mathbbm{1}(\gamma_j = 1)\sum_{i \in G_j}\mathbbm{1}(\delta_i=1)} (1-\theta_v)^{e_2-1 +\mathbbm{1}(\gamma_j = 1) \sum_{i \in G_j}\mathbbm{1}(\delta_i=0)}\\
&= \theta_v^{me_1-m + \sum_{j=1}^m\mathbbm{1}(\gamma_j = 1)\sum_{i \in G_j}\mathbbm{1}(\delta_i=1)} (1-\theta_v)^{me_2-m +\sum_{j=1}^m\mathbbm{1}(\gamma_j = 1) \sum_{i \in G_j}\mathbbm{1}(\delta_i=0)}.
\end{align*}
This is of the form of a Beta distribution, so that
\begin{align*}
  \theta_v \sim  \text{Beta}\Biggl(me_1-m + \sum_{j=1}^m\mathbbm{1}(\gamma_j = 1)&\sum_{i \in G_j}\mathbbm{1}(\delta_i=1) + 1,\\
  &me_2-m +\sum_{j=1}^m\mathbbm{1}(\gamma_j = 1) \sum_{i \in G_j}\mathbbm{1}(\delta_i=0) + 1\Biggr).
\end{align*}
\paragraph{5. Group signal strength ratio ($c_g$).}
We have 
\begin{align*}
    c_g \sim \pi(c_g\mid \vect\gamma,\vect\delta,\vect{y},\vect\beta,\sigma,\theta_v,\theta_g,c_v,\hat{\vect{A}}) &= \pi(c_g\mid\vect\gamma,\vect\delta,\vect\beta,\sigma,c_v,\hat{\vect{A}})\\
    &\propto \pi(c_g)\pi(\vect\beta \mid c_g,\vect\gamma,\vect\delta,\vect\beta,\sigma,c_v,\hat{\vect{A}}).
\end{align*}
The first term is uniform, so it can be ignored. To isolate the components of the second term containing $c_g$ we need only consider the case $\gamma_j = 1$, as it is not possible to discriminate between $\vect\delta\in\{0,1\}$ here, so that
\begin{align*}
            \pi(\vect\beta \mid c_g,\vect\gamma,\vect\delta,\vect\beta,\sigma,c_v,\hat{\vect{A}}) &= c_g^{\sum_{j=1}^m p_j\mathbbm{1}\{\gamma_{j} = 1\}}\prod_{i=1}^p \exp \left\{\frac{-1}{\sigma} |\beta_i| c_g c_v \mathbbm{1}(\delta_i=1) v_{r_v(\hat{\vect{A}} \vect\beta, i)}\right\} \\
&\quad\times\prod_{j=1}^m \exp \left\{-\frac{c_g\mathbbm{1}(\gamma_j=1)}{\sigma} \|\vect\kappa^{(j)}\vect\beta^{(j)}\|_2\sqrt{p_j}  w_{r_g(\hat{\vect{A}} \vect\beta, j)}\right\}\\
&=  c_g^{\sum_{j=1}^m p_j\mathbbm{1}(\gamma_j = 1)}\exp\Bigg\{-\frac{c_g}{\sigma}\Bigg[\sum_{i=1}^p c_v|\beta_i| v_{r_v(\hat{\vect{A}} \vect\beta, i)}\mathbbm{1}(\delta_i=1) \\
&\quad+ \sum_{j=1}^m \|\vect\kappa^{(j)}\vect\beta^{(j)}\|_2\sqrt{p_j}w_{r_g(\hat{\vect{A}} \vect\beta, j)}\mathbbm{1}(\gamma_j=1)\Bigg]\Bigg\}.
        \end{align*}
This is of the form of a Gamma distribution, so that, truncated at $[0,1]$,
\begin{align*}
    &c_g \sim  \text{Gamma}\Bigg(1 + \sum_{j=1}^m p_j \mathbbm{1}(\gamma_j = 1), \\
&\frac{1}{\sigma}\Bigg[\sum_{i=1}^p c_v|\beta_i| \alpha v_{r_v(\hat{\vect{A}} \vect\beta, i)}\mathbbm{1}(\delta_i=1)+ \sum_{j=1}^m \|\vect\kappa^{(j)}\vect\beta^{(j)}\|_2\sqrt{p_j} (1-\alpha) w_{r_g(\hat{\vect{A}} \vect\beta, j)}\mathbbm{1}(\gamma_j=1)\Bigg]\Bigg).
\end{align*}
\paragraph{6. Variable signal strength ratio ($c_v$).}
Here, we justify our choice of using MH to sample $c_v$. We would like to sample $c$ from
\begin{align*}
    c_v \sim \pi(c_v\mid\boldsymbol\gamma,\boldsymbol\delta,\mathbf{y},\boldsymbol\beta,\sigma,\theta_v,\theta_g,c_g,\hat{\mathbf{A}}) &= \pi(c_v\mid\boldsymbol\gamma,\boldsymbol\delta,\boldsymbol\beta,\sigma,c_g,\hat{\mathbf{A}}) \\
    &\propto \pi(c_v)\pi(\boldsymbol\beta \mid c_v,\boldsymbol\gamma,\boldsymbol\delta,\boldsymbol\beta,\sigma,c_g,\hat{\mathbf{A}}).
\end{align*}
The first term is uniform, so it vanishes. Therefore, considering only the terms containing $c_v$ in the second component (which occur when $\gamma_j = 1$ and $\delta_i = 1$), we have
        \begin{align}
\pi(\boldsymbol\beta \mid c_v,\boldsymbol\gamma,\boldsymbol\delta,\sigma,c_g,A) &= c_v^{\sum_{i=1}^p\mathbbm{1}\{\boldsymbol\delta_{i} = 1\}}\prod_{i=1}^p \exp \left\{\frac{-c_g}{\sigma} |\beta_i|c_v \mathbbm{1}(\delta_i = 1) v_{r_v(\hat{\mathbf{A}} \boldsymbol\beta, i)}\right\}\nonumber \\
&\quad\times\prod_{j=1}^m \exp \left\{-\frac{c_g\mathbbm{1}(\gamma_j=1)}{\sigma} \|\boldsymbol\kappa^{(j)}\boldsymbol\beta^{(j)}\|_2\sqrt{p_j}  w_{r_g(\hat{\mathbf{A}} \boldsymbol\beta, j)}\right\}\nonumber\\
&=  c_v^{\sum_{i=1}^p \mathbbm{1}(\delta_i = 1)}\exp\Bigg(\sum_{i=1}^p \frac{-c_g}{\sigma}c_v |\beta_i| \mathbbm{1}(\delta_i = 1)v_{r_v(\hat{\mathbf{A}} \boldsymbol\beta, i)}\label{eqn:c_v_mh_form} \\
&\quad+ \sum_{j=1}^m \frac{-c_g\mathbbm{1}(\gamma_j=1)}{\sigma}\|\boldsymbol\kappa^{(j)}\boldsymbol\beta^{(j)}\|_2\sqrt{p_j}w_{r_g(\hat{\mathbf{A}} \boldsymbol\beta, j)}\Bigg).\nonumber
\end{align}
However, this is not of any clear distributional form. The $\boldsymbol\kappa$ term in the $\ell_2$ norm contains $c_v$, and we are unable to extract it from the norm. There are two potential solutions: 1. Approximate $\|\boldsymbol\kappa^{(j)}\boldsymbol\beta^{(j)}\|_2$ as some function of $c_v$, via an approximation such as an affine one. 2. Use MH. We opt for the latter option, using a $\text{Gamma}(2,2)$ proposal, truncated at $[0,1]$ (which is a relatively uninformative prior that favors sparsity). This proposal is chosen as the form of Equation \ref{eqn:c_v_mh_form} would be of a Gamma distribution if we approximated the $\ell_2$ term, and as $c_g$ is also sampled from a Gamma distribution.
\subsubsection{Stochastic and maximization step derivations}\label{appendix:saem_bsgs_opt}
\paragraph{Update for $\sigma$.}
The solution is as for BGSLOPE (Appendix \ref{appendix:saem_bgslope}), but using the SGS penalty instead:
\begin{align*}
\sigma^\text{MLE}_{[t]} &=\frac{1}{2(n+2)} \Biggl[\sum_{i=1}^p|(\beta_i)_{[t]}|(\hat{a}_i)_{[t]} v_{r_v(\hat{\vect{A}}_{[t]} \vect\beta_{[t]}, i)} +\sum_{j=1}^m \|\hat{\vect{A}}^{(j)}_{[t]}\vect\beta_{[t]}^{(j)}\|_2\sqrt{p_j}  w_{r_g(\hat{\vect{A}}_{[t]} \vect\beta_{[t]}, j)}\\
&\quad+ \Biggl(\Biggl(\sum_{i=1}^p|(\beta_i)_{[t]}|(\hat{a}_i)_{[t]} v_{r_v(\hat{\vect{A}}_{[t]} \vect\beta_{[t]}, i)} +\sum_{j=1}^m \|\hat{\vect{A}}^{(j)}_{[t]}\vect\beta_{[t]}^{(j)}\|_2\sqrt{p_j}  w_{r_g(\hat{\vect{A}}_{[t]} \vect\beta_{[t]}, j)}\Biggr)^2\\
&\quad\quad\quad\quad\quad\quad\quad\quad+4(n+2)\|\vect{y}-\vect{X}\vect\beta_{[t]}\|_2^2\Biggr)^{1/2}\Biggr],
\end{align*}
or more simply
\begin{align*}
&\sigma^\text{MLE}_{[t]} = \frac{K_2' + \sqrt{(K_2')^2 + 4K_1'(n+2)}}{2(n+2)},\; \text{where},\\
&K_1'=\|\vect{y}-\vect{X}\vect\beta_{[t]}\|_2^2, \\ 
&K_2'=\sum_{i=1}^p|(\beta_i)_{[t]}|(\hat{a}_i)_{[t]} v_{r_v(\hat{\vect{A}}_{[t]} \vect\beta_{[t]}, i)} +\sum_{j=1}^m \|\hat{\vect{A}}^{(j)}_{[t]}\vect\beta_{[t]}^{(j)}\|_2\sqrt{p_j}  w_{r_g(\hat{\vect{A}}_{[t]} \vect\beta_{[t]}, j)}.
\end{align*}
The classical MLE formula of $\sigma$ is also recovered here if the penalization is removed.
\subsubsection{BSGS-$\alpha$}\label{appendix:bsgs-alpha-posterior}
The updates follow those of BSGS, with an additional update; given $\alpha \sim \mathcal{U}[0,1]$,
\begin{align*}
    &\pi(\alpha \mid \vect\gamma,\vect\delta, \vect{y},\vect\beta,\sigma,\theta_g,\theta_v,c_g,c_v,\hat{\vect{A}})\propto \pi(\alpha) \pi(\vect\beta\mid  \gamma,\delta, \vect\beta,\sigma,c_v,c_g,\hat{\vect{A}})\\
    &\propto \prod_{j=1}^m c_g^{p_j \mathbbm{1}\{\gamma_j = 1\}} c_v^{\sum_{i\in G_j}\mathbbm{1}\{\delta_{i} = 1\}}\prod_{i=1}^p \exp \left\{-\frac{1}{\sigma} |\beta_i|\hat{a}_i \alpha v_{r_v(\hat{\vect{A}} \vect\beta, i)}\right\}\\
&\quad\times \prod_{j=1}^m \exp \left\{-\frac{1}{\sigma} \|\hat{\vect{A}} ^{(j)}\vect\beta^{(j)}\|_2\sqrt{p_j} (1-\alpha) w_{r_g(\hat{\vect{A}}  \vect\beta, j)}\right\}\\
    &\propto  \exp \left\{-\frac{\alpha}{\sigma} \left[\sum_{i=1}^p |\beta_i|\hat{a}_i v_{r_v(\hat{\vect{A}} \vect\beta, i)} -  \sum_{j=1}^m \|\hat{\vect{A}} ^{(j)}\vect\beta^{(j)}\|_2\sqrt{p_j}  w_{r_g(\hat{\vect{A}}  \vect\beta, j)}\right]\right\}.
\end{align*}
 This posterior is in the form of an exponential distribution with rate
 \begin{equation*}
      \lambda =  \frac{1}{\sigma} \left[\sum_{i=1}^p |\beta_i|\hat{a}_i v_{r_v(\hat{\vect{A}} \vect\beta, i)} -  \sum_{j=1}^m \|\hat{\vect{A}}^{(j)}\vect\beta^{(j)}\|_2\sqrt{p_j} w_{r_g(\hat{\vect{A}} \vect\beta, j)}\right].
 \end{equation*} 
 However, the rate can be non-positive, so MH is used with a $\text{Beta}(10,0.5)$ proposal. 
 
\subsubsection{Algorithm}
\begin{algorithm}[H]
   \caption{SAEM for BSGS}
   \label{alg:bsgs_saem_alg}
\begin{algorithmic}[1]
   \STATE {\bfseries Input:} Initial values ($\vect\beta_{[0]}, \sigma_{[0]},(c_g)_{[0]}, (c_v)_{[0]}, (\theta_g)_{[0]},(\theta_v)_{[0]}$), maximum number of iterations $T > 0$, tolerance $\epsilon > 0$
\WHILE{$t \leq T$ \AND $\|\vect\beta_{[t+1]} - \vect\beta_{[t]} \|_2^2> \epsilon$}
\STATE \textbf{Simulation step:}
\STATE Sample $(\gamma_j)_{[t]}$ from Equation \ref{eqn:bsgs_gamma_update} for $j=1,\ldots,m$
\STATE Sample $(\delta_i)_{[t]}$ from Equation \ref{eqn:bsgs_delta_update} for $i=1,\ldots,p$
\STATE Sample $(\theta_g)_{[t]}$ from Equation \ref{eqn:bsgs_theta_g_update}
\STATE Sample $(\theta_v)_{[t]}$ from Equation \ref{eqn:bsgs_theta_v_update}
\STATE Sample $(c_g)_{[t]}$ from Equation \ref{eqn:bsgs_c_g_update}
\STATE Sample $(c_v)_{[t]}$ from Equation \ref{eqn:bsgs_c_v_update}
\STATE \textbf{Stochastic approximation step:}
\STATE Calculate $\vect\beta^\text{MLE}_{[t]}$ via Equation \ref{eqn:sgs_opt_problem}
\STATE Calculate $\sigma_{[t]}^\text{MLE}$ from Equation \ref{eqn:sgs_opt_problem_sigma}
\STATE Update parameters
\begin{equation*}
    \vect\beta_{[t+1]} = \vect\beta_{[t]} + \eta_t(\vect\beta^\text{MLE}_{[t]}-  \vect\beta_{[t]}),\quad \sigma_{[t+1]} = \sigma_{[t]} + \eta_t(\sigma^\text{MLE}_{[t]}-\sigma_{[t]})
\end{equation*}
\STATE Update step size $\eta_t =
\begin{cases}
1, & \text{if } t \leq 20, \\
\frac{1}{t-20}, & \text{if } t > 20
\end{cases}$
\STATE $t = t+1$
\ENDWHILE
 \STATE {\bfseries Output:} $(\vect\beta_{[t+1]}, \vect\delta_{[t]})$, such that $\hat{\vect\delta}_{[t]} \leftarrow \frac{1}{20} \sum_{k=t-19}^{t} \vect\delta_{[k]}$ (if $t<20$ all iterations of $\vect\delta$ are used), giving the fitted coefficients $\hat{\vect\beta}_{[t+1]} = \vect\beta_{[t+1]} \cdot \mathbbm{1}\left(\hat{\vect\delta}_{[t]} > 0.5\right)$
\end{algorithmic}
\end{algorithm}
\subsection{Illustrative example}\label{appendix:illustrative-example}
\begin{figure}[H]
    \centering
\includegraphics[width=1\linewidth]{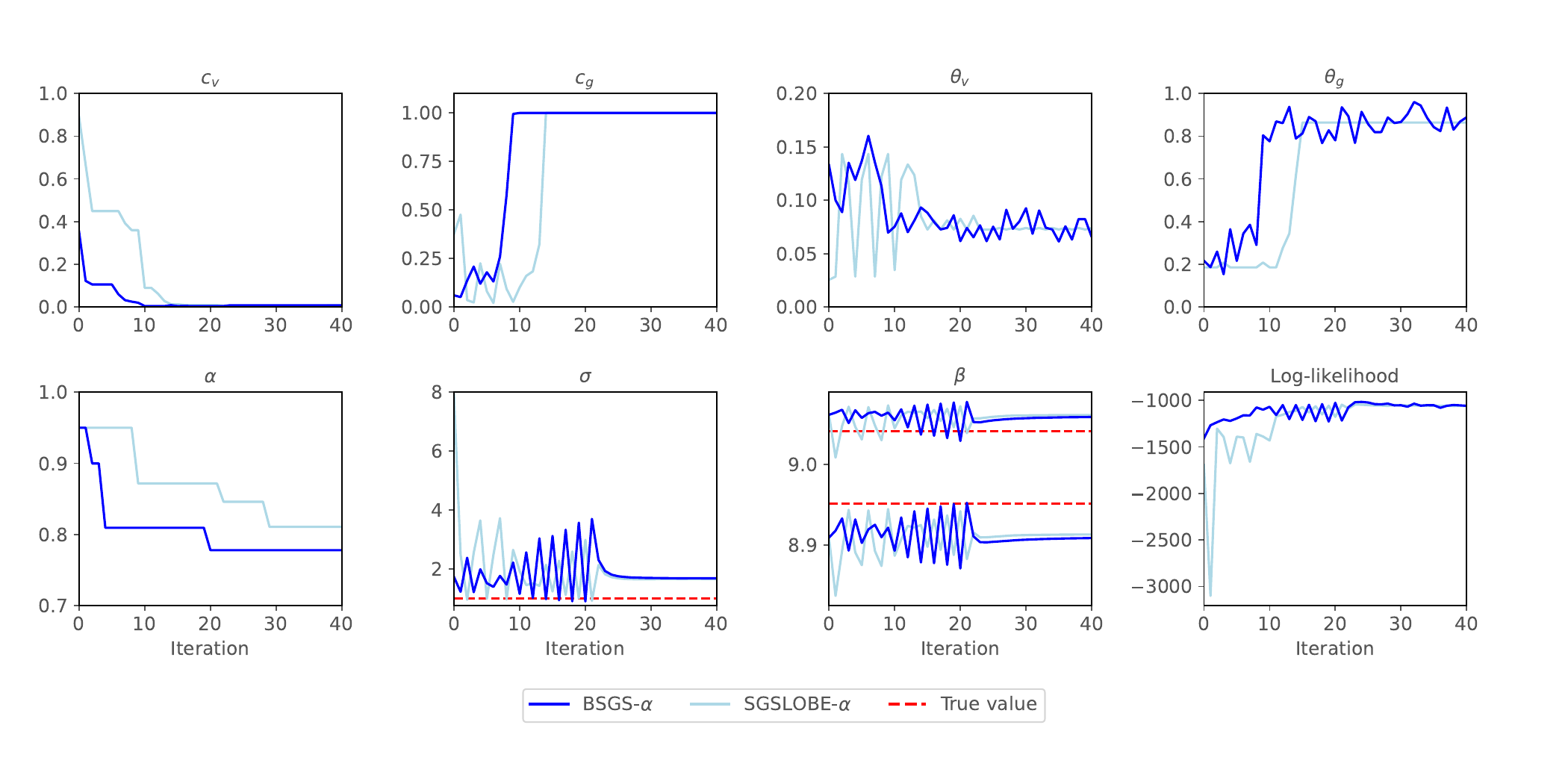}
    \caption{Bayesian latent variables and model parameters for BSGS-$\alpha$ and SGSLOBE-$\alpha$ for the illustrative example (Section \ref{section:illustrative-example}), shown for the first $40$ iterations and the true values for $\sigma,\beta$. Two active $\beta$ values are shown. BSGS-$\alpha$ converged in $202$ iterations and SGSLOBE-$\alpha$ in $350$.}
\label{fig:illustrative-example-alpha}
\end{figure}

\begin{figure}[H]
    \centering
\includegraphics[width=.7\linewidth]{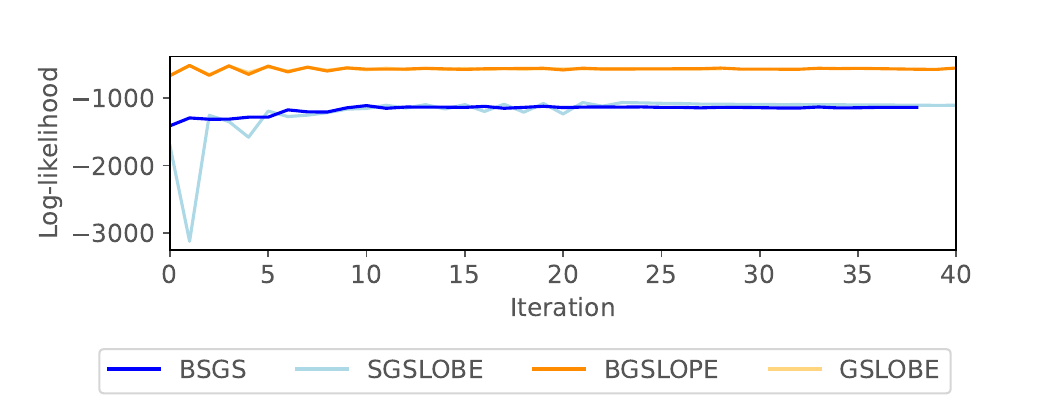}
    \caption{The log-likelihood for BSGS, SGSLOBE, BGSLOPE, and GSLOBE for the illustrative example (Section \ref{section:illustrative-example}).}
\label{fig:illustrative-example-likelihood}
\end{figure}
\subsection{Initializations and sensitivity analysis}\label{appendix:stability-analysis}
\begin{figure}[H]
    \centering
\includegraphics[width=1\linewidth]{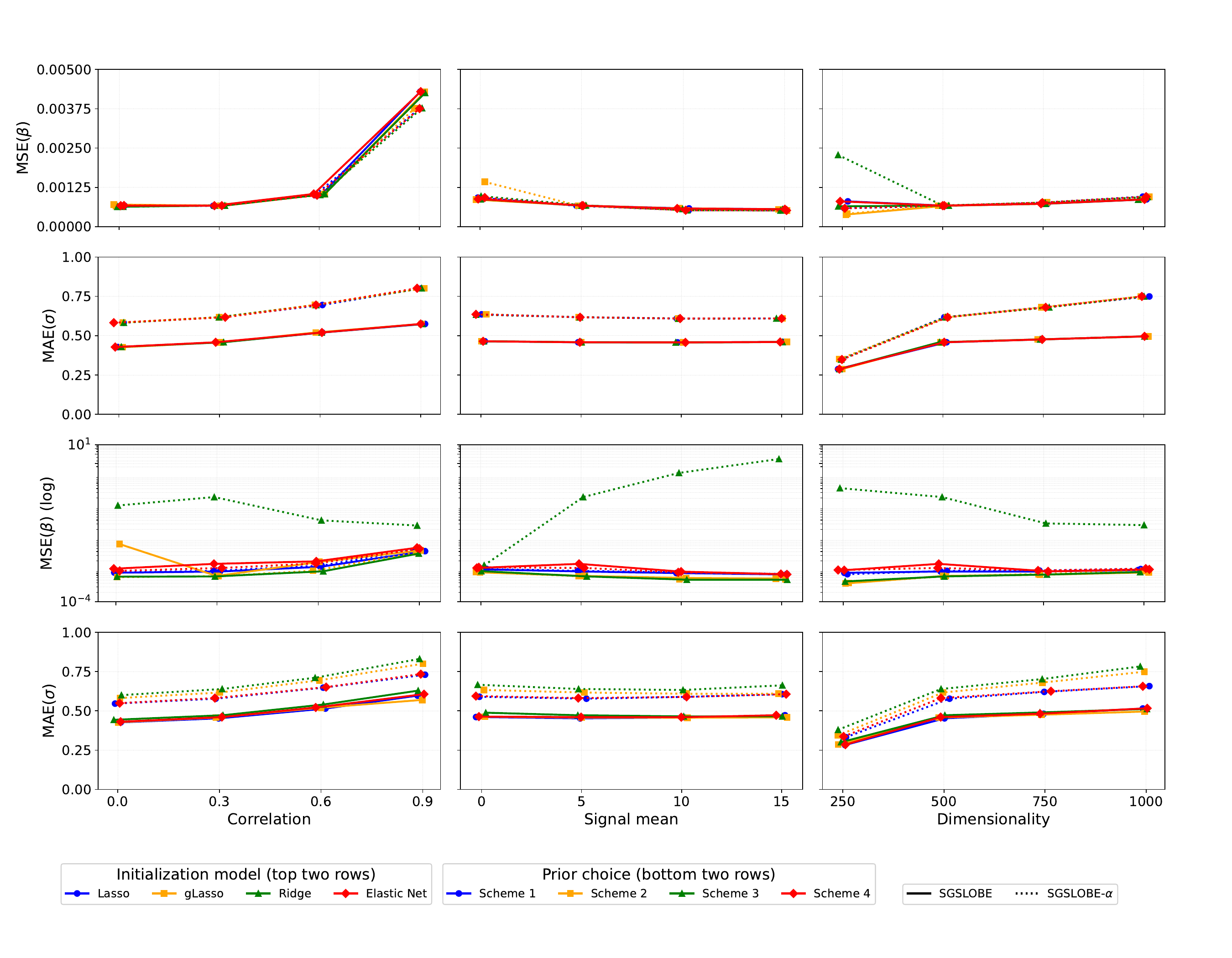}    \caption{MSE($\boldsymbol\beta$) and $\text{MAE}(\sigma)$ for SGSLOBE and SGSLOBE-$\alpha$ under different $\boldsymbol\beta$ initialization models (top two rows) and Beta prior choices (bottom two rows), with a small amount of jitter added to allow the differences to be seen.}
    \label{fig:study-1-init-sgslobe}
\end{figure}
\begin{figure}[H]
    \centering
\includegraphics[width=1\linewidth]{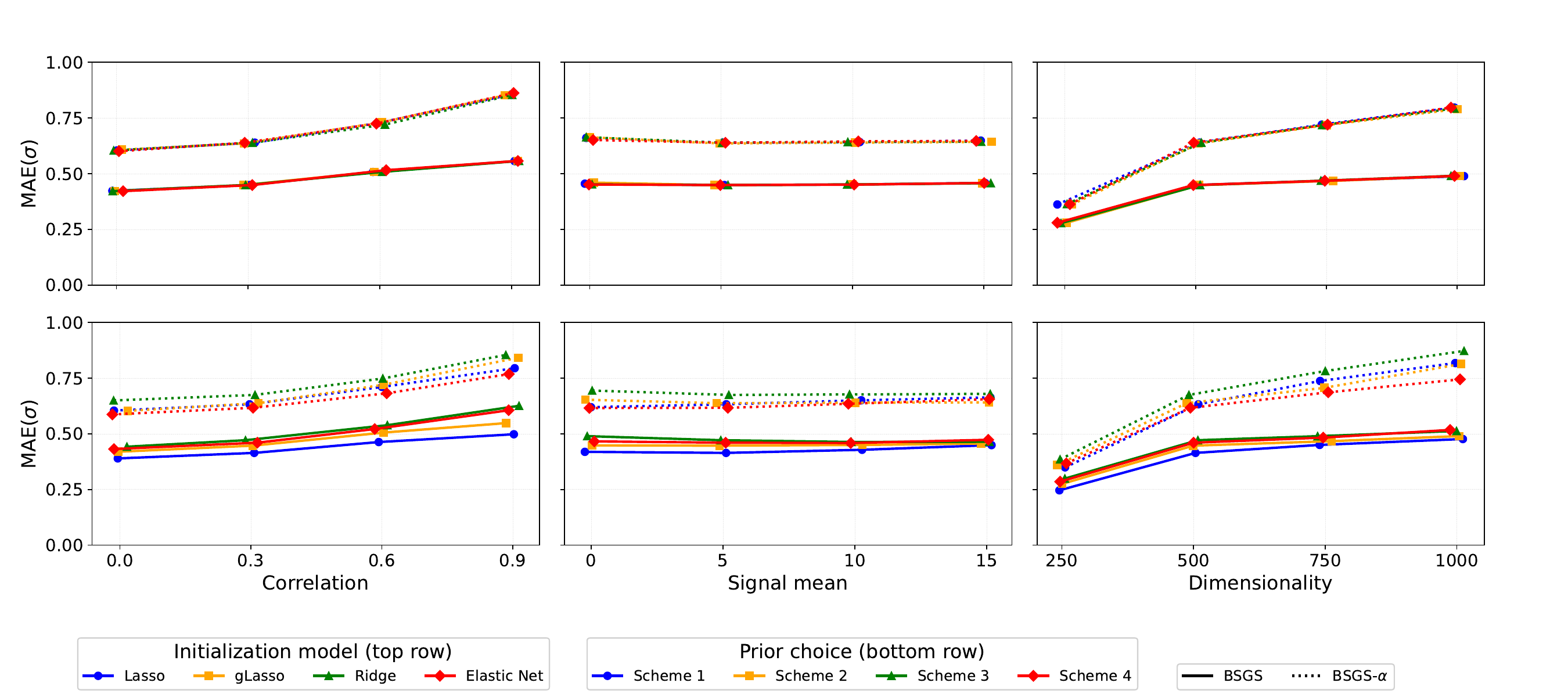}
    \caption{$\text{MAE}(\sigma)$ for BSGS and BSGS-$\alpha$ under different $\boldsymbol\beta$ initialization models (top row) and Beta prior choices (bottom row), with a small amount of jitter added to allow the differences to be seen.}
    \label{fig:study-1-init-main-text-sigma}
\end{figure}
\begin{figure}[H]
    \centering
\includegraphics[width=1\linewidth]{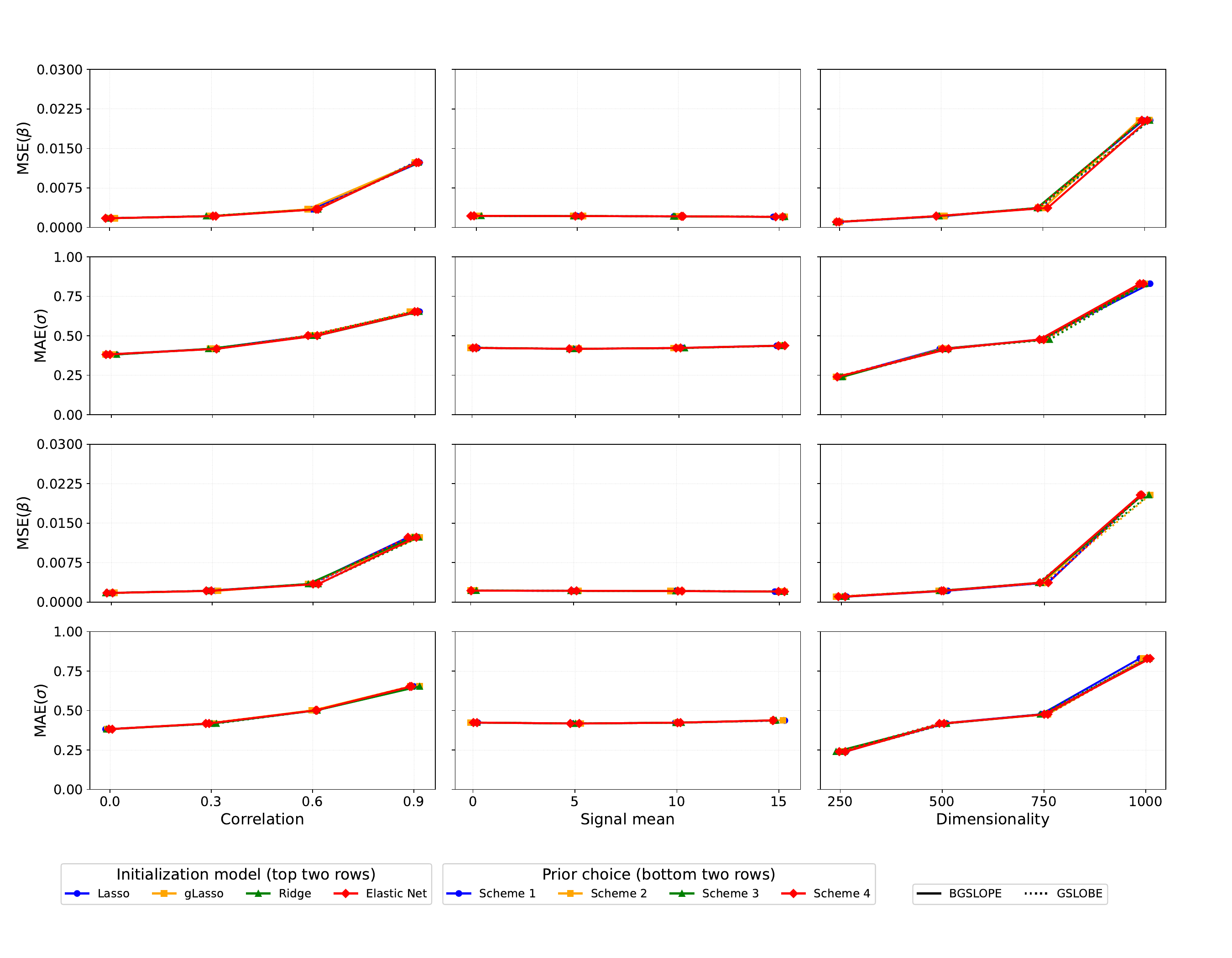}
    \caption{MSE($\boldsymbol\beta$) and $\text{MAE}(\sigma)$ for BGSLOPE and GSLOBE under different $\boldsymbol\beta$ initialization models (top two rows) and Beta prior choices (bottom two rows), with a small amount of jitter added to allow the differences to be seen.}
    \label{fig:study-1-init-bgslope}
\end{figure}

\begin{table}[H]
\centering
\resizebox{\textwidth}{!}{%
\begin{tabular}{lcccccccc}
\toprule
\textbf{Model} & \multicolumn{4}{c}{\textbf{$\vect\beta$ initialization}} & \multicolumn{4}{c}{\textbf{Beta prior choice}} \\
\cmidrule(lr){2-5} \cmidrule(lr){6-9}
 & Lasso & gLasso & Ridge & Elastic net & Scheme 1 & Scheme 2 & Scheme 3 & Scheme 4 \\
\midrule
  BSGS & $\underset{\scriptscriptstyle \textcolor{gray}{(5.8)}}{82.6}$ & $\underset{\scriptscriptstyle \textcolor{gray}{(5.6)}}{\mathbf{80.7}}$ & $\underset{\scriptscriptstyle \textcolor{gray}{(5.4)}}{81.6}$ & $\underset{\scriptscriptstyle \textcolor{gray}{(5.8)}}{82.2}$ & $\underset{\scriptscriptstyle \textcolor{gray}{(12.2)}}{191.4}$ & $\underset{\scriptscriptstyle \textcolor{gray}{(5.6)}}{78.1}$ & $\underset{\scriptscriptstyle \textcolor{gray}{(3.7)}}{\mathbf{73.6}}$ & $\underset{\scriptscriptstyle \textcolor{gray}{(10.1)}}{153.2}$ \\
  SGSLOBE & $\underset{\scriptscriptstyle \textcolor{gray}{(3.9)}}{79.2}$ & $\underset{\scriptscriptstyle \textcolor{gray}{(4.1)}}{78.4}$ & $\underset{\scriptscriptstyle \textcolor{gray}{(4.4)}}{\mathbf{76.4}}$ & $\underset{\scriptscriptstyle \textcolor{gray}{(3.8)}}{79.4}$ & $\underset{\scriptscriptstyle \textcolor{gray}{(5.5)}}{\mathbf{71.3}}$ & $\underset{\scriptscriptstyle \textcolor{gray}{(4.4)}}{78.8}$ & $\underset{\scriptscriptstyle \textcolor{gray}{(2.3)}}{76.8}$ & $\underset{\scriptscriptstyle \textcolor{gray}{(5.3)}}{\mathbf{71.3}}$ \\
BSGS-$\alpha$ & $\underset{\scriptscriptstyle \textcolor{gray}{(15.2)}}{124.3}$ & $\underset{\scriptscriptstyle \textcolor{gray}{(14.8)}}{123.0}$ & $\underset{\scriptscriptstyle \textcolor{gray}{(15.7)}}{\mathbf{121.7}}$ & $\underset{\scriptscriptstyle \textcolor{gray}{(15.3)}}{123.8}$ & $\underset{\scriptscriptstyle \textcolor{gray}{(17.8)}}{258.3}$ & $\underset{\scriptscriptstyle \textcolor{gray}{(15.4)}}{\mathbf{120.2}}$ & $\underset{\scriptscriptstyle \textcolor{gray}{(16.2)}}{128.7}$ & $\underset{\scriptscriptstyle \textcolor{gray}{(17.8)}}{257.1}$ \\
  SGSLOBE-$\alpha$ & $\underset{\scriptscriptstyle \textcolor{gray}{(2.1)}}{\mathbf{133.1}}$ & $\underset{\scriptscriptstyle \textcolor{gray}{(2.3)}}{133.7}$ & $\underset{\scriptscriptstyle \textcolor{gray}{(2.0)}}{133.4}$ & $\underset{\scriptscriptstyle \textcolor{gray}{(2.1)}}{\mathbf{133.1}}$ & $\underset{\scriptscriptstyle \textcolor{gray}{(8.2)}}{141.9}$ & $\underset{\scriptscriptstyle \textcolor{gray}{(2.0)}}{\mathbf{133.1}}$ & $\underset{\scriptscriptstyle \textcolor{gray}{(3.4)}}{136.8}$ & $\underset{\scriptscriptstyle \textcolor{gray}{(8.2)}}{141.4}$ \\
 BGSLOPE & $\underset{\scriptscriptstyle \textcolor{gray}{(10.0)}}{155.0}$ & $\underset{\scriptscriptstyle \textcolor{gray}{(10.0)}}{\mathbf{154.3}}$ & $\underset{\scriptscriptstyle \textcolor{gray}{(9.7)}}{158.6}$ & $\underset{\scriptscriptstyle \textcolor{gray}{(9.9)}}{\mathbf{154.3}}$ & $\underset{\scriptscriptstyle \textcolor{gray}{(10.0)}}{154.9}$ & $\underset{\scriptscriptstyle \textcolor{gray}{(10.0)}}{155.5}$ & $\underset{\scriptscriptstyle \textcolor{gray}{(10.0)}}{155.7}$ & $\underset{\scriptscriptstyle \textcolor{gray}{(9.8)}}{\mathbf{145.1}}$ \\
 GSLOBE & $\underset{\scriptscriptstyle \textcolor{gray}{(5.7)}}{56.4}$ & $\underset{\scriptscriptstyle \textcolor{gray}{(6.2)}}{57.0}$ & $\underset{\scriptscriptstyle \textcolor{gray}{(5.4)}}{60.0}$ & $\underset{\scriptscriptstyle \textcolor{gray}{(5.7)}}{\mathbf{56.2}}$ & $\underset{\scriptscriptstyle \textcolor{gray}{(5.7)}}{\mathbf{56.5}}$ & $\underset{\scriptscriptstyle \textcolor{gray}{(5.7)}}{\mathbf{56.5}}$ & $\underset{\scriptscriptstyle \textcolor{gray}{(5.7)}}{56.7}$ & $\underset{\scriptscriptstyle \textcolor{gray}{(5.7)}}{\mathbf{56.5}}$ \\
\bottomrule
\end{tabular}
}
\caption[Mean number of iterations for the $\vect\beta$ initializations]{Mean number of iterations for the $\vect\beta$ initializations and Beta prior schemes, with standard errors shown in \textcolor{gray}{grey}.}
\label{tbl:init-study-mean-num-it}
\end{table}
\begin{table}[H]
\centering
\resizebox{\textwidth}{!}{%
\begin{tabular}{lcccccccc}
\toprule
\textbf{Model} & \multicolumn{4}{c}{\textbf{MSE($\vect\beta$)}} & \multicolumn{4}{c}{\textbf{MAE($\sigma$)}} \\
\cmidrule(lr){2-5} \cmidrule(lr){6-9}
               & Lasso & gLasso & Ridge & Elastic net & Lasso & gLasso & Ridge & Elastic net \\
\midrule
BSGS & $\underset{\scriptscriptstyle \textcolor{gray}{(0.0001)}}{\mathbf{0.0010}}$ & $\underset{\scriptscriptstyle \textcolor{gray}{(0.0001)}}{\mathbf{0.0010}}$ & $\underset{\scriptscriptstyle \textcolor{gray}{(0.0043)}}{0.0018}$ & $\underset{\scriptscriptstyle \textcolor{gray}{(0.0036)}}{0.0016}$ & $\underset{\scriptscriptstyle \textcolor{gray}{(0.0066)}}{0.4527}$ & $\underset{\scriptscriptstyle \textcolor{gray}{(0.0075)}}{0.4532}$ & $\underset{\scriptscriptstyle \textcolor{gray}{(0.0064)}}{\mathbf{0.4524}}$ & $\underset{\scriptscriptstyle \textcolor{gray}{(0.0076)}}{0.4531}$ \\
  SGSLOBE & $\underset{\scriptscriptstyle \textcolor{gray}{(0.0002)}}{\mathbf{0.0010}}$ & $\underset{\scriptscriptstyle \textcolor{gray}{(0.0001)}}{\mathbf{0.0010}}$ & $\underset{\scriptscriptstyle \textcolor{gray}{(0.0001)}}{\mathbf{0.0010}}$ & $\underset{\scriptscriptstyle \textcolor{gray}{(0.0002)}}{\mathbf{0.0010}}$ & $\underset{\scriptscriptstyle \textcolor{gray}{(0.0066)}}{\mathbf{0.4613}}$ & $\underset{\scriptscriptstyle \textcolor{gray}{(0.0066)}}{\mathbf{0.4613}}$ & $\underset{\scriptscriptstyle \textcolor{gray}{(0.0065)}}{0.4615}$ & $\underset{\scriptscriptstyle \textcolor{gray}{(0.0066)}}{0.4614}$ \\
BSGS-$\alpha$ & $\underset{\scriptscriptstyle \textcolor{gray}{(0.1814)}}{0.1120}$ & $\underset{\scriptscriptstyle \textcolor{gray}{(0.1240)}}{\mathbf{0.0503}}$ & $\underset{\scriptscriptstyle \textcolor{gray}{(0.2214)}}{0.1661}$ & $\underset{\scriptscriptstyle \textcolor{gray}{(0.1773)}}{0.1020}$ & $\underset{\scriptscriptstyle \textcolor{gray}{(0.0120)}}{0.6611}$ & $\underset{\scriptscriptstyle \textcolor{gray}{(0.0123)}}{\mathbf{0.6593}}$ & $\underset{\scriptscriptstyle \textcolor{gray}{(0.0126)}}{0.6598}$ & $\underset{\scriptscriptstyle \textcolor{gray}{(0.0125)}}{0.6605}$ \\
SGSLOBE-$\alpha$ & $\underset{\scriptscriptstyle \textcolor{gray}{(0.0001)}}{\mathbf{0.0010}}$ & $\underset{\scriptscriptstyle \textcolor{gray}{(0.0003)}}{\mathbf{0.0010}}$ & $\underset{\scriptscriptstyle \textcolor{gray}{(0.0008)}}{0.0011}$ & $\underset{\scriptscriptstyle \textcolor{gray}{(0.0001)}}{\mathbf{0.0010}}$ & $\underset{\scriptscriptstyle \textcolor{gray}{(0.0083)}}{0.6302}$ & $\underset{\scriptscriptstyle \textcolor{gray}{(0.0083)}}{0.6306}$ & $\underset{\scriptscriptstyle \textcolor{gray}{(0.0082)}}{\mathbf{0.6301}}$ & $\underset{\scriptscriptstyle \textcolor{gray}{(0.0083)}}{0.6302}$ \\
BGSLOPE & $\underset{\scriptscriptstyle \textcolor{gray}{(0.0016)}}{\mathbf{0.0046}}$ & $\underset{\scriptscriptstyle \textcolor{gray}{(0.0016)}}{\mathbf{0.0046}}$ & $\underset{\scriptscriptstyle \textcolor{gray}{(0.0016)}}{\mathbf{0.0046}}$ & $\underset{\scriptscriptstyle \textcolor{gray}{(0.0016)}}{\mathbf{0.0046}}$ & $\underset{\scriptscriptstyle \textcolor{gray}{(0.0212)}}{0.4680}$ & $\underset{\scriptscriptstyle \textcolor{gray}{(0.0212)}}{0.4681}$ & $\underset{\scriptscriptstyle \textcolor{gray}{(0.0212)}}{\mathbf{0.4679}}$ & $\underset{\scriptscriptstyle \textcolor{gray}{(0.0212)}}{0.4680}$ \\
GSLOBE & $\underset{\scriptscriptstyle \textcolor{gray}{(0.0016)}}{\mathbf{0.0046}}$ & $\underset{\scriptscriptstyle \textcolor{gray}{(0.0016)}}{\mathbf{0.0046}}$ & $\underset{\scriptscriptstyle \textcolor{gray}{(0.0016)}}{\mathbf{0.0046}}$ & $\underset{\scriptscriptstyle \textcolor{gray}{(0.0016)}}{\mathbf{0.0046}}$ & $\underset{\scriptscriptstyle \textcolor{gray}{(0.0212)}}{0.4681}$ & $\underset{\scriptscriptstyle \textcolor{gray}{(0.0212)}}{0.4681}$ & $\underset{\scriptscriptstyle \textcolor{gray}{(0.0212)}}{\mathbf{0.4680}}$ & $\underset{\scriptscriptstyle \textcolor{gray}{(0.0212)}}{0.4681}$ \\
\bottomrule
\end{tabular}
}
\caption[Mean MSE($\vect\beta$) and MAE($\sigma$) values for the initialization schemes]{Mean MSE($\vect\beta$) and MAE($\sigma$) values for the initialization schemes, with the best scheme for each model in \textbf{bold} and standard errors shown in \textcolor{gray}{grey}.}
\label{tbl:init-study-beta-lasso}   
\end{table}
\begin{table}[H]
\centering
\resizebox{\textwidth}{!}{%
\begin{tabular}{lcccccccc}
\toprule
\textbf{Model} & \multicolumn{4}{c}{\textbf{MSE($\vect\beta$)}} & \multicolumn{4}{c}{\textbf{MAE($\sigma$)}} \\
\cmidrule(lr){2-5} \cmidrule(lr){6-9}
              & Scheme 1 & Scheme 2 & Scheme 3 & Scheme 4  & Scheme 1 & Scheme 2 & Scheme 3 & Scheme 4  \\
\midrule
 BSGS & $\underset{\scriptscriptstyle \textcolor{gray}{(0.0001)}}{0.0014}$ & $\underset{\scriptscriptstyle \textcolor{gray}{(0.0001)}}{0.0010}$ & $\underset{\scriptscriptstyle \textcolor{gray}{(0.0001)}}{\mathbf{0.0009}}$ & $\underset{\scriptscriptstyle \textcolor{gray}{(0.0004)}}{0.0015}$ & $\underset{\scriptscriptstyle \textcolor{gray}{(0.0077)}}{\mathbf{0.4219}}$ & $\underset{\scriptscriptstyle \textcolor{gray}{(0.0063)}}{0.4500}$ & $\underset{\scriptscriptstyle \textcolor{gray}{(0.0069)}}{0.4782}$ & $\underset{\scriptscriptstyle \textcolor{gray}{(0.0077)}}{0.4686}$ \\
SGSLOBE & $\underset{\scriptscriptstyle \textcolor{gray}{(0.0001)}}{0.0012}$ & $\underset{\scriptscriptstyle \textcolor{gray}{(0.0029)}}{0.0016}$ & $\underset{\scriptscriptstyle \textcolor{gray}{(0.0001)}}{\mathbf{0.0009}}$ & $\underset{\scriptscriptstyle \textcolor{gray}{(0.0004)}}{0.0016}$ & $\underset{\scriptscriptstyle \textcolor{gray}{(0.0065)}}{0.4639}$ & $\underset{\scriptscriptstyle \textcolor{gray}{(0.0066)}}{\mathbf{0.4605}}$ & $\underset{\scriptscriptstyle \textcolor{gray}{(0.0070)}}{0.4792}$ & $\underset{\scriptscriptstyle \textcolor{gray}{(0.0080)}}{0.4683}$ \\
  BSGS-$\alpha$ & $\underset{\scriptscriptstyle \textcolor{gray}{(0.0001)}}{\mathbf{0.0010}}$ & $\underset{\scriptscriptstyle \textcolor{gray}{(0.1023)}}{0.0331}$ & $\underset{\scriptscriptstyle \textcolor{gray}{(0.4257)}}{1.5796}$ & $\underset{\scriptscriptstyle \textcolor{gray}{(0.0001)}}{0.0014}$ & $\underset{\scriptscriptstyle \textcolor{gray}{(0.0111)}}{0.6537}$ & $\underset{\scriptscriptstyle \textcolor{gray}{(0.0134)}}{0.6576}$ & $\underset{\scriptscriptstyle \textcolor{gray}{(0.0143)}}{0.6971}$ & $\underset{\scriptscriptstyle \textcolor{gray}{(0.0087)}}{\mathbf{0.6327}}$ \\
SGSLOBE-$\alpha$ & $\underset{\scriptscriptstyle \textcolor{gray}{(0.0001)}}{0.0013}$ & $\underset{\scriptscriptstyle \textcolor{gray}{(0.0001)}}{\mathbf{0.0009}}$ & $\underset{\scriptscriptstyle \textcolor{gray}{(0.3309)}}{0.5055}$ & $\underset{\scriptscriptstyle \textcolor{gray}{(0.0003)}}{0.0014}$ & $\underset{\scriptscriptstyle \textcolor{gray}{(0.0074)}}{\mathbf{0.5877}}$ & $\underset{\scriptscriptstyle \textcolor{gray}{(0.0081)}}{0.6292}$ & $\underset{\scriptscriptstyle \textcolor{gray}{(0.0128)}}{0.6572}$ & $\underset{\scriptscriptstyle \textcolor{gray}{(0.0085)}}{0.5905}$ \\
 BGSLOPE & $\underset{\scriptscriptstyle \textcolor{gray}{(0.0016)}}{\mathbf{0.0046}}$ & $\underset{\scriptscriptstyle \textcolor{gray}{(0.0016)}}{\mathbf{0.0046}}$ & $\underset{\scriptscriptstyle \textcolor{gray}{(0.0016)}}{\mathbf{0.0046}}$ & $\underset{\scriptscriptstyle \textcolor{gray}{(0.0016)}}{\mathbf{0.0046}}$ & $\underset{\scriptscriptstyle \textcolor{gray}{(0.0212)}}{\mathbf{0.4681}}$ & $\underset{\scriptscriptstyle \textcolor{gray}{(0.0212)}}{0.4683}$ & $\underset{\scriptscriptstyle \textcolor{gray}{(0.0212)}}{0.4684}$ & $\underset{\scriptscriptstyle \textcolor{gray}{(0.0212)}}{0.4683}$ \\
GSLOBE & $\underset{\scriptscriptstyle \textcolor{gray}{(0.0016)}}{\mathbf{0.0046}}$ & $\underset{\scriptscriptstyle \textcolor{gray}{(0.0016)}}{\mathbf{0.0046}}$ & $\underset{\scriptscriptstyle \textcolor{gray}{(0.0016)}}{\mathbf{0.0046}}$ & $\underset{\scriptscriptstyle \textcolor{gray}{(0.0016)}}{\mathbf{0.0046}}$ & $\underset{\scriptscriptstyle \textcolor{gray}{(0.0212)}}{0.4681}$ & $\underset{\scriptscriptstyle \textcolor{gray}{(0.0212)}}{0.4681}$ & $\underset{\scriptscriptstyle \textcolor{gray}{(0.0213)}}{0.4682}$ & $\underset{\scriptscriptstyle \textcolor{gray}{(0.0212)}}{\mathbf{0.4679}}$ \\
 
\bottomrule
\end{tabular}
}
\caption[Mean MSE($\vect\beta)$ and MAE($\sigma$) values for the Beta prior schemes]{Mean MSE($\vect\beta)$ and MAE($\sigma$) values for the Beta prior schemes, with the best prior scheme for each model in \textbf{bold} and standard errors shown in \textcolor{gray}{grey}.}
\label{tbl:init-study-beta}
\end{table}
\begin{figure}[H]
    \centering
\includegraphics[width=1\linewidth]{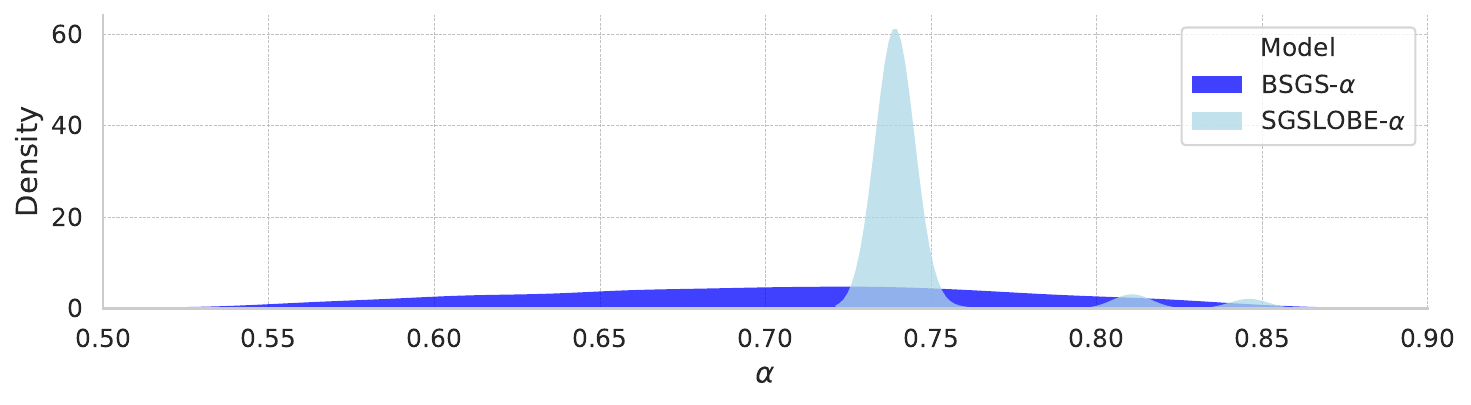}
    \caption{Distribution of learned $\alpha$ values for BSGS-$\alpha$ and SGSLOBE-$\alpha$ in the simulations from Section \ref{section:sensitivity_analysis}.}
    \label{fig:alpha-hist}
\end{figure}
\section{Other approaches}\label{section:other_approaches}
Here, we provide detailed descriptions of the key approaches to model selection for SLOPE models used in our comparison study, expanding on the brief overview presented in Section \ref{section:model_selection_approaches}. 
\subsection{Summary of all approaches}\label{appendix:summary_all_models}
\begin{table}[H]
\centering
\resizebox{\textwidth}{!}{%
\begin{tabular}{llp{6cm}r}
\toprule
\textbf{Acronym} & \textbf{Full name} & \textbf{Description} & \textbf{Section} \\
\arrayrulecolor{gray!25}\midrule

\multicolumn{4}{l}{\textbf{Our Bayesian proposals}} \\  
\specialrule{0.3pt}{0pt}{0pt}

\textcolor{blue!70!black}{BGSLOPE} & Bayesian Group SLOPE & One of the main proposals. & \ref{section:bgslope} \\
\specialrule{0.3pt}{0pt}{0pt}
\textcolor{blue!70!black}{GSLOBE} & Group SLOBE & Using the conditional mean for the updates in BGSLOPE. & \ref{section:slobe} \\
\specialrule{0.3pt}{0pt}{0pt}
\textcolor{blue!70!black}{BSGS} & Bayesian Sparse-group SLOPE & One of the main proposals. & \ref{section:bsgs} \\
\specialrule{0.3pt}{0pt}{0pt}
\textcolor{blue!70!black}{SGSLOBE} & Sparse-group SLOBE & Using the conditional mean for the updates in BSGS. & \ref{section:slobe} \\
\specialrule{0.3pt}{0pt}{0pt}
\textcolor{blue!70!black}{BSGS-$\alpha$} & -- & BSGS with $\alpha$ estimated via SAEM. & \ref{section:sgs_saem_algorithm} \\
\specialrule{0.3pt}{0pt}{0pt}
\textcolor{blue!70!black}{SGSLOBE-$\alpha$} & -- & SGSLOBE with $\alpha$ estimated via SAEM. & \ref{section:slobe} \\
\specialrule{0.3pt}{0pt}{0pt}

\multicolumn{4}{l}{\textbf{Other Bayesian methods}} \\  
\specialrule{0.3pt}{0pt}{0pt}

ABSLOPE & Adaptive Bayesian SLOPE & The spike-and-slab Bayesian SLOPE model. & \ref{section:abslope} \\
\specialrule{0.3pt}{0pt}{0pt}
SLOBE & -- & Using the conditional mean for the updates in ABSLOPE. & \ref{section:slobe} \\
\specialrule{0.3pt}{0pt}{0pt}
\multicolumn{4}{l}{\textbf{Other methods}} \\  
\specialrule{0.3pt}{0pt}{0pt}
AS-SGS & Adaptively Scaled SGS& Adaptive version of scaled regression, applicable only to SGS. & \ref{section:as_sgs} \\
\specialrule{0.3pt}{0pt}{0pt}
CV & Cross-validation &Picking the optimal $\lambda$ using cross-validation, chosen as the 1se model. & -- \\
\specialrule{0.3pt}{0pt}{0pt}
Knockoffs & -- &The Knockoffs filter applied using SLOPE models. & \ref{section:knockoffs} \\
\specialrule{0.3pt}{0pt}{0pt}
Oracle & -- &Fitting a SLOPE model with $\lambda$ equal to the (oracle) noise level. & --\\
\specialrule{0.3pt}{0pt}{0pt}
Scaled regression & -- &Iteratively estimates the noise and regression coefficients. & \ref{section:scaled_regression}\\
\specialrule{0.3pt}{0pt}{0pt}
\textcolor{blue!70!black}{TSO} & Two-step Orthogonal &Uses lasso and Gram–Schmidt to obtain an orthogonal low-dimensional input, then applies SLOPE.& \ref{section:tso}\\
\specialrule{0.3pt}{0pt}{0pt}
TS-SLOPE & Two-step SLOPE & Applies SLOPE twice using the proximal operator. & \ref{section:two-step}\\
\arrayrulecolor{black}\bottomrule
\end{tabular}
}
\caption[Summary of approaches used in synthetic study]{Summary of the approaches used in the synthetic study. Our proposals are highlighted in \textcolor{blue!70!black}{blue}.}
\label{tbl:fp_model_summary}
\end{table}
\subsection{Scaled regression }\label{section:scaled_regression}
\textit{Scaled Regression} jointly estimates the coefficients and noise in a penalized regression model. It uses unbiased estimators when $n > p$ and iterative procedures when $p \geq n$ \citep{Sun2012ScaledRegression}. Originally introduced for the lasso, it was extended to SLOPE in \citet{Bogdan2015SLOPEAdaptiveOptimization}. The noise is calculated as
    \begin{equation*}
\hat{\sigma}^2 = \frac{\| \vect{y} - \vect{X} \hat{\vect{\beta}}(\lambda) \|_2^2}{n - |\hat{S}_v^\lambda|},      
    \end{equation*}
where $|\hat{S}_v^\lambda|$ is the number of non-zero elements in $\hat{\vect{\beta}}(\lambda)$ (the degrees of freedom). The regression coefficients are calculated by fitting the penalized regression model with regularization parameter $\lambda = \hat\sigma/n$ (the $1/n$ factor is from the loss function). This process is repeated until $\hat{S}_v^\lambda$ stabilizes between iterations. The procedure is described fully in Algorithm \ref{alg:scaled_regression}. The algorithm can be adapted to gSLOPE and SGS by using their fitted regression coefficients in place of the SLOPE estimates in step 5.

\begin{algorithm}[H]
\caption{Scaled regression for SLOPE models \citep{Bogdan2015SLOPEAdaptiveOptimization}}
\begin{algorithmic}[1]
\REQUIRE $\vect{y}$, $\vect{X}$
\STATE Initialize: $\hat{S}^+_v \gets \emptyset$
\WHILE{$\hat{S}^+_v \neq \hat{S}_v$}
    \STATE $\hat{S}_v \gets \hat{S}_v^+$
    \STATE Set $\hat{\sigma}^2 = \frac{\| \vect{y} - \vect{X} \hat{\vect{\beta}}(\lambda) \|_2^2}{n - |\hat{S}_v^\lambda|}$
    \STATE Compute the solution $\hat{\vect{\beta}}$ via a SLOPE model with $\lambda = \hat\sigma/n$
    \STATE Set $\hat{S}_v^+ \gets \mathrm{supp}(\hat{\vect{\beta}})$
\ENDWHILE
\end{algorithmic}
\label{alg:scaled_regression}
\end{algorithm}

\subsubsection{AS-SGS}\label{section:as_sgs}
Theorem 1 in \citet{Feser2023Sparse-groupFDR-control} (which shows FDR control for SGS), assumed, without loss of generality, that $\lambda=1$ to derive the weights given in Equations \ref{eqn:sgs_var_pen_max_2} and \ref{eqn:sgs_grp_pen_max_2}. If no such assumption is made, the weights are derived as:
\begin{eqnarray}
	&v_i^\text{max}(\lambda) &= \max_{j=1,\dots,m} \left\{\frac{\Phi^{-1}\left(1-\frac{q_vi}{2p}\right) -   \frac{1}{3}(1-\alpha) \lambda a_j w_j}{\alpha \lambda}\right\}, \; i= 1,\dots,p, \label{eqn:sgs_var_pen_max_2}\\
	&w_i^\text{max}(\lambda) &=\max_{j=1,\dots,m}\left\{\frac{F^{-1}_\text{FN}\left(1-\frac{q_gi}{m}\right)-\alpha \lambda \sum_{k \in G_j}v_k }{(1-\alpha)\lambda p_j}\right\}, \; i = 1,\dots,m.\label{eqn:sgs_grp_pen_max_2}
\end{eqnarray}
Therefore, the scaled regression approach is modified to form \textit{Adaptively Scaled SGS} (AS-SGS) \citep{Feser2023Sparse-groupFDR-control}, so that the penalty sequences are updated each time a new value of $\lambda = \hat\sigma$ is estimated, making the sequences adaptive to the noise level (described in Algorithm \ref{alg:as_sgs}).

The AS-SGS model can be simplified by denoting 
\begin{eqnarray}
	&\tilde{v}_i^\text{max} &= \max_{j=1,\dots,m} \left\{\Phi^{-1}\left(1-\frac{q_vi}{2p}\right) -   \frac{1}{3}(1-\alpha) \lambda a_j w_j\right\}, \; i= 1,\dots,p, \\
	&\tilde{w}_i^\text{max}&=\max_{j=1,\dots,m}\left\{\frac{F^{-1}_\text{FN}\left(1-\frac{q_gi}{m}\right)-\alpha \lambda \sum_{k \in G_j}v_k }{p_j}\right\}, \; i = 1,\dots,m,
\end{eqnarray}
leading to the penalty (via cancellation of $\alpha$ and $\lambda$)
\begin{equation*}
     \sum_{i=1}^p \tilde{v}_i |b|_{(i)}  + \sum_{j=1}^m \|\boldsymbol{b}^{(j)}\|_2 \sqrt{p_j} \tilde{w}_j.
\end{equation*}
\paragraph{SLOPE and gSLOPE.}
For SLOPE, if we similarly remove the assumption $\lambda = 1$, the penalty sequences are given by $v_i(\lambda) = \frac{1}{\lambda}\Phi^{-1}\left(1-\frac{i q_v}{2p}\right).$
However, as the tuning parameter enters the penalty sequences in a linear term, it cancels out so that
\begin{equation}
   J_\text{slope}(\boldsymbol{\beta}; \mathbf{v}) = \lambda \sum_{i=1}^p \frac{1}{\lambda}v_i |\beta|_{(i)} = \sum_{i=1}^p v_i |\beta|_{(i)},
\end{equation} 
which reduces to the SLOPE penalty under $\lambda = 1$. The same is true for gSLOPE, and as a result, adaptively scaled regression is only applicable to SGS.
\begin{algorithm}[H]
\caption{Adaptively Scaled SGS (AS-SGS) \citep{Feser2023Sparse-groupFDR-control}}
\begin{algorithmic}[1]
\REQUIRE $\mathbf{y}$, $\mathbf{X}$
\STATE Initialize: $\hat{S}_v^+ \gets \emptyset$
\WHILE{$\hat{S}^+_v \neq \hat{S}_v$}
    \STATE Set $\hat{S}_v \gets \hat{S}_v^+$
    \STATE Set $\hat{\sigma}^2 = \frac{\| \mathbf{y} - \mathbf{X} \hat{\boldsymbol{\beta}}(\lambda) \|_2^2}{n - |\hat{S}_v^\lambda|}$
    \STATE Compute $\mathbf{v}^{\max}(\lambda)$ and $\mathbf{w}^{\max}(\lambda)$ (Equations \ref{eqn:sgs_var_pen_max_2} and \ref{eqn:sgs_grp_pen_max_2})
    \STATE Compute $\hat{\boldsymbol\beta}$ using SGS with $\lambda = \hat\sigma/n$, $\mathbf{v}^{\max}(\lambda)$, and $\mathbf{w}^{\max}(\lambda)$
    \STATE Set $\hat{S}_v^+ \gets \mathrm{supp}(\hat{\boldsymbol{\beta}})$
\ENDWHILE
\end{algorithmic}
\label{alg:as_sgs}
\end{algorithm}
\subsection{Knockoffs}\label{section:knockoffs}
\textit{Knockoffs} has been gaining popularity in the multiple testing literature, due to its flexibility and powerful properties. Knockoffs provides FDR control by generating knockoff copies of variables that act as negative controls. By comparing each variable to its knockoff, the procedure approximates the number of false positives and retains only variables that show a clear advantage over their knockoff counterparts.

The original Knockoffs procedure was limited to low-dimensional settings \citep{Barber2015ControllingKnockoffs}. \citet{Barber2019AInference} extended it to high-dimensional data using the fixed-X paradigm, which was later generalized to the model-X framework \citep{Candes2018}. This framework also allows for a Bayesian construction, which the authors briefly explored.

\paragraph{Construction.}
Formally, a set of \textit{knockoff variables} $\vect{X}_\text{KO} = (\vect{X}_\text{KO})_1, \ldots, (\vect{X}_\text{KO})_p$ for random variables $\vect{X} = \vect{X}_1, \ldots, \vect{X}_p$ satisfy \citep{Candes2018}:
\begin{enumerate}
\item \textit{Exchangeability}: $(\vect{X}, \vect{X}_\text{KO})_{\operatorname{swap}(M)} \overset{d}{=} (\vect{X}, \vect{X}_\text{KO})$ for any $M \subset {1,\ldots,p}$,
\item \textit{Independence}: $\vect{X}_\text{KO} \perp\!\!\!\perp  \vect{y} \mid \vect{X}$,
\end{enumerate}
where $(\vect{X}, \vect{X}_\text{KO})_{\operatorname{swap}(M)}$ denotes the matrix obtained by swapping $\vect{X}_i$ and $(\vect{X}_\text{KO})_i$ for each $i \in M$. Knockoff variables can be generated using various methods, including Markov models \citep{Sesia2019GeneKnockoffs}, Metropolis-Hastings \citep{Bates2021MetropolizedSampling}, deep learning \citep{Romano2020DeepKnockoffs}, and non-parametric approaches \citep{Blain2024WhenKnockoffs}. In this manuscript, the second-order model-X construction is used, which matches the first two moments of $(\vect{X}, \vect{X}_\text{KO})$ and $(\vect{X}, \vect{X}_\text{KO})_{\operatorname{swap}(M)}$ \citep{Candes2018}.

\paragraph{Feature statistic.}
The knockoff variables are combined with the originals to form the concatenated design matrix $[\vect{X} \; \vect{X}_\text{KO}]$, which is used as input to a penalized regression model (typically with a lasso penalty) to produce $[\hat{\vect\beta} \; \hat{\vect{\beta}}_\text{KO}]$. A feature statistic is then computed for each variable. The coefficient-difference statistic is often used, defined for $i\in[p]$ by
\begin{equation}\label{eqn:knockoffs_stat}
    W_i = |\hat{\beta}_i| - |(\hat{\beta}_\text{KO})_i|.
\end{equation}
See \citet{Weinstein2020AStatistics} for a comparison of different statistics.

The feature statistic must satisfy the \textit{coin flip property}, which specifies that swapping a variable with its knockoff must reverse the sign of its statistic, and the \textit{sufficiency property}, meaning it depends on the design only through the covariance and on the response through marginal correlations \citep{Barber2019AInference}. Because knockoffs serve as controls, a null variable is equally likely to be selected as its knockoff, making $W_i$ symmetric around zero \citep{Ren2024DerandomisedControl}.

\paragraph{Filter.}
The Knockoffs filter applies a data-adaptive significance threshold to the feature statistic,
\begin{equation}\label{eqn:knockoff_threshold}
T = \min \left\{ t > 0: \frac{\#\left\{i \in [p]: W_i \leq -t \right\}}{\#\left\{i \in [p]: W_i \geq t\right\}} \leq q_\text{KO} \right\},
\end{equation}
where $q_\text{KO} \in (0,1)$ is the target FDR level, such that the filter generates a set of active variables $\hat{S}_v = \left\{i \in [p]: W_i \geq T \right\}$. FDR control follows from exchangeability, which makes the signs of the null statistics $W_i$ independent coin flips. Consequently, the threshold $T$ gives a conservative estimate of false discoveries \citep{Barber2015ControllingKnockoffs}.

\paragraph{Limitations and improvements.}
The Knockoffs framework has some key limitations. Doubling the data dimensionality can be computationally expensive, which has motivated the development of screening rules for Knockoffs \citep{Liu2022Model-FreeFeatures, Pan2022FeatureData}. Additionally, the knockoff generation process can be random, producing different results across runs. Consequently, several approaches have been proposed to \textit{derandomise} Knockoffs by aggregating results from multiple runs \citep{Ren2023DerandomizingKnockoffs}, including using e-values \citep{Ren2024DerandomisedControl}.

\subsubsection{Knockoffs with SLOPE}\label{section:slope_knockoffs}
SLOPE and Knockoffs both control the FDR, but under different assumptions. SLOPE assumes an orthogonal design matrix, limiting the FDR properties to low dimensions. The model-X Knockoffs framework has no dimensionality restrictions, but assumes the feature distribution $F_X$ is known and the conditional distribution $F_{Y\mid X}$ can be arbitrary and unknown. These assumptions are reasonable in genetics, where $F_{Y \mid X}$ is difficult to estimate due to complex correlation structures, whereas $F_X$ is often known or easily estimated, as it can be experimentally controlled \citep{Sesia2019GeneKnockoffs}. Knockoffs also assumes exchangeability of the knockoff variables and that the coin flip and sufficiency properties hold. Violations can make the feature statistic asymmetric, compromising FDR control \citep{Blain2024WhenKnockoffs}.

The flexibility of Knockoffs stems from its use of penalized regression, allowing any estimator that satisfies the assumptions. As such, we explore combining Knockoffs with SLOPE to achieve FDR control in general settings. \citet{Humayoo2018Model-freeRate} showed that SLOPE with Knockoffs can improve FDR control and power over the lasso in simulations, though several questions remain.

These are: Does the SLOPE statistic satisfy the coin flip property? This has not been addressed in \citet{Humayoo2018Model-freeRate} or elsewhere. How does the SLOPE FDR parameter ($q_v$) interact with the Knockoffs parameter ($q_\text{KO}$)? \citet{Ren2024DerandomisedControl} show, for eBH and Knockoffs, that setting the Knockoffs parameter lower than the eBH one is preferable for FDR control. Finally, how should the SLOPE design be optimized for Knockoffs, and are new penalty sequences or knockoff constructions needed? This is especially relevant since SLOPE’s steep decay and clustering can reduce contrast between $\vect{X}$ and $\vect{X}_\text{KO}$.

Although these questions are beyond the scope of this manuscript, the following section explores in more detail how the Knockoffs framework can be applied to SGS.

\subsubsection{Sparse-group Knockoffs}\label{section:sparse_group_knockoffs}
Knockoffs have been applied to group regression in \citet{Dai2016TheRegression}, where multi-task regression is reformulated as a group lasso problem and a group-specific Knockoffs construction is proposed. For the feature statistic, the norm is modified to be 
\begin{equation*}
J_\text{glasso, KO}(\vect{b}) = \sum_{j=1}^m\|\vect{b}^{(j)}\|_2 + \sum_{j=1}^m\|(\vect{b}_\text{KO})^{(j)}\|_2.    
\end{equation*}
The path statistic is then used, defined as the point along the path at which a group enters the model. To allow for direct comparison to the variable feature statistic, in this manuscript, we adopt the group effects statistic for $j \in [m]$,
\begin{equation}\label{eqn:knockoffs_stat_grp}
    W = \|\hat{\vect\beta}^{(j)}\|_2 - \|(\hat{\vect{\beta}}_\text{KO})^{(j)}\|_2.
\end{equation}
The Knockoffs construction and filter proceed as in the variable case. The knockoff variables are assigned to their own groups, so that each original group has a corresponding knockoff group.

Although multi-layer Knockoffs filters exist \citep{Barber2017TheHypotheses, gablenz, Gu2024PinpointingKnockoffs, Katsevich2019MultilayerResolutions}, Knockoffs have not been specifically studied in the context of sparse-group regression models. Two main concerns in this setting are the construction of knockoffs and the coin flip property. Constructing knockoffs compatible with both variable and group penalties is challenging: the coin flip property does not hold at a variable level under a group knockoff construction \citep{Gu2024PinpointingKnockoffs}. Our investigation for the sparse-group lasso (SGL) and SGS shows that the coin flip property holds only if the knockoffs are placed in the same group as the originals, which may dilute group penalization, and holds less often for SGS due to its dual sorting procedures (Figure \ref{fig:knockoffs-coin-flip}).

In this manuscript, we present an initial implementation of sparse-group Knockoffs using SGS. SGS provides both variable- (Equation \ref{eqn:knockoffs_stat}) and group-level (Equation \ref{eqn:knockoffs_stat_grp}) feature statistics, with corresponding thresholds used to select active variables and groups, enabling bi-level FDR control. To satisfy the coin flip property, the knockoff variables are placed into the same groups as the originals; this is in contrast to the group-only setup, where they form separate groups. This gives the SGS Knockoffs problem as
\begin{align*}
   [\hat{\vect\beta}(\lambda) \; \hat{\vect{\beta}}_\text{KO}(\lambda)] \in \argmin_{\vect{b} \in \mathbb{R}^{2p}} \bigg\{ \frac{1}{2n}\|\vect{y} - [\vect{X} \; \vect{X}_\text{KO}]\vect{b}\|_2^2
   &+ \lambda\alpha \sum_{i=1}^{2p}v_i |b|_{(i)}\\ 
   &+  \lambda(1-\alpha)\sum_{j=1}^{2m}w_j \sqrt{p_j} \|\vect{b}^{(j)}\|_2\bigg\}.
\end{align*}

The screening rules reduce the computational cost of Knockoffs, but the synthetic data results in Section \ref{section:simulation_study} show that it is the most computationally intensive model selection approach for SGS.

\begin{figure}[H]
    \centering
    \begin{subfigure}[b]{0.48\linewidth}
        \centering
        \includegraphics[width=\linewidth]{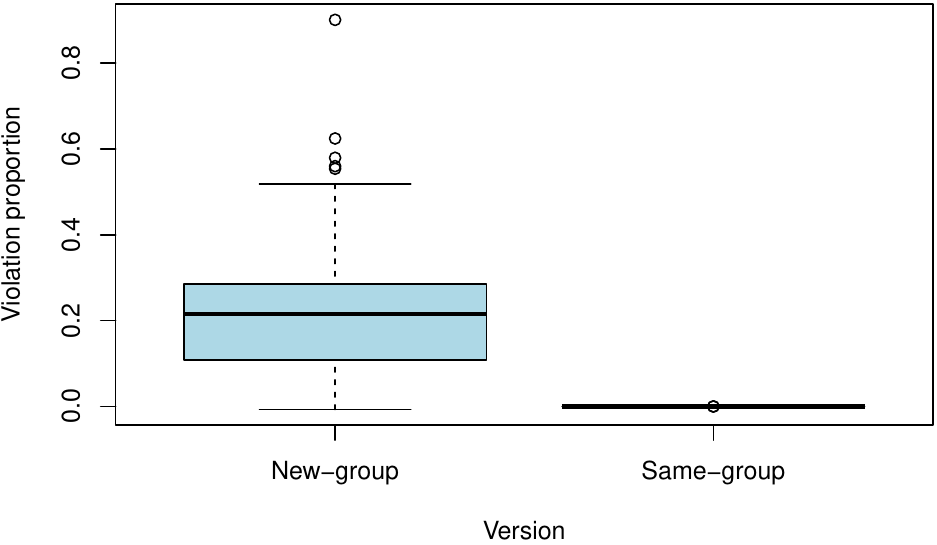}
        \caption{SGL}
    \label{fig:coin_flip_test_sgl}
    \end{subfigure}
    \hfill
    \begin{subfigure}[b]{0.48\linewidth}
        \centering
        \includegraphics[width=\linewidth]{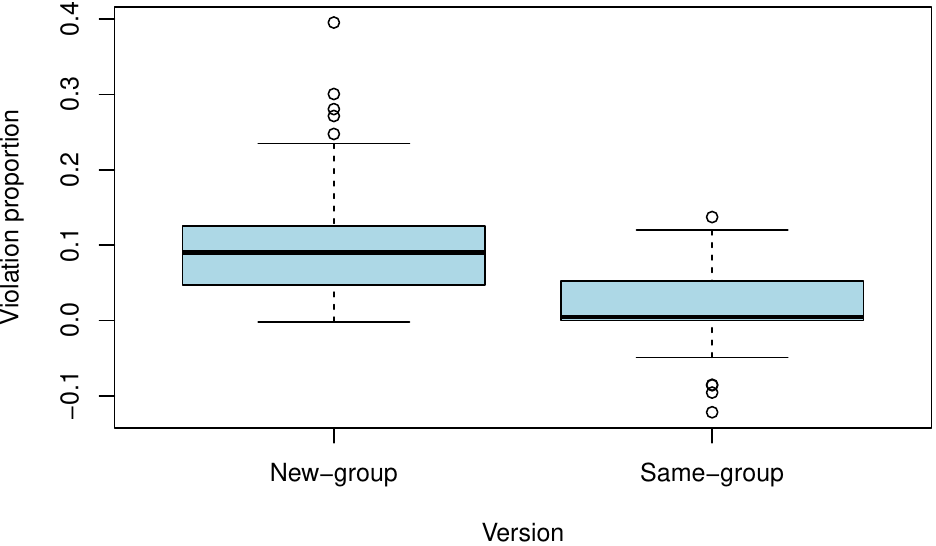}
        \caption{SGS}
        \label{fig:coin_flip_test_sgs}
    \end{subfigure}
    \caption{The proportion of instances in which the coin flip property did not hold for active variables, using a second-order construction, shown for the sparse-group lasso and SGS. The Knockoff filter was run twice: once with the original design, $[\mathbf{X} \; \mathbf{X}_\text{KO}]$, and once with the signal variables swapped with their knockoff copies.}
    \label{fig:knockoffs-coin-flip}
\end{figure}
 
\subsection{Two-step}\label{section:two-step}
Two-step procedures combine the positive elements of multiple models into a stronger overall model. \citet{10.1214/17-EJS1287} proposed a two-step method for prediction, using the lasso followed by least squares refitting to debias coefficients. The approach addresses CV’s computational inefficiency and lack of finite-sample guarantees. Similarly, \citet{10290898} introduced the two-stage prior lasso (TSPLASSO), performing two lasso stages for gene selection and sample classification. For overviews of multi-step regression, see relevant summaries \citep{doi:10.1177/09622802251343597,10.1214/08-AOS646, zhang2008twostepproceduresvariableselection}, and for FDR-controlling two-step methods see \citet{SARKAR20081072}.

\subsubsection{Two-step SLOPE}\label{section:two-step-slope}
The \textit{Two-step SLOPE} (TS-SLOPE) procedure has been shown to recover the true model with high probability
\citep{Graczyk2023AGeometry,Hejny2024UnveilingRegularizers}. The procedure is defined by the steps
\begin{enumerate}
    \item Obtain an initial estimate $\hat{\vect\beta}_{[1]}$ from applying SLOPE.
    \item Obtain a truncated estimate $\hat{\vect\beta}_{[2]} = \operatorname{prox}_\text{slope}(\hat{\vect\beta}_{[1]})$, where $\operatorname{prox}_\text{slope}$ is the proximal operator of the SLOPE penalty (given in Algorithm 3 in \citet{Bogdan2015SLOPEAdaptiveOptimization}).
\end{enumerate}
The final estimate is shown to recover the pattern of the true signal, although it is also shown to be heavily biased \citep{Hejny2024UnveilingRegularizers}. 

This approach can be extended to gSLOPE and SGS using the model-specific proximal operators. The gSLOPE proximal operator is given in Algorithm 2 in \citet{Gossmann2015}. For SGS, the proximal operators are applied one after another, so that step 2 becomes \begin{equation*}
\hat{\vect\beta}_{[2]} = \operatorname{prox}_\text{slope}(\operatorname{prox}_\text{gslope}(\hat{\vect\beta}_{[1]})).    
\end{equation*} 
Applying the operators in the reverse order was also tested and yielded nearly identical results.

\subsection{A short simulation study}\label{appendix:fdr_sim_study}
\paragraph{Setup.} The response was generated using a linear model $\vect{y} = \vect{X}\vect\beta + \vect{\mathcal{N}}_n(\vect{0},\vect{I}_n)$ with design matrix $\vect{X}\sim \vect{\mathcal{N}}_p(\vect{0}, \vect\Sigma)\in \mathbb{R}^{100 \times 500}$. Three correlation structures for $\vect\Sigma$ were considered: 
    \begin{enumerate}
        \item $\vect\Sigma = \vect{I}_p$ (independent features).
        \item $\Sigma_{i,i} = 1$, $\Sigma_{i,j} = 0.3$ for $i \neq j$ (moderate correlation).
        \item $\Sigma_{i,i} = 1$, $\Sigma_{i,j} = 0.9$ for $i \neq j$ (strong correlation).
    \end{enumerate}
The signal was set to $\beta_i \sim \mathcal{N}(0, 10), i \in S_v$, with $0.95$ variable sparsity proportion. The following methods are applied using a $50$-length lasso path:
\begin{itemize}
    \item Stability Selection (SS) \citep{Meinshausen2010StabilitySelection} and Complementary Pairs Stability Selection (CPSS) \citep{Shah2013VariableSelection}: target family-wise error rate of $0.1$ and $B=100$ subsamples.
    \item Knockoffs \citep{Barber2019AInference}: Using the second-order Knockoffs construction.
    \item FDR Stability Selection (FDR-SS) \citep{Ahmed2011FalseStudies}: $B=15$ subsamples and $D=5$ permuted datasets.
    \item Kappa Selection (KS) \citep{Sun2013ConsistentStability}: $B=40$ subsamples and a threshold of $\alpha_n=0.1$.
    \item ET-LASSO \citep{Yang2019ET-Lasso:Data}.
    \item CV: includes the min and 1se versions.
\end{itemize}
The lasso was used, instead of SLOPE, to ensure a fair comparison, as the majority of methods compared here were developed for the lasso.

\paragraph{Results.} Table \ref{tbl:short_sim_study} shows the FDR (with target level $0.1$) and sensitivity for the approaches considered: Knockoffs is found to have the lowest FDR levels for all cases.
\begin{table}[H]
\centering
\resizebox{\textwidth}{!}{%
\begin{tabular}{llrrrrrrrr}
\toprule
\textbf{Case} & \textbf{Metric} & \textbf{CPSS} & \textbf{CV 1se} & \textbf{CV min} & \textbf{ET-LASSO} & \textbf{FDR-SS} & \textbf{Knockoffs} & \textbf{KS} & \textbf{SS} \\
\midrule
\multirow{2}{*}{Case 1} 
& FDR         & $0.25$ & $0.54$ & $0.75$ & $0.28$ & $0.75$ & $\vect{0.23}$ & $0.87$ & $0.25$ \\
& Sensitivity & $0.17$ & $0.24$ & $0.40$ & $0.17$ & $0.31$ & $0.15$ & $\vect{0.54}$ & $0.18$ \\
\midrule
\multirow{2}{*}{Case 2} 
& FDR         & $0.32$ & $0.62$ & $0.76$ & $0.58$ & $0.80$ & $\vect{0.30}$ & $0.88$ & $0.33$ \\
& Sensitivity & $0.16$ & $0.22$ & $0.35$ & $0.18$ & $0.31$ & $0.10$ & $\vect{0.50}$ & $0.16$ \\
\midrule
\multirow{2}{*}{Case 3} 
& FDR         & $0.77$ & $0.81$ & $0.85$ & $0.82$ & $0.91$ & $\vect{0.71}$ & $0.92$ & $0.79$ \\
& Sensitivity & $0.05$ & $0.06$ & $0.11$ & $0.06$ & $0.19$ & $0.02$ & $\vect{0.28}$ & $0.05$ \\
\bottomrule
\end{tabular}
}
\caption[Metrics for short synthetic study]{FDR and sensitivity for model selection methods under three cases, with the best performance in each case highlighted in \textbf{bold}.}
\label{tbl:short_sim_study}
\end{table}

\section{Simulation study}
\subsection{Computational details}
SLOPE optimization was performed using the \texttt{SLOPE} R package \citep{SLOPEpackage}, which uses the Fast Iterative Shrinkage-Thresholding Algorithm (FISTA) \citep{doi:10.1137/080716542}. gSLOPE and SGS optimizations were performed using the \texttt{sgs} R package \citep{sgs-r-package}, which uses ATOS \citep{Pedregosa2018AdaptiveSplitting}. Knockoffs is implemented using the \texttt{knockoff} R package. AS-SGS, scaled regression, TSO, and TS-SLOPE were implemented in R, as was the SAEM algorithm used to fit the Bayesian models. The \texttt{glmnet} R package \citep{Friedman2010b} is used for the lasso in TSO.

\begin{table}[H]
    \centering
    \begin{tabular}{@{}llcc@{}}
        \toprule
        \textbf{Category} & \textbf{Parameter} & \multicolumn{2}{c}{\textbf{Values}} \\
        \cmidrule(lr){3-4}
         &  & {\textbf{Synthetic}}& {\textbf{Real}} \\
        \midrule
        \multicolumn{4}{@{}l}{\textbf{Data (baseline/default)}} \\
        \midrule
        & $p$ & $500$ & -\\
        & $n$ & $400$ & - \\
        & $m$ & $37$ & - \\
        &Group sizes & $[5,25]$ & - \\
        & Signal $\beta$ ($s=5$) & $\mathcal{N}(5,10)$ & - \\
        & Variable sparsity ($\xi_v$) & $0.3$ & - \\
        & Group sparsity ($\xi_g)$ & $0.2$ & - \\
        & Within group correlation ($\rho_w$) & $0.3$ & -\\
        & Across group correlation ($\rho_a$) & $0$ & -\\
        & Noise ($\sigma$) & $1$& - \\  
        \midrule
        \multicolumn{4}{@{}l}{\textbf{Optimization algorithm (ATOS)}} \\
        \midrule
        & Maximum iterations & $5000$ & $5000$ \\
        & Backtracking & $0.7$ & $0.7$ \\
        & Maximum backtracking iterations  & $100$ & $100$ \\
        &Convergence tolerance & $10^{-5}$& $10^{-5}$ \\
        &Standardization & $\ell_2$ & $\ell_2$ \\
        &Intercept & Yes & Yes \\
        &Warm starts & Yes & Yes \\
       \midrule
        \multicolumn{4}{@{}l}{\textbf{Bayesian algorithm (SAEM)}} \\
        \midrule
        & Maximum iterations & $500$ & $500$ \\
        &Convergence tolerance & $10^{-5}$& $10^{-5}$ \\
        &Standardization & $\ell_2$ & $\ell_2$ \\
                \midrule
        \multicolumn{4}{@{}l}{\textbf{Model specific parameters}} \\
        \midrule
        & $\alpha$ (SGS only)  & $0.95$ & $0.99$ \\
       & $q_v$  & $0.1$ & $0.1$ \\
       & $q_g$  & $0.1$ & $0.1$\\
        & $\boldsymbol\beta$ initialization model (Bayesian only)  & Lasso & Lasso \\
        & $\theta$ Beta prior (ABSLOPE)  & $d_1 = d_2= 0.01n$ & $d_1 = d_2= 0.01n$ \\
        & $\theta$ Beta prior (BGSLOPE)  & $d_1 = d_2= 0.01n$ & $d_1 = d_2= 0.01n$ \\
        & $\theta_g$ Beta prior (BSGS)  & $d_1 = 0.003n, d_2=
0.015n$ & $d_1 = 5, d_2 =1$\\
        & $\theta_v$ Beta prior (BSGS)  & $e_1 = 0.003n, e_2 =
0.015n$ & $e_1 = 5, e_2 = 1$ \\
        & Path length ($l$) & $20$ & $20$ \\
        & Path termination ($\lambda_l$) & $0.1\lambda_1$ & $0.1\lambda_1$ \\
        &Path shape & Log-linear & Log-linear\\
        \bottomrule
    \end{tabular}
        \caption{Default model, data, and algorithm parameters for the synthetic and real data analyses. Note that only Knockoffs and CV fit a path.}
    \label{tbl:appendix_model_data_simulation}
\end{table}
\subsection{Results}\label{appendix:model_selection_synthetic_study}
\begin{figure}[H]
    \centering
\includegraphics[width=1\linewidth]{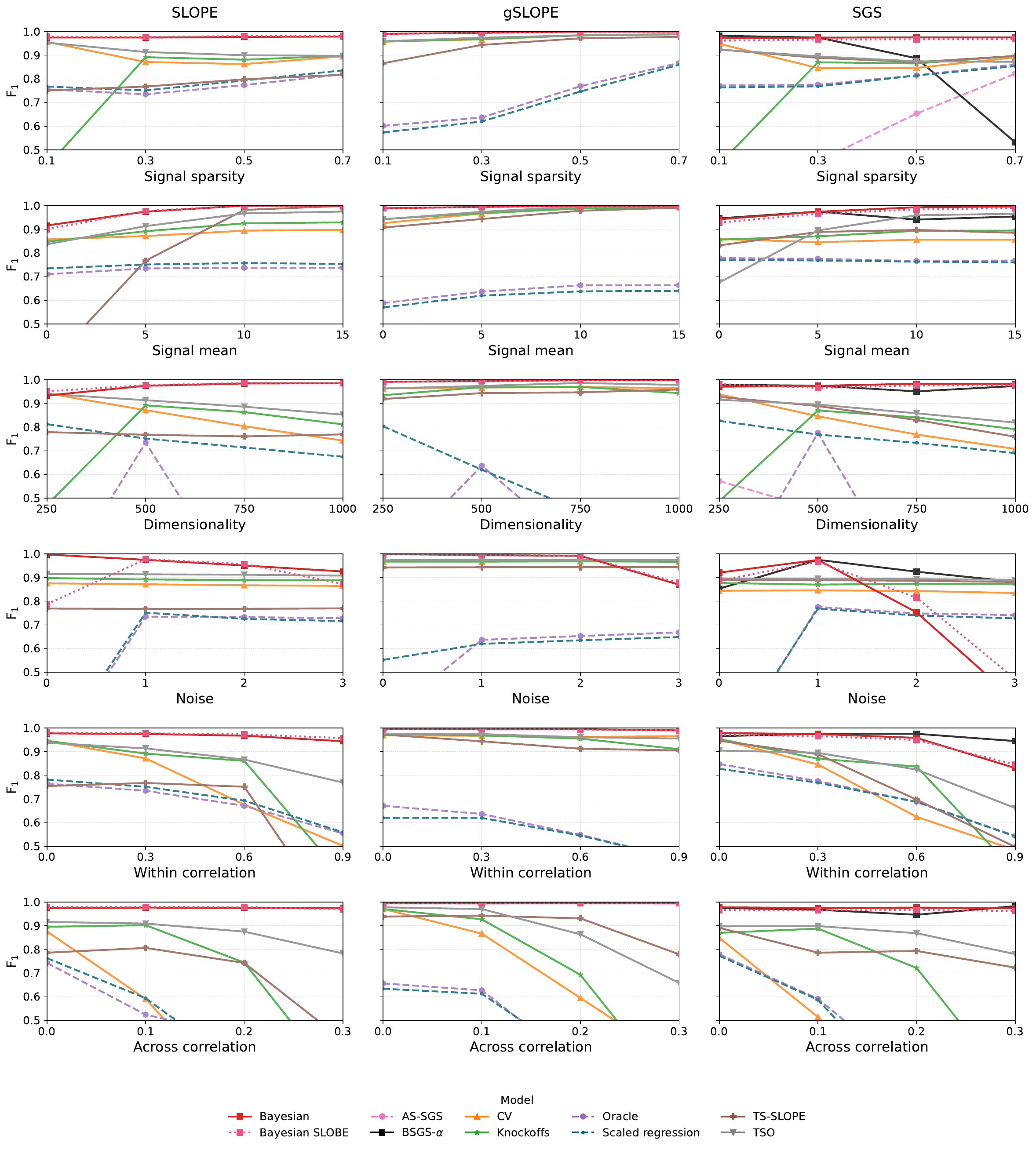}
    \caption{$\text{F}_1$ scores for all model selection approaches, shown for all cases considered, split into the type of model (SLOPE, gSLOPE, SGS).}
    \label{fig:all-f1-plot}
\end{figure}
\begin{figure}[H]
    \centering
\includegraphics[width=1\linewidth]{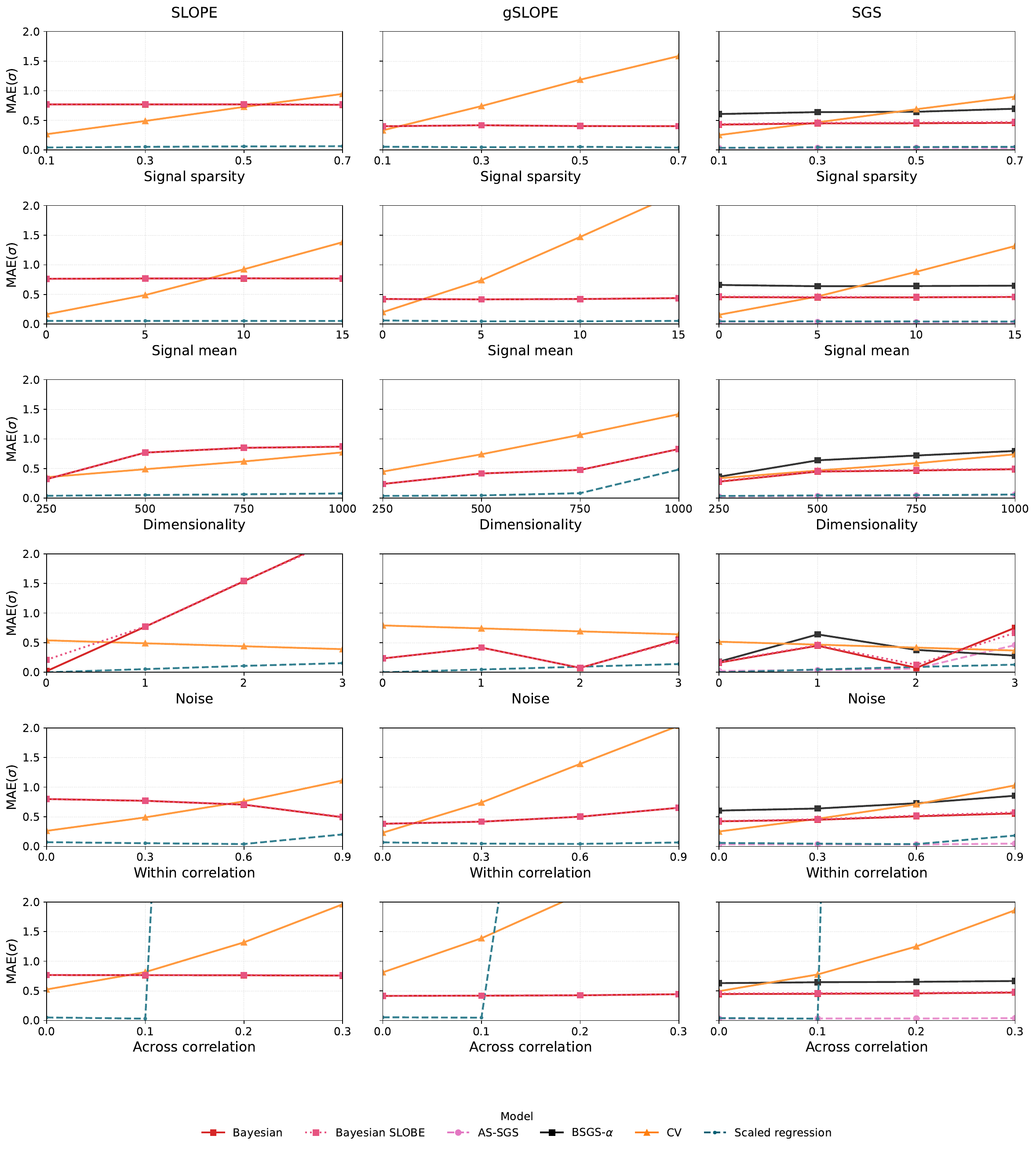}
    \caption{$\text{MAE}(\sigma)$ for all model selection approaches that estimate the noise, shown for all cases considered, split into the type of model (SLOPE, gSLOPE, SGS).}
    \label{fig:all-sigma-plot}
\end{figure}
\subsubsection{Impact of signal}
\begin{figure}[H]
    \centering
\includegraphics[width=1\linewidth]{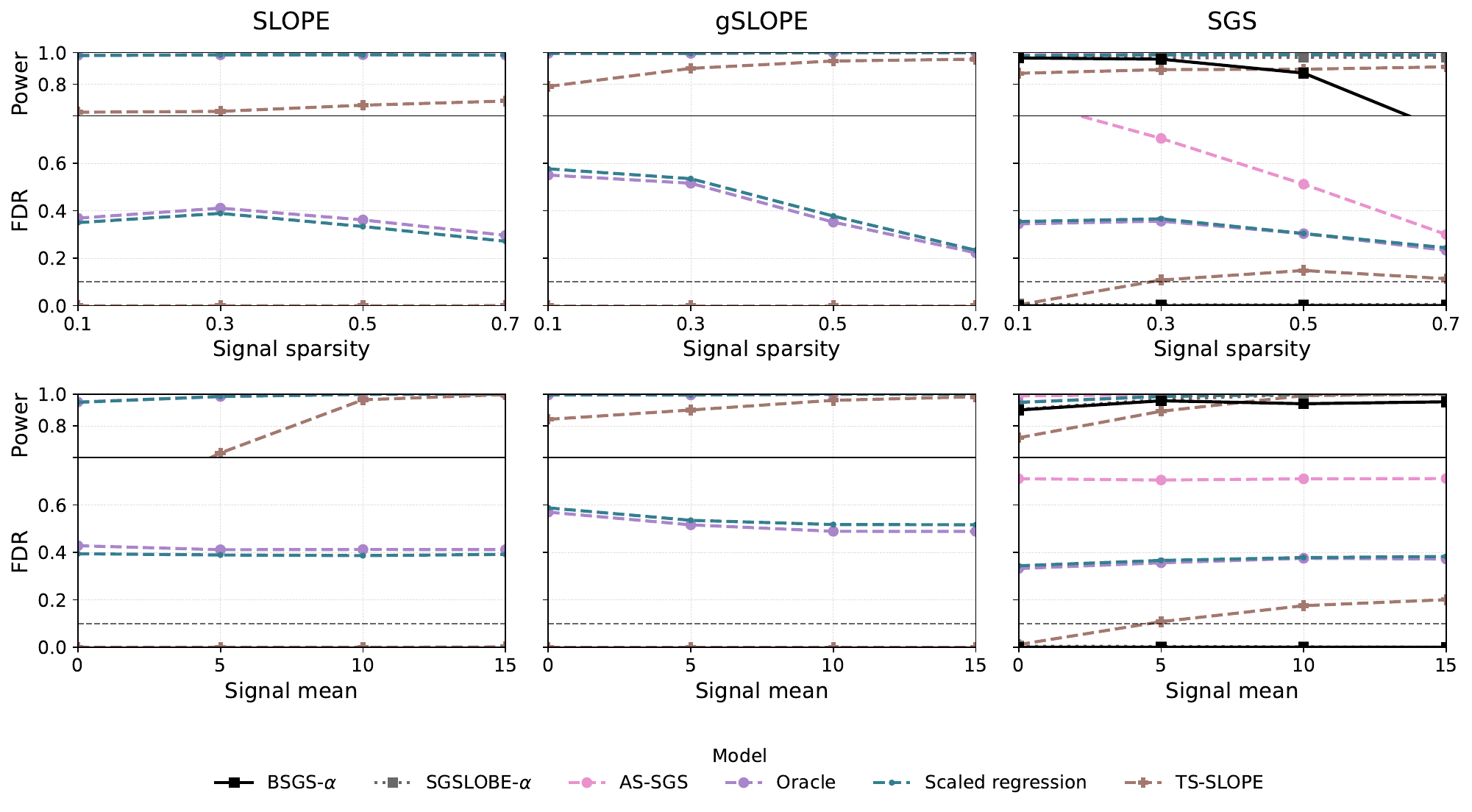}
    \caption{FDR (bottom plots) and power (top plots) for the other model selection approaches, as functions of the sparsity proportion (top row) and signal strength (bottom row), split into the type of model (SLOPE, gSLOPE, SGS).}
    \label{fig:impact-of-signal-other-methods}
\end{figure}

\subsubsection{Impact of data-generating parameters}
\begin{figure}[H]
    \centering
\includegraphics[width=1\linewidth]{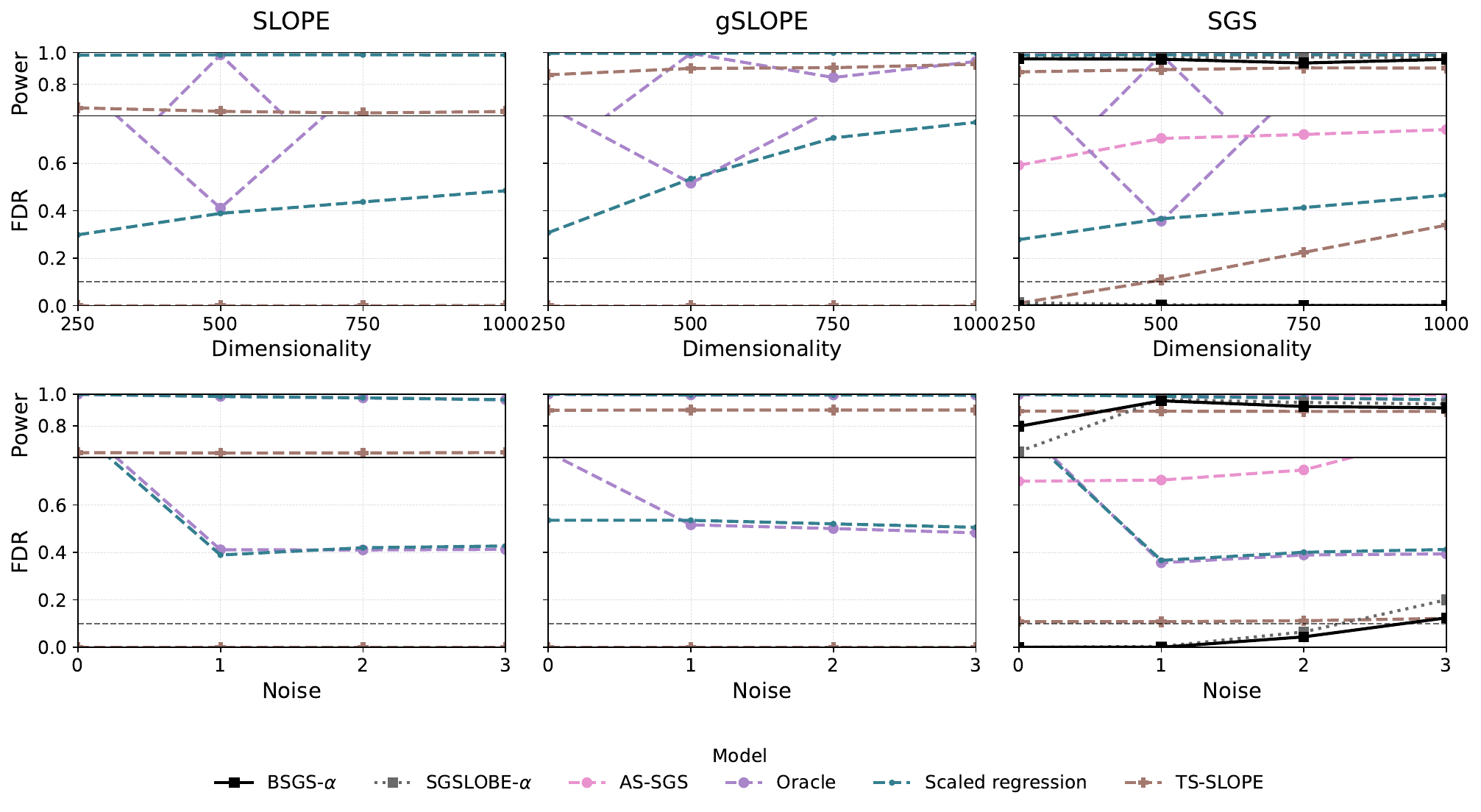}
    \caption{FDR (bottom plots) and power (top plots) for the other model selection approaches, as functions of the dimensionality (top row) and noise (bottom row), split into the type of model (SLOPE, gSLOPE, SGS).}
    \label{fig:noise-dim-other}
\end{figure}

\begin{figure}[H]
    \centering
\includegraphics[width=1\linewidth]{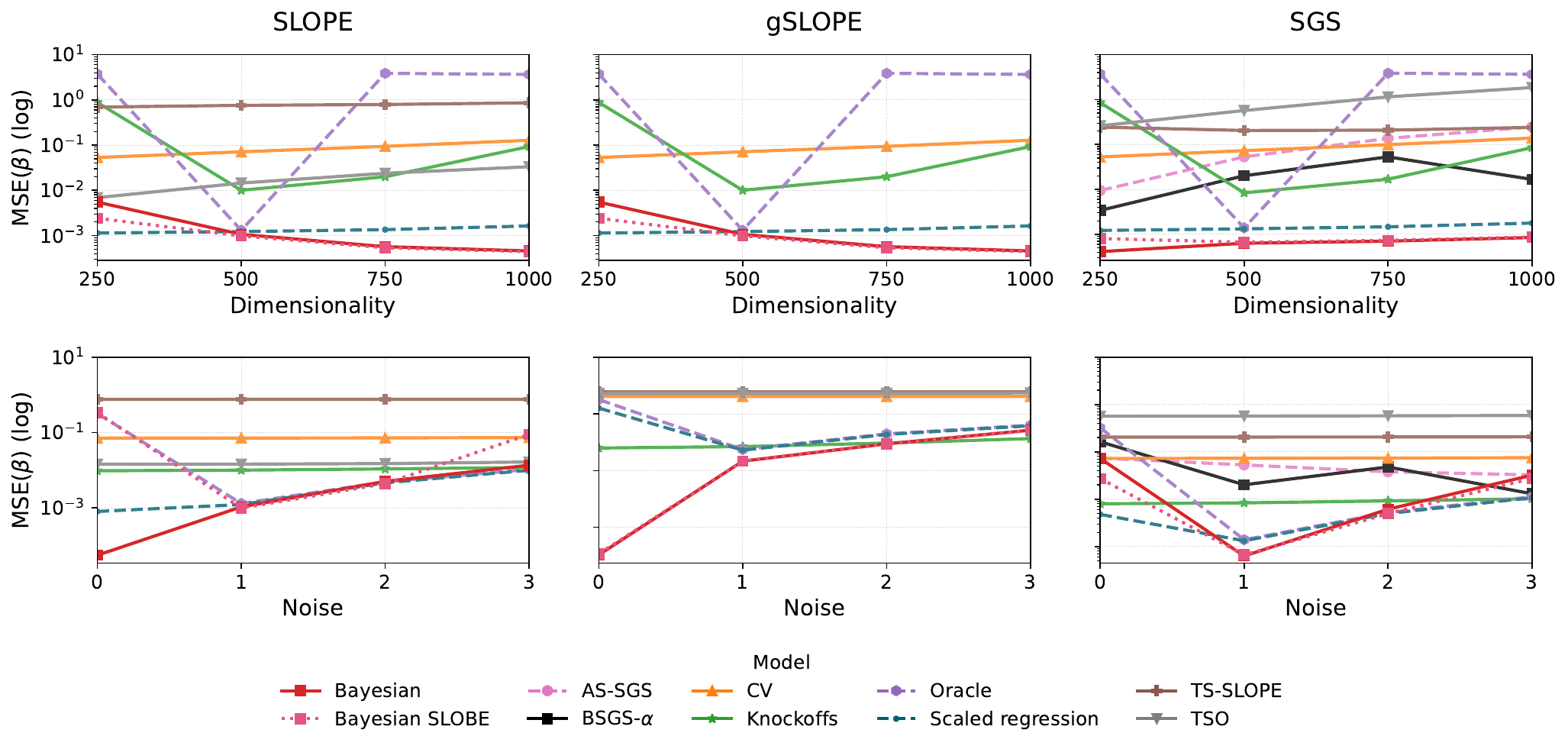}
    \caption{MSE($\boldsymbol\beta$) (log scale) for all model selection approaches, as a function of the dimensionality (top row) and noise (bottom row), split into the type of model (SLOPE, gSLOPE, SGS).}
    \label{fig:noise-dim-mse-beta}
\end{figure}
\begin{figure}[H]
    \centering
\includegraphics[width=1\linewidth]{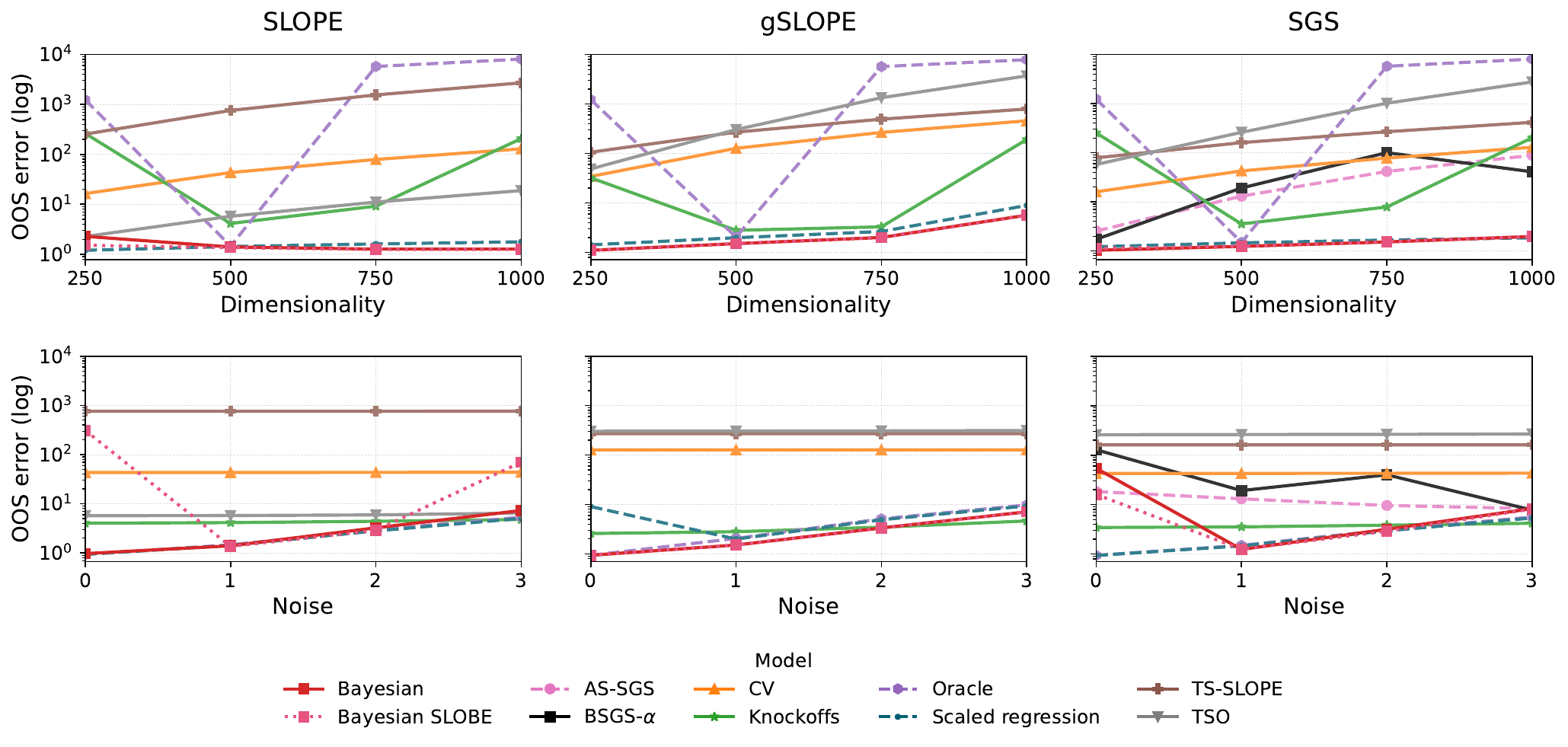}
    \caption{OOS error (log scale) for all model selection approaches, as a function of the dimensionality (top row) and noise (bottom row), split into the type of model (SLOPE, gSLOPE, SGS).}
    \label{fig:noise-pred}
\end{figure}
\begin{figure}[H]
    \centering
\includegraphics[width=1\linewidth]{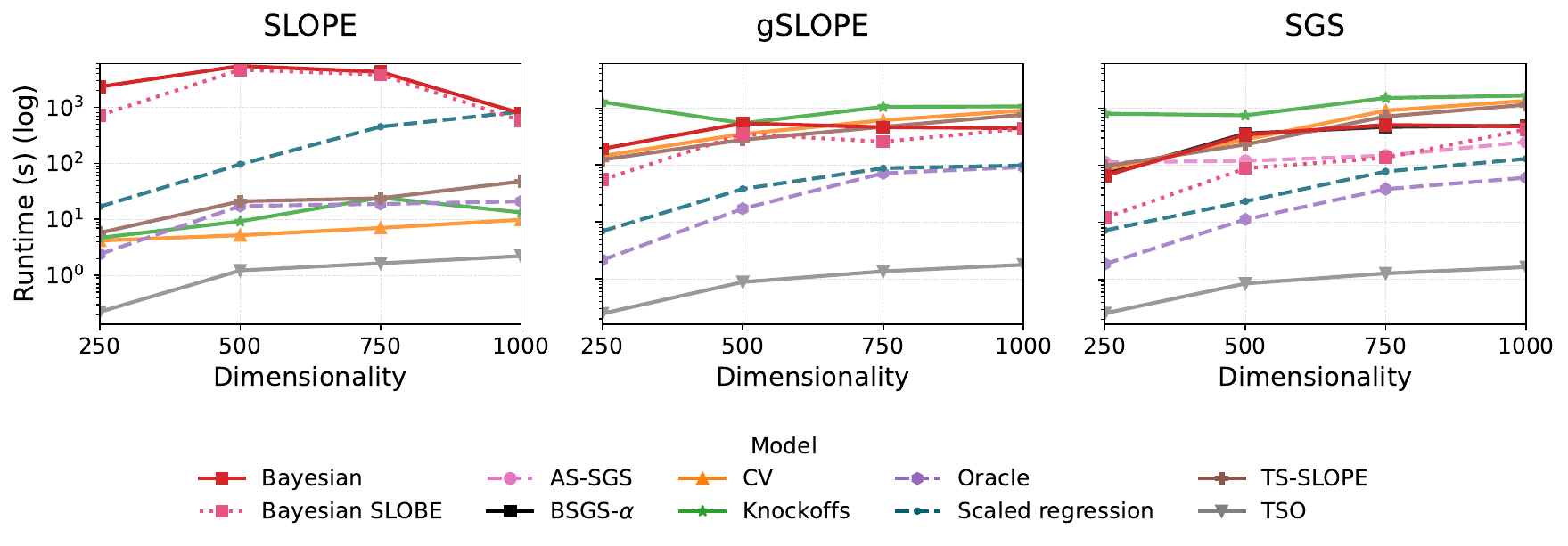}
    \caption{Runtime in seconds (log scale) for all model selection approaches, as a function of the dimensionality, split into the type of model (SLOPE, gSLOPE, SGS).}
    \label{fig:study-1-dim-runtime}
\end{figure}
\subsubsection{Correlation}
\begin{figure}[H]
    \centering
\includegraphics[width=1\linewidth]{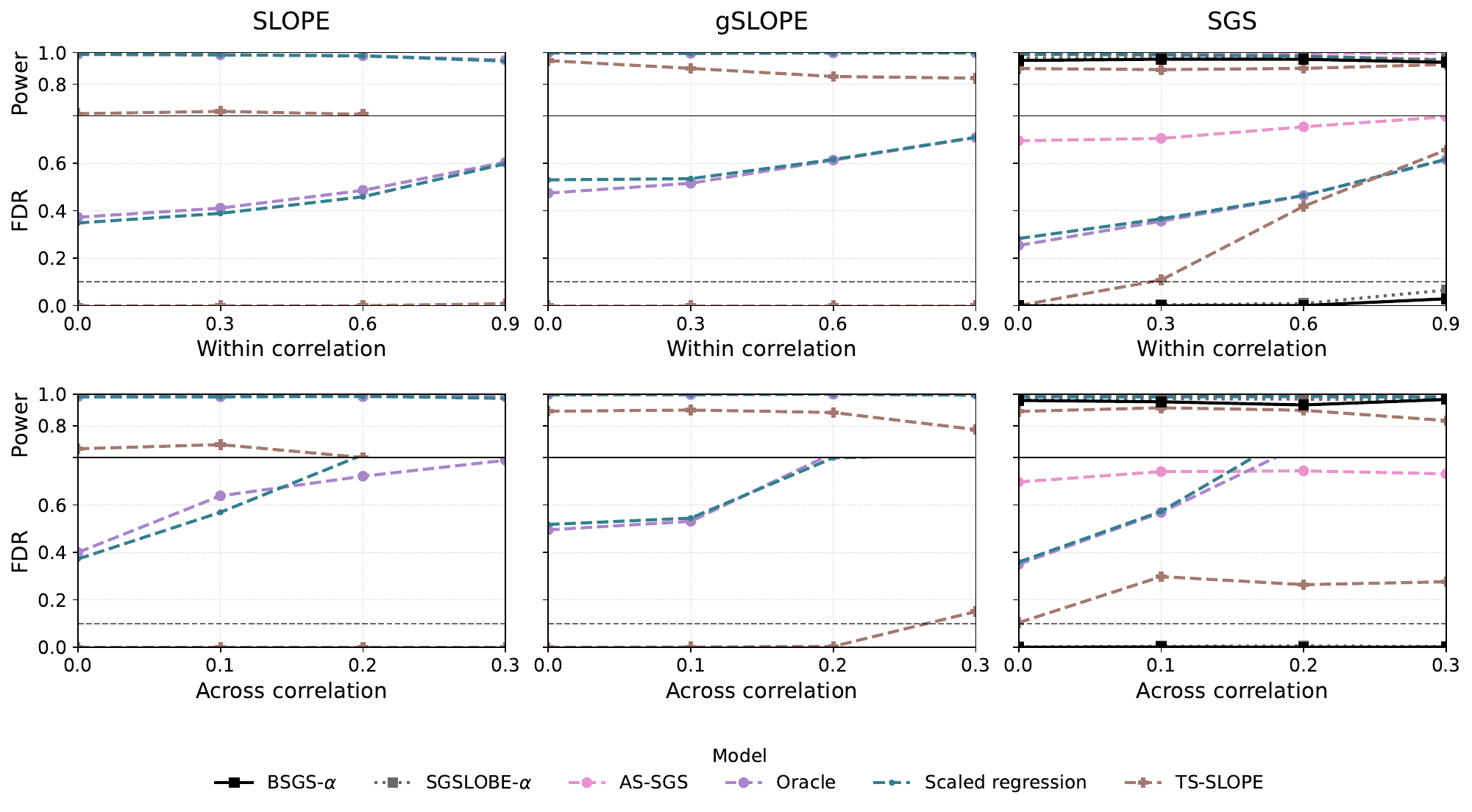}
    \caption{FDR (bottom plots) and power (top plots) for the other model selection approaches, as functions of within-group correlation (top row) and across-group correlation (bottom row), split into the type of model (SLOPE, gSLOPE, SGS).}
    \label{fig:corr-other}
\end{figure}
\begin{figure}[H]
    \centering
\includegraphics[width=1\linewidth]{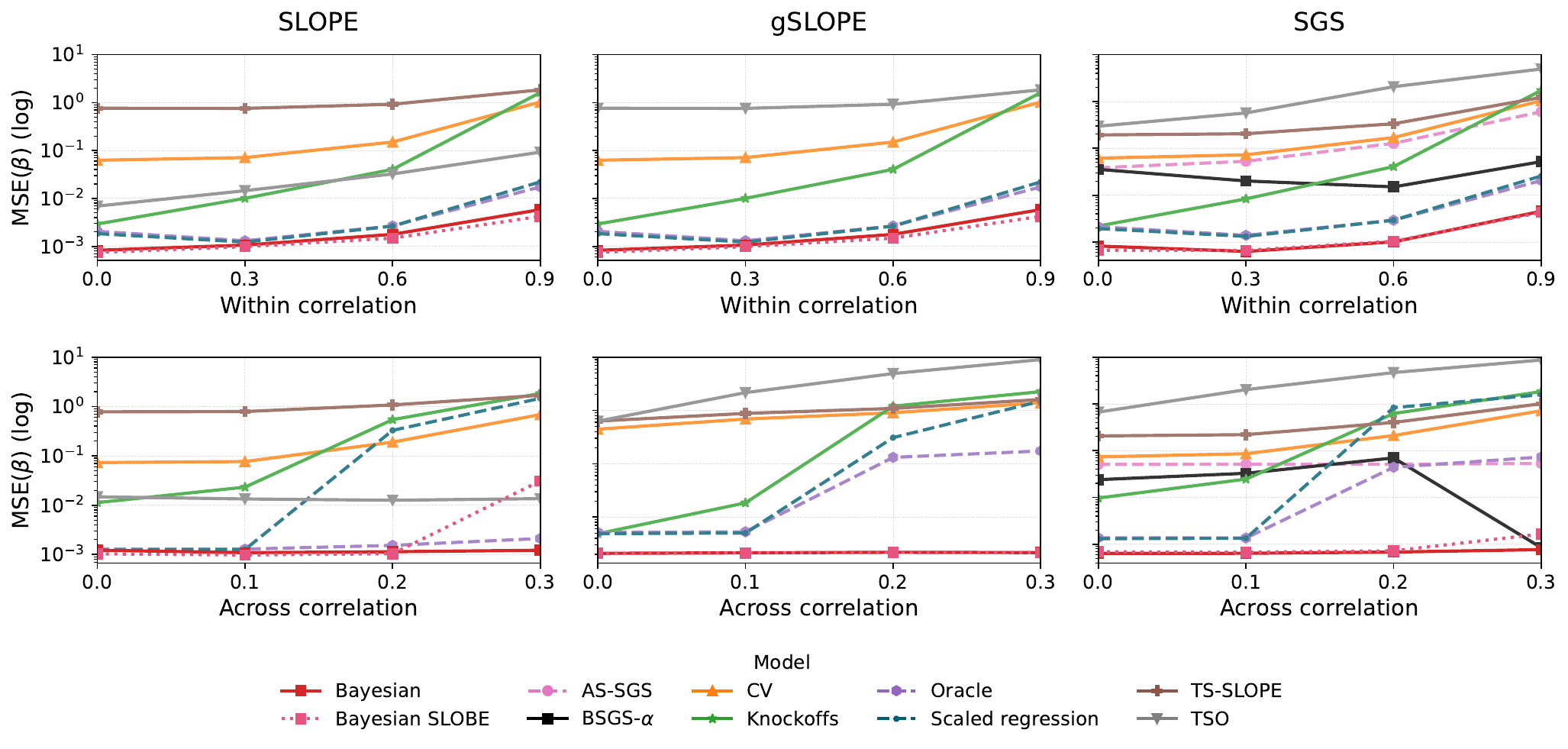}
    \caption{MSE($\boldsymbol\beta$) (log scale) for all model selection approaches, as a function of within-group correlation (top row) and across-group correlation (bottom row), split into the type of model (SLOPE, gSLOPE, SGS).}
    \label{fig:corr-mse-beta}
\end{figure}
\begin{figure}[H]
    \centering
\includegraphics[width=1\linewidth]{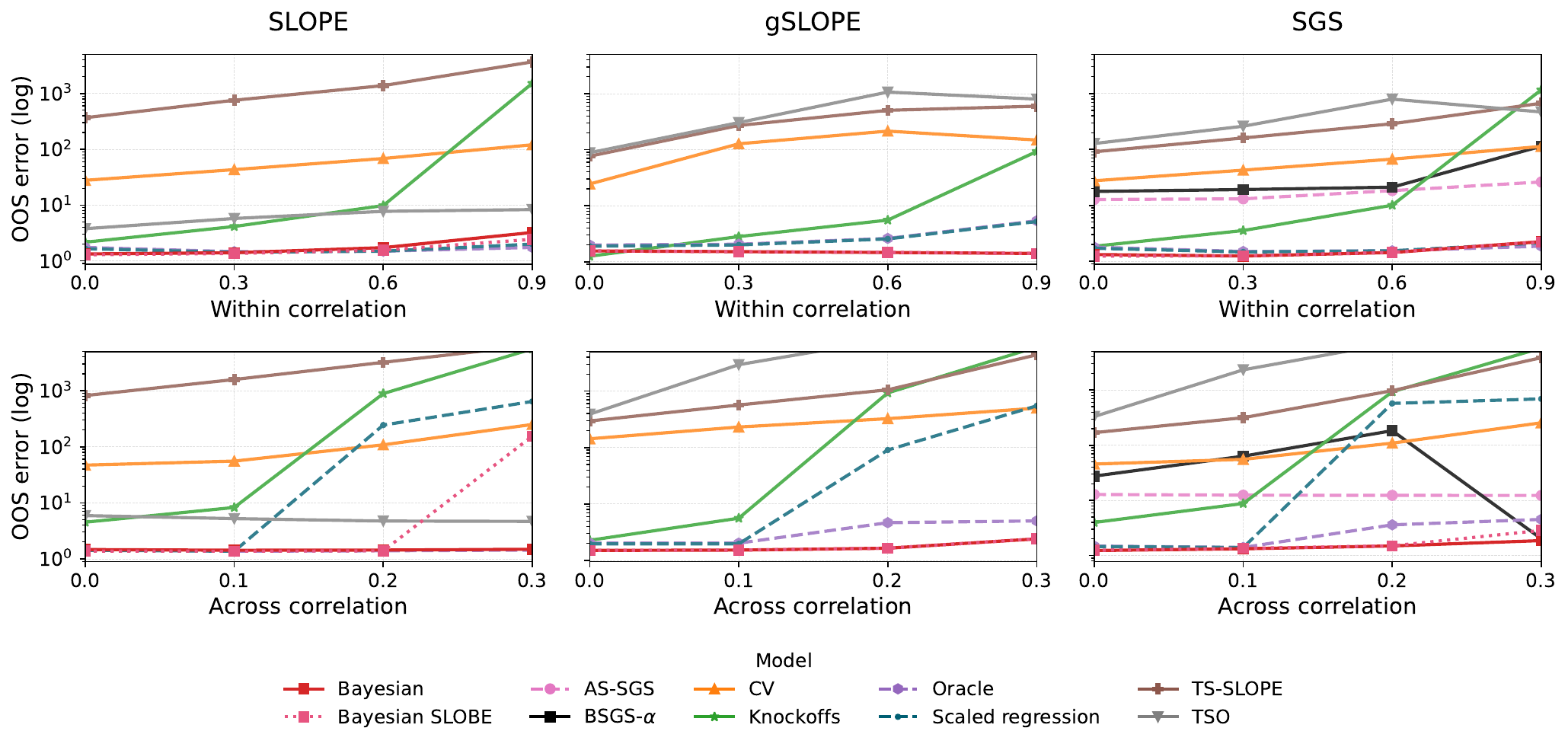}
    \caption{OOS error (log scale) for all model selection approaches, as a function of within-group correlation (top row) and across-group correlation (bottom row), split into the type of model (SLOPE, gSLOPE, SGS).}
    \label{fig:corr-pred}
\end{figure}
\subsection{Summary}

\begin{table}[H]
\centering
\begin{tabular}{lccc}
\toprule
\textbf{Model} & \textbf{FDR group} & \textbf{Power group} & \textbf{$\text{F}_1$ group} \\
\midrule
AS-SGS & $\underset{\scriptscriptstyle \textcolor{gray}{(0.01)}}{0.17}$ & $\underset{\scriptscriptstyle \textcolor{gray}{(4\times 10^{-3})}}{\mathbf{0.99}}$ & $\underset{\scriptscriptstyle \textcolor{gray}{(0.01)}}{0.88}$ \\
BSGS & $\underset{\scriptscriptstyle \textcolor{gray}{(0.01)}}{0.14}$ & $\underset{\scriptscriptstyle \textcolor{gray}{(3\times 10^{-3})}}{\mathbf{0.99}}$ & $\underset{\scriptscriptstyle \textcolor{gray}{(0.01)}}{0.90}$ \\
SGSLOBE & $\underset{\scriptscriptstyle \textcolor{gray}{(0.01)}}{0.17}$ & $\underset{\scriptscriptstyle \textcolor{gray}{(3\times 10^{-3})}}{\mathbf{0.99}}$ & $\underset{\scriptscriptstyle \textcolor{gray}{(0.01)}}{0.89}$ \\
BSGS-$\alpha$ & $\underset{\scriptscriptstyle \textcolor{gray}{(2\times 10^{-3})}}{\mathbf{0.02}}$ & $\underset{\scriptscriptstyle \textcolor{gray}{(0.01)}}{0.96}$ & $\underset{\scriptscriptstyle \textcolor{gray}{(0.01)}}{0.96}$ \\
SGSLOBE-$\alpha$ & $\underset{\scriptscriptstyle \textcolor{gray}{(4\times 10^{-3})}}{\mathbf{0.03}}$ & $\underset{\scriptscriptstyle \textcolor{gray}{(4\times 10^{-3})}}{\mathbf{0.99}}$ & $\underset{\scriptscriptstyle \textcolor{gray}{(4\times 10^{-3})}}{\mathbf{0.97}}$ \\
CV & $\underset{\scriptscriptstyle \textcolor{gray}{(3\times 10^{-3})}}{\mathbf{0.09}}$ & $\underset{\scriptscriptstyle \textcolor{gray}{(0.01)}}{0.98}$ & $\underset{\scriptscriptstyle \textcolor{gray}{(4\times 10^{-3})}}{0.91}$ \\
Knockoffs & $\underset{\scriptscriptstyle \textcolor{gray}{(3\times 10^{-3})}}{\mathbf{0.02}}$ & $\underset{\scriptscriptstyle \textcolor{gray}{(0.01)}}{0.96}$ & $\underset{\scriptscriptstyle \textcolor{gray}{(0.01)}}{\mathbf{0.97}}$ \\
Oracle & $\underset{\scriptscriptstyle \textcolor{gray}{(0.01)}}{0.46}$ & $\underset{\scriptscriptstyle \textcolor{gray}{(4\times 10^{-3})}}{0.93}$ & $\underset{\scriptscriptstyle \textcolor{gray}{(0.01)}}{0.65}$ \\
Scaled & $\underset{\scriptscriptstyle \textcolor{gray}{(0.01)}}{0.42}$ & $\underset{\scriptscriptstyle \textcolor{gray}{(2\times 10^{-3})}}{\mathbf{0.99}}$ & $\underset{\scriptscriptstyle \textcolor{gray}{(0.01)}}{0.71}$ \\
TS-SLOPE & $\underset{\scriptscriptstyle \textcolor{gray}{(4\times10^{-4})}}{\mathbf{7 \times 10^{-4}}}$ & $\underset{\scriptscriptstyle \textcolor{gray}{(0.01)}}{0.94}$ & $\underset{\scriptscriptstyle \textcolor{gray}{(0.01)}}{0.96}$ \\
TSO & $\underset{\scriptscriptstyle \textcolor{gray}{(2\times 10^{-3})}}{\mathbf{0.03}}$ & $\underset{\scriptscriptstyle \textcolor{gray}{(0.01)}}{0.95}$ & $\underset{\scriptscriptstyle \textcolor{gray}{(0.01)}}{0.95}$ \\
\bottomrule
\end{tabular}
\caption{Group metrics averaged across the six simulation cases considered for SGS, shown with standard errors. The best performing model for each metric within each model type is highlighted in \textbf{bold} (aside from FDR, for which any that have $\text{FDR} \leq 0.1$ are in bold).}
\label{tbl:full-summary-grp}
\end{table}
\begin{figure}[H]
    \centering
\includegraphics[width=.8\linewidth]{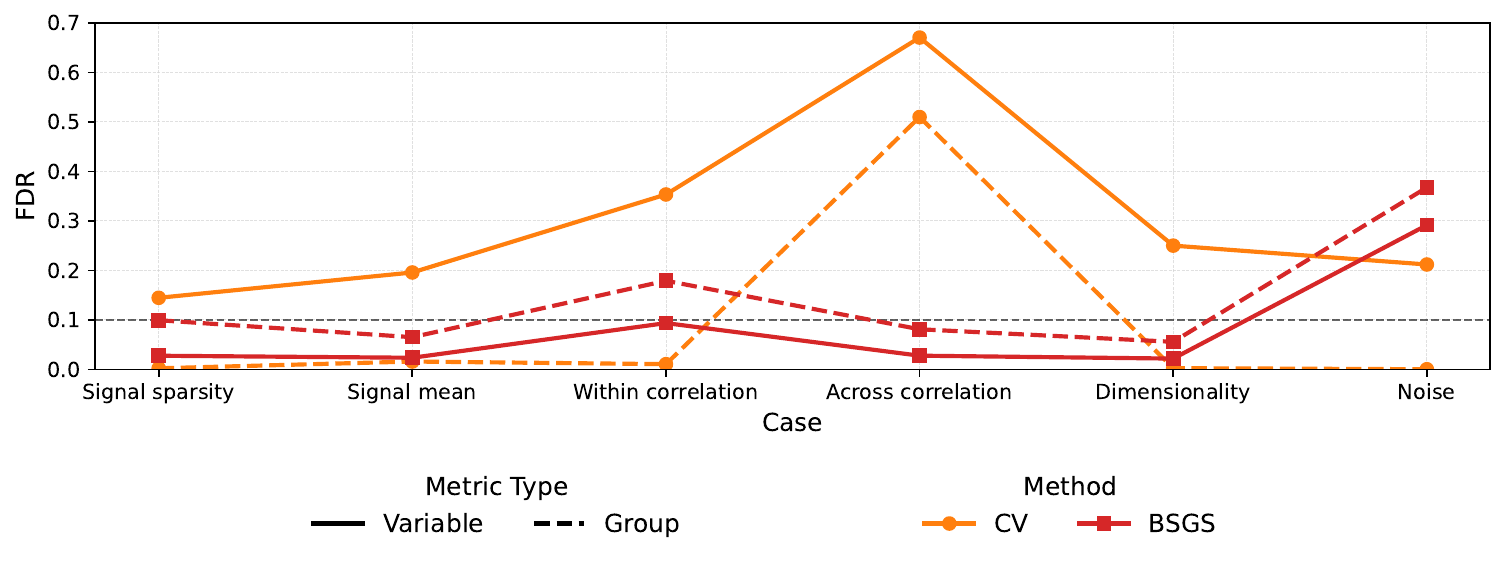}
    \caption{Variable and group FDR shown for CV and BSGS for all synthetic cases considered.}
    \label{fig:study-1-cv-bsgs-comparison}
\end{figure}
\subsection{Additional simulations}
\subsubsection{Equal groups}
\begin{figure}[H]
    \centering
\includegraphics[width=1\linewidth]{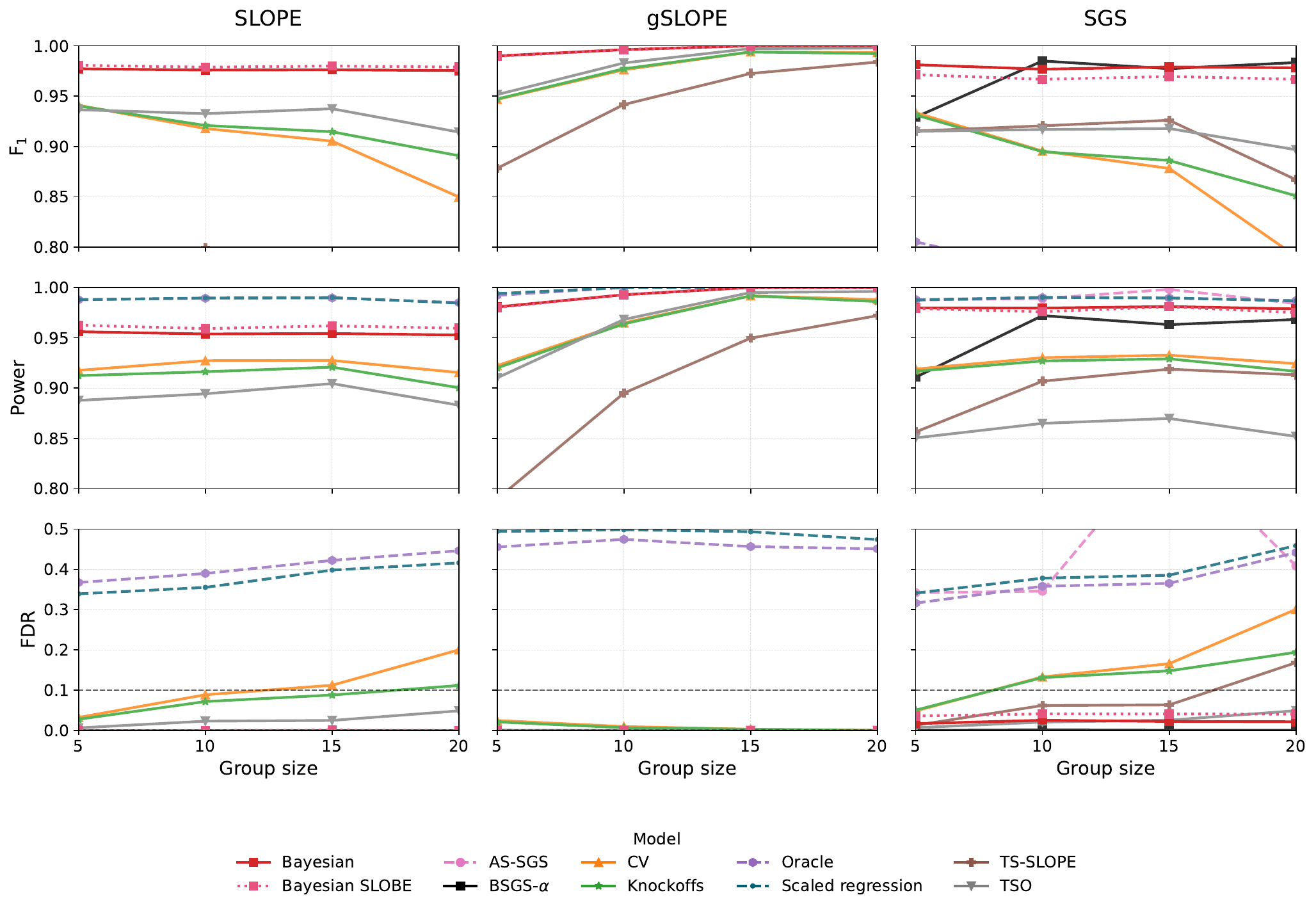}
    \caption{$\text{F}_1$ score (top row), power (middle row), and FDR (bottom plots) for all model selection approaches as a function of group size for equal groups.}
    \label{fig:study-2-case-1}
\end{figure}

\subsection{Real data experiments}
\subsubsection{Dataset information}\label{appendix:real_data_info}
The following real datasets are used in this manuscript:
\begin{itemize}
    \item BRCA1: Gene expression data for breast cancer tissue samples \citep{NCI}.
    \begin{itemize}
        \item Response: Gene expression measurements for the BRCA1 gene.
        \item Data matrix: Gene expression measurements for the other genes.
        \item Grouping structure: Variables were grouped via singular value decomposition.
    \end{itemize}
\item Cancer: Breast cancer patients treated with tamoxifen for 5 years \citep{Ma2004ATamoxifen}.
    \begin{itemize}
        \item Response: A synthetic (continuous) response was generated.
        \item Data matrix: Gene expression measurements.
        \item Grouping structure: Genes were assigned to pathways using the regulatory target gene sets, downloaded from \url{https://gsea-msigdb.org/gsea/msigdb/human/collections.jsp}. Only the C3 set was used.
    \end{itemize}
 \item Carbotax: Carbotax study of ovarian tumour growth \citep{carbotax}.
    \begin{itemize}
        \item Response: Relative tumour volume ($\log_2$ scale).
        \item Data matrix: Gene expression measurements. $10000$ factors were randomly sampled from a collection of $34964$.
        \item Grouping structure: Variables were grouped using K-means clustering \citep{1056489}.
    \end{itemize}
    \item Colitis: Blood cells data for classifying whether a patient has colitis \citep{Burczynski2006MolecularCells}.
    \begin{itemize}
        \item Response: A synthetic (continuous) response was generated.
        \item Data matrix: Gene expression measurements.
        \item Grouping structure: Genes were assigned to pathways using the regulatory target gene sets, downloaded from \url{https://gsea-msigdb.org/gsea/msigdb/human/collections.jsp}. Only the C3 set was used.
    \end{itemize}
    \item Rhee: HIV protease gene mutations and their impact on resistance to the drug Nelfinavir \citep{doi:10.1073/pnas.0607274103}. 
    \begin{itemize}
        \item Response: Results of the drug susceptibility assay. Higher values correspond to increased drug resistance.
        \item Data matrix: Binary indicators of mutations at specific positions in the HIV protease gene.
        \item Grouping structure: Variables were grouped via singular value decomposition.
    \end{itemize}
    \item Scheetz: Gene expression data in the mammalian eye \citep{scheetz2006RegulationDisease}. 
    \begin{itemize}
        \item Response: Gene expression measurements for the TRIM32 gene.
        \item Data matrix: Gene expression measurements for other genes.
        \item Grouping structure: Variables were grouped using K-means clustering \citep{1056489}.
    \end{itemize}
    \item Trust-experts: Survey response data as to how much participants trust experts ({\em e.g.}, doctors, nurses, scientists) to provide COVID-19 news and information \citep{Salomon2021TheVaccination}.
    \begin{itemize}
        \item Response: The trust level of each participant.
        \item Data matrix: Contingency table including factors about participants ({\em e.g.}, age, gender, ethnicity).
        \item Grouping structure: The factor levels are grouped into their original factors.
    \end{itemize}
\end{itemize}

\begin{table}[H]
  \centering
  \resizebox{\textwidth}{!}{%
\begin{tabular}{lrrrrrrr}
    \toprule
     \textbf{Dataset} & $p$ & $n$ & $m$ & \textbf{Group sizes} & \textbf{Grouping} & \textbf{Response} & \textbf{Downloaded on} \\
    \midrule
Cancer& $7057$ & $60$ & $1277$&$[1,292]$& Gene sets&Semi-synthetic continuous&08/2024 \\
Colitis&  $11999$ & $127$ & $1528$&$[1,497]$& Gene sets&Semi-synthetic continuous & 08/2024 \\
BRCA1& $10000$ & $536$ & $189$&$[1,3780]$& SVD&Continuous&05/2024 \\
Carbotax&$10000$ & $101$ & $200$&$[1,126]$& K-means&Continuous&08/2024 \\
Rhee&$361$ & $842$ & $207$ & $[1,9]$ & SVD&Continuous&10/2025\\
Scheetz&$18975$ & $120$ & $379$ & $[1,165]$ & K-means&Continuous&08/2024\\
Trust-experts&$101$ & $9759$ & $7$ & $[4,51]$ & Factors&Continuous&05/2024\\
    \bottomrule
  \end{tabular}
  }
  \caption[Dataset information for real data study]{Dataset information for the seven datasets used in the real data analysis.}
  \label{tbl:model_selection_real_data_info}
\end{table}
\subsubsection{Results}
\begin{table}[H]
\centering
\small
\resizebox{.9\textwidth}{!}{%
\begin{tabular}{lrrrrrr}
    \toprule
    \textbf{Dataset} & \multicolumn{2}{c}{\textbf{SLOPE}} & \multicolumn{2}{c}{\textbf{gSLOPE}} & \multicolumn{2}{c}{\textbf{SGS}} \\
    \cmidrule(lr){2-3} \cmidrule(lr){4-5} \cmidrule(lr){6-7}
    &ABSLOPE &SLOBE &BGSLOPE &GSLOBE &BSGS &SGSLOBE \\
    \midrule
  Cancer & $\underset{\scriptscriptstyle \textcolor{gray}{(0.111)}}{0.803}$ & $\underset{\scriptscriptstyle \textcolor{gray}{(0.095)}}{1.149}$ & $\underset{\scriptscriptstyle \textcolor{gray}{(0.031)}}{\vect{0.258}}$ & $\underset{\scriptscriptstyle \textcolor{gray}{(0.045)}}{0.403}$ & $\underset{\scriptscriptstyle \textcolor{gray}{(0.039)}}{0.291}$ & $\underset{\scriptscriptstyle \textcolor{gray}{(0.039)}}{0.291}$ \\
Colitis & $\underset{\scriptscriptstyle \textcolor{gray}{(0.052)}}{0.419}$ & $\underset{\scriptscriptstyle \textcolor{gray}{(0.112)}}{1.228}$ & $\underset{\scriptscriptstyle \textcolor{gray}{(0.105)}}{0.759}$ & $\underset{\scriptscriptstyle \textcolor{gray}{(0.735)}}{5.187}$ & $\underset{\scriptscriptstyle \textcolor{gray}{(0.052)}}{0.401}$ & $\underset{\scriptscriptstyle \textcolor{gray}{(0.052)}}{\vect{0.400}}$ \\
    \arrayrulecolor{gray!40}\midrule\arrayrulecolor{black}
  BRCA1 & $\underset{\scriptscriptstyle \textcolor{gray}{(0.043)}}{\vect{0.445}}$ & $\underset{\scriptscriptstyle \textcolor{gray}{(0.056)}}{1.028}$ & $\underset{\scriptscriptstyle \textcolor{gray}{(0.094)}}{0.884}$ & $\underset{\scriptscriptstyle \textcolor{gray}{(0.056)}}{1.028}$ & $\underset{\scriptscriptstyle \textcolor{gray}{(0.043)}}{0.467}$ & $\underset{\scriptscriptstyle \textcolor{gray}{(0.047)}}{\vect{0.445}}$ \\
  Carbotax & $\underset{\scriptscriptstyle \textcolor{gray}{(0.175)}}{0.838}$ & $\underset{\scriptscriptstyle \textcolor{gray}{(0.224)}}{0.999}$ & $\underset{\scriptscriptstyle \textcolor{gray}{(0.156)}}{\vect{0.742}}$ & $\underset{\scriptscriptstyle \textcolor{gray}{(0.224)}}{0.999}$ & $\underset{\scriptscriptstyle \textcolor{gray}{(0.195)}}{0.878}$ & $\underset{\scriptscriptstyle \textcolor{gray}{(0.177)}}{0.844}$ \\
  Rhee & $\underset{\scriptscriptstyle \textcolor{gray}{(0.007)}}{\vect{0.138}}$ & $\underset{\scriptscriptstyle \textcolor{gray}{(0.037)}}{0.970}$ & $\underset{\scriptscriptstyle \textcolor{gray}{(0.008)}}{0.140}$ & $\underset{\scriptscriptstyle \textcolor{gray}{(0.037)}}{0.970}$ & $\underset{\scriptscriptstyle \textcolor{gray}{(0.007)}}{\vect{0.138}}$ & $\underset{\scriptscriptstyle \textcolor{gray}{(0.007)}}{\vect{0.138}}$ \\
  Scheetz & $\underset{\scriptscriptstyle \textcolor{gray}{(0.102)}}{0.442}$ & $\underset{\scriptscriptstyle \textcolor{gray}{(0.388)}}{1.026}$ & $\underset{\scriptscriptstyle \textcolor{gray}{(0.112)}}{0.479}$ & $\underset{\scriptscriptstyle \textcolor{gray}{(0.388)}}{1.026}$ & $\underset{\scriptscriptstyle \textcolor{gray}{(0.236)}}{0.715}$ & $\underset{\scriptscriptstyle \textcolor{gray}{(0.079)}}{\vect{0.397}}$ \\
  Trust-experts & $\underset{\scriptscriptstyle \textcolor{gray}{(0.009)}}{\vect{0.345}}$ & $\underset{\scriptscriptstyle \textcolor{gray}{(0.009)}}{0.346}$ & $\underset{\scriptscriptstyle \textcolor{gray}{(0.009)}}{\vect{0.345}}$ & $\underset{\scriptscriptstyle \textcolor{gray}{(0.009)}}{\vect{0.345}}$ & $\underset{\scriptscriptstyle \textcolor{gray}{(0.009)}}{\vect{0.345}}$ & $\underset{\scriptscriptstyle \textcolor{gray}{(0.009)}}{\vect{0.345}}$ \\
    \bottomrule
  \end{tabular}
  }
  \caption[NMSE for SLOBE variants on real data]{NMSE for the Bayesian methods and the SLOBE variants for each model type (SLOPE, gSLOPE, SGS), with standard errors shown in \textcolor{gray}{grey}. The best method overall for a dataset is highlighted in \textbf{bold}.}
  \label{tbl:real-data-pred-scores-variant}
\end{table}
\begin{table}[H]
\small
\centering
\begin{tabular}{lrrrrrrrr}
    \toprule
    \textbf{Dataset} & \multicolumn{2}{c}{\textbf{BSGS real data}} & \multicolumn{2}{c}{\textbf{BSGS $\alpha=0.95$}} & \multicolumn{2}{c}{\textbf{BSGS default}} & \multicolumn{2}{c}{\textbf{BSGS-$\alpha$}} \\
    \cmidrule(lr){2-3} \cmidrule(lr){4-5} \cmidrule(lr){6-7} \cmidrule(lr){8-9}
    &Base &SLOBE &Base &SLOBE &Base &SLOBE &Base &SLOBE \\
    \midrule
  Cancer & $\underset{\scriptscriptstyle \textcolor{gray}{(0.039)}}{0.291}$ & $\underset{\scriptscriptstyle \textcolor{gray}{(0.039)}}{0.291}$ & $\underset{\scriptscriptstyle \textcolor{gray}{(0.035)}}{0.282}$ & $\underset{\scriptscriptstyle \textcolor{gray}{(0.036)}}{0.281}$ & $\underset{\scriptscriptstyle \textcolor{gray}{(0.045)}}{0.403}$ & $\underset{\scriptscriptstyle \textcolor{gray}{(0.039)}}{0.290}$ & $\underset{\scriptscriptstyle \textcolor{gray}{(0.036)}}{\vect{0.272}}$ & $\underset{\scriptscriptstyle \textcolor{gray}{(0.036)}}{0.275}$ \\
Colitis & $\underset{\scriptscriptstyle \textcolor{gray}{(0.052)}}{0.401}$ & $\underset{\scriptscriptstyle \textcolor{gray}{(0.052)}}{\vect{0.400}}$ & $\underset{\scriptscriptstyle \textcolor{gray}{(0.735)}}{5.187}$ & $\underset{\scriptscriptstyle \textcolor{gray}{(0.735)}}{5.187}$ & $\underset{\scriptscriptstyle \textcolor{gray}{(0.103)}}{0.783}$ & $\underset{\scriptscriptstyle \textcolor{gray}{(0.103)}}{0.773}$ & $\underset{\scriptscriptstyle \textcolor{gray}{(0.098)}}{0.740}$ & $\underset{\scriptscriptstyle \textcolor{gray}{(0.098)}}{0.737}$ \\
    \arrayrulecolor{gray!40}\midrule\arrayrulecolor{black}
BRCA1 & $\underset{\scriptscriptstyle \textcolor{gray}{(0.043)}}{0.467}$ & $\underset{\scriptscriptstyle \textcolor{gray}{(0.047)}}{\vect{0.445}}$ & $\underset{\scriptscriptstyle \textcolor{gray}{(0.075)}}{0.609}$ & $\underset{\scriptscriptstyle \textcolor{gray}{(0.051)}}{0.459}$ & $\underset{\scriptscriptstyle \textcolor{gray}{(0.056)}}{1.028}$ & $\underset{\scriptscriptstyle \textcolor{gray}{(0.049)}}{0.873}$ & $\underset{\scriptscriptstyle \textcolor{gray}{(0.056)}}{0.510}$ & $\underset{\scriptscriptstyle \textcolor{gray}{(0.056)}}{0.500}$ \\
  Carbotax & $\underset{\scriptscriptstyle \textcolor{gray}{(0.195)}}{0.878}$ & $\underset{\scriptscriptstyle \textcolor{gray}{(0.177)}}{0.844}$ & $\underset{\scriptscriptstyle \textcolor{gray}{(0.197)}}{0.875}$ & $\underset{\scriptscriptstyle \textcolor{gray}{(0.162)}}{0.828}$ & $\underset{\scriptscriptstyle \textcolor{gray}{(0.224)}}{0.999}$ & $\underset{\scriptscriptstyle \textcolor{gray}{(0.176)}}{0.823}$ & $\underset{\scriptscriptstyle \textcolor{gray}{(0.228)}}{0.986}$ & $\underset{\scriptscriptstyle \textcolor{gray}{(0.164)}}{\vect{0.789}}$ \\
  Rhee & $\underset{\scriptscriptstyle \textcolor{gray}{(0.007)}}{\vect{0.138}}$ & $\underset{\scriptscriptstyle \textcolor{gray}{(0.007)}}{\vect{0.138}}$ & $\underset{\scriptscriptstyle \textcolor{gray}{(0.007)}}{0.139}$ & $\underset{\scriptscriptstyle \textcolor{gray}{(0.007)}}{\vect{0.138}}$ & $\underset{\scriptscriptstyle \textcolor{gray}{(0.041)}}{0.832}$ & $\underset{\scriptscriptstyle \textcolor{gray}{(0.037)}}{0.970}$ & $\underset{\scriptscriptstyle \textcolor{gray}{(0.008)}}{0.139}$ & $\underset{\scriptscriptstyle \textcolor{gray}{(0.007)}}{\vect{0.138}}$ \\
  Scheetz & $\underset{\scriptscriptstyle \textcolor{gray}{(0.236)}}{0.715}$ & $\underset{\scriptscriptstyle \textcolor{gray}{(0.079)}}{\vect{0.397}}$ & $\underset{\scriptscriptstyle \textcolor{gray}{(0.263)}}{0.808}$ & $\underset{\scriptscriptstyle \textcolor{gray}{(0.088)}}{0.416}$ & $\underset{\scriptscriptstyle \textcolor{gray}{(0.388)}}{1.026}$ & $\underset{\scriptscriptstyle \textcolor{gray}{(0.388)}}{1.026}$ & $\underset{\scriptscriptstyle \textcolor{gray}{(0.388)}}{1.026}$ & $\underset{\scriptscriptstyle \textcolor{gray}{(0.132)}}{0.498}$ \\
  Trust-experts & $\underset{\scriptscriptstyle \textcolor{gray}{(0.009)}}{\vect{0.345}}$ & $\underset{\scriptscriptstyle \textcolor{gray}{(0.009)}}{\vect{0.345}}$ & $\underset{\scriptscriptstyle \textcolor{gray}{(0.009)}}{\vect{0.345}}$ & $\underset{\scriptscriptstyle \textcolor{gray}{(0.009)}}{\vect{0.345}}$ & $\underset{\scriptscriptstyle \textcolor{gray}{(0.009)}}{0.347}$ & $\underset{\scriptscriptstyle \textcolor{gray}{(0.009)}}{0.347}$ & $\underset{\scriptscriptstyle \textcolor{gray}{(0.009)}}{\vect{0.345}}$ & $\underset{\scriptscriptstyle \textcolor{gray}{(0.009)}}{\vect{0.345}}$ \\
    \bottomrule
  \end{tabular}
  \caption[NMSE for BSGS variants on the real data]{NMSE for the BSGS models and the SLOBE variants, with standard errors shown in \textcolor{gray}{grey}. BSGS real data (which is the primary BSGS model used in Section \ref{section:real_data} and uses $\alpha = 0.99$), BSGS $\alpha = 0.95$, and BSGS-$\alpha$ all use the $\theta_g, \theta_v \sim \text{Beta}(5,1)$ priors. BSGS default uses the default prior scheme from the synthetic study (Scheme 1). The best method overall for a dataset is highlighted in \textbf{bold}.}
  \label{tbl:real-data-pred-scores-bsgs}
\end{table}
\end{document}